\documentclass[
reprint,
superscriptaddress,
nofootinbib,
amsmath,amssymb,
aps,
prx,
]{revtex4-2}

\usepackage{graphicx}
\usepackage{dcolumn}
\usepackage{bm}
\usepackage[colorlinks,linkcolor=blue,urlcolor=blue, citecolor=blue]{hyperref}

\usepackage{dsfont}
\usepackage{physics}
\usepackage{enumitem}
\usepackage{tcolorbox}
\usepackage{mathtools}
\usepackage{tabularx}
\newcolumntype{C}{>{\centering\arraybackslash}X}

\usepackage{amsmath, amssymb, amsfonts, amsthm}
\usepackage{thmtools, thm-restate}

\newtheorem{theorem}{Theorem}
\newtheorem*{theorem*}{Theorem}

\newtheorem{lemma}[theorem]{Lemma}

\allowdisplaybreaks

\begin{document}

\title{Energy-Consumption Advantage of Quantum Computation}

\author{Florian Meier}
\email[]{florianmeier256@gmail.com}
\affiliation{Vienna Center for Quantum Science and Technology, Atominstitut, Technische Universität Wien, 1020 Vienna, Austria}
\author{Hayata Yamasaki}
\email[]{hayata.yamasaki@gmail.com}
\affiliation{Department of Physics, Graduate School of Science, The University of Tokyo, 7--3--1 Hongo, Bunkyo-ku, Tokyo 113--0033, Japan}

\begin{abstract}
Energy consumption in solving computational problems has been gaining growing attention as one of the key performance measures for computers. Quantum computation is known to offer advantages over classical computation in terms of various computational resources; however, proving its energy-consumption advantage has been challenging due to the lack of a theoretical foundation linking the physical concept of energy with the computer-scientific notion of complexity for quantum computation. To bridge this gap, we introduce a general framework for studying the energy consumption of quantum and classical computation, based on a computational model conventionally used for studying query complexity in computational complexity theory. Within this framework, we derive an upper bound for the achievable energy consumption of quantum computation, accounting for imperfections in implementation appearing in practice. As part of this analysis, we construct a protocol for Landauer erasure with finite precision in a finite number of steps, which constitutes a contribution of independent interest. Additionally, we develop techniques for proving a nonzero lower bound of energy consumption of classical computation, based on the energy-conservation law and Landauer's principle. Using these general bounds, we rigorously prove that quantum computation achieves an exponential energy-consumption advantage over classical computation for solving a paradigmatic computational problem---Simon's problem. Furthermore, we propose explicit criteria for experimentally demonstrating this energy-consumption advantage of quantum computation, analogous to the experimental demonstrations of quantum computational supremacy. These results establish a foundational framework and techniques to explore the energy consumption of computation, opening an alternative way to study the advantages of quantum computation.
\end{abstract}

\maketitle

\section{\label{sec:introduction}Introduction}
With growing interest in the sustainability of our society, energy consumption is nowadays considered an important part of performance measures for benchmarking computers.
It is expected that quantum computation will be no exception; its energy efficiency will ultimately be one of the deciding factors as to whether quantum computers will be used on a large scale~\cite{Preskill2018quantumcomputingin,Auffeves2022,Jaschke2023}.
Originally, quantum computers emerged as a promising platform to solve certain computational problems that would otherwise be unfeasible to solve on classical computers.
The advantage of quantum computation is generally examined in terms of computational complexity, which quantifies computational resources required for solving the problems, such as time complexity, communication complexity, and query complexity~\cite{Watrous2009}.
Energy is, however, a different computational resource from the above ones.
A priori, an advantage in some computational resource does not necessarily imply that in another; for example, quantum computation is believed to achieve an exponential advantage in time complexity over classical computation but does not provide such an advantage in the required amount of memory space~\cite{10.1007/s00037-003-0177-8}.
Whether quantum computation can offer a significant energy-consumption advantage over classical computation is a fundamental question that has been asked since the earliest days of quantum computation~\cite{Benioff1982,Feynman1986}, but has not been explored rigorously as of now due to the lack of theoretical foundation to relate the known advantages of quantum computation and the advantage in its energy consumption.
For example, some existing works~\cite{Jaschke2023,FellousAsiani2022} numerically evaluate achievable upper bounds of energy consumption of quantum computation, but it is not theoretically proven in these works how the energy consumption of quantum computation scales on large scales in general; even more problematically, it has been challenging to provide a fundamental lower bound of the energy consumption of any classical computation, which the existing works do not rigorously analyze.
It has been, therefore, a fundamental open problem to formulate the theoretical framework to compare an upper bound of energy consumption of quantum computation and a lower bound of energy consumption of classical computation fairly in a technology-independent way, so as to prove the energy-consumption advantage of quantum computation over classical computation in a commonly studied setting of computation~\cite{Auffeves2022}.

\begin{figure}
    \centering
    \includegraphics[width=3.4in]{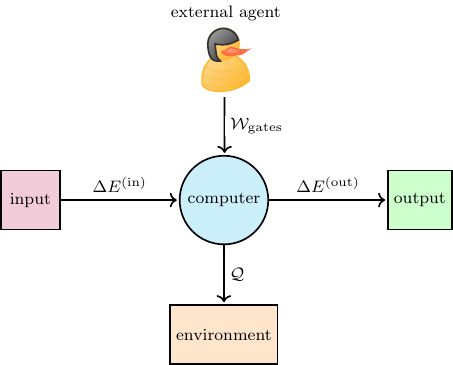}
    \caption{Our thermodynamic model of quantum and classical computation. To solve a computational problem, an external agent manipulates the computer at a work cost $\mathcal{W}_\mathrm{gates}$. In the computation, the input for the problem is provided from an input system to the computer, with energy $\Delta E^{(\mathrm{in})}$ flowing into the computer. At the end of the computation, the computer outputs the result of the computation into the output system, with energy $\Delta E^{(\mathrm{out})}$ flowing out; then, the internal state of the computer must be reinitialized into a fixed state $\approx\ket{00\cdots 0}$, so that the computer can be used for solving the next computational task. Throughout manipulating and initializing the computer, energy in the computer may flow into the environment as a heat $\mathcal{Q}$. As explained in the main text, the energy consumption $\mathcal{W}_\mathrm{gates}+\Delta E^{(\mathrm{in})}-\Delta E^{(\mathrm{out})}$ of the computation is classified into three contributions: the energetic cost, the control cost, and the initialization cost. Our analysis provides a comprehensive account of each contribution to bound the overall energy consumption of quantum and classical computation.}
    \label{fig:thermo_diagram}
\end{figure}

\paragraph*{Main achievements of our work.}
In this work, we carry out a rigorous, in-depth study of the energy consumption of quantum and classical computation.
Through our work, we introduce a general framework to analyze energy consumption in both quantum and classical computation, which incorporates computational and thermodynamic aspects, as illustrated in Fig.~\ref{fig:thermo_diagram}.

In Fig.~\ref{fig:thermo_diagram}, the computation is modeled as a thermodynamic cycle.
This model is composed of five parts: the computer, agent, environment, input, and output.
The computer, either quantum or classical, is the physical platform with (qu)bits, where an algorithm is carried out.
To execute the algorithm, an external agent invests some work $\mathcal W_\mathrm{gates}$ into the computer to perform a sequence of gates on the (qu)bits in the computer.
Part of this work is an \textit{energetic cost} used for changing the energy of the internal state of the computer, while the rest is the \textit{control cost} dissipated into the environment due to energy loss arising from, e.g., friction and electric resistance. 
Moreover, as in a conventional setting of studying query complexity to separate the complexities in quantum and classical computation~\cite{Watrous2009,Nielsen2010,BUHRMAN200221}, we structure the algorithm to process the input coming from an oracle --- a black-box circuit outside the computer that is unknown to the agent, so that the agent should perform the computation to learn a property of this black box through the queries to the oracle.
For this input, we use an additional input system representing the oracle and label the energy of the states transferred from the input system to the computer as $\Delta E^{(\text{in})}$.
The result of the entire computation is to be output and stored in another separate system, the output system.
In our case, we solve a decision problem, so the output system should be a single bit.
The energy of the state output from the computer incurs an energetic cost $\Delta E^{(\mathrm{out})}$.
At the end of the computation, the final state of the computer is reinitialized into its initial state, so that, afterward, the agent should be able to carry out the next computation by the same computer.
Aside from the control cost, irreversible operations such as the erasure of the final state of the computer for the initialization also dissipate heat into the environment, which incurs an \textit{initialization cost}.
All-in-all, we summarize the heat dissipated into the environment as $\mathcal Q$.
The energy consumption of the computation is all the energy used for the above process except for that dissipating into the environment, given by
\begin{equation}
    \mathcal{W}_\mathrm{gates}+\Delta E^{(\mathrm{in})}-\Delta E^{(\mathrm{out})}.
\end{equation}

Our methodology pioneers general techniques to clarify the upper and lower bounds of the energy consumption for the computation in the query-complexity setting.
The upper bound of achievable energy consumption of quantum algorithms is obtained by meticulously considering the work cost of implementing gates to finite precision, erasing the computer's memory state for the initialization based on Landauer’s principle~\cite{Landauer1961}, and reducing the effect of noise by quantum error correction.
Under the scalability assumption that each elementary operation has, at most, a polynomial amount of energy consumption, we derive an upper bound of energy consumption of quantum computation by bounding the required number of elementary operations for solving computational problems, including quantum error correction.
Also notably, we develop a novel technique for deriving a fundamental lower bound of the energy required for classical algorithms to solve computational problems.
In contrast with the upper bound, it is, in general, unknown whether one can find a nonzero lower bound of the energy consumption of each elementary operation; after all, in the idealized limit of zero friction or zero electrical resistance, each elementary operation may consume almost zero energy.
To overcome this challenge, we will develop the technique to argue that a lower bound of the query complexity in solving the computational problems can be turned into a lower bound of the Laudauer cost, i.e., the energy consumption to erase the information obtained by the queries.

Based on these bounds, we rigorously prove that a quantum algorithm can attain an exponential energy-consumption advantage over any classical algorithm for an exemplary computational problem---Simon's problem~\cite{Simon1994, Simon1997}.
Furthermore, for this problem, we clarify how to implement the quantum algorithm using a cryptographic primitive to feasibly implement the input system in Fig.~\ref{fig:thermo_diagram} and provide explicit criteria for demonstrating the energy-consumption advantage of quantum computation in experiments.
In this way, we expand the traditional focus from conventional complexity-theoretic notions in theoretical computer science to energy consumption---a critical aspect of computational resources that was often overlooked in the existing theoretical studies.
These results lay a solid foundation for the exploration of the advantage of quantum computation over classical computation in terms of energy consumption.

\paragraph*{Key physical insights.}
In our analysis of energy consumption, we develop a comprehensive account of the above three types of costs that contribute to energy consumption: the energetic cost, the control cost, and the initialization cost.

The energetic cost to change the energy of the state of the computer would be hard to evaluate at each step of the algorithm, but to avoid this hardness, our analysis evaluates the sum of the energetic cost of the overall computation rather than each step.
The hardness arises from the fact that an energetic cost per gate can be both positive and negative depending on the state of the computer; after all, when performing a gate in the algorithm, the same gate may change a low-energy state to a high-energy state and also a high-energy state to a low-energy state.
Thus, if one wants to know the required energetic cost at each step of the algorithm, the internal state of the computer at the step should also be kept track of, which would be as hard as simulating each step of the algorithm.
By contrast, in our analysis, instead of focusing on each step, we regard the overall computation as a thermodynamic cycle as described in Fig.~\ref{fig:thermo_diagram}.
As a result, the net energetic cost of all the steps sums up to zero due to the energy-conservation law, in both quantum and classical computation,
leading to the total energy consumption of the computation coming solely from heat dissipation due to inefficient control and irreversible erasure.

The control cost, arising from the energy loss in implementing each gate, also needs careful analysis.
On one hand, the quantum computer is often said to consume a large amount of energy due to the complicated experimental setup under the current technology and, in this regard, might be considered to be more energy-consuming than the classical computer.
But this argument overlooks the fact that the energy loss in any scalable implementation of each quantum gate is reasonably upper bounded by, at most, a polynomial factor per physical gate (and usually upper bounded by a constant factor) as the size of problems to be solved by the computer increases; that is, the apparent difference in the energy efficiencies between quantum and classical computers is that of the cost per implementing each physical gate.
With this observation, we derive an achievable upper bound of the control cost of quantum computation from the number of gates used for performing the overall algorithm.
On the other hand, as opposed to the above saying, it is also often said that the quantum computer could save energy compared to the classical computer due to the large quantum speedup in solving some problems.
This logic also overlooks another fact that the control cost per implementing a classical gate can be made extremely small compared to that of a quantum gate.
Due to this fact, the quantum advantage in time complexity does not straightforwardly lead to that of energy consumption over classical computation, especially in the case where the control cost of each classical gate is negligibly small; for example, imagine a limit of idealization where the friction and the electric resistance are negligible, and then, no energy loss occurs regardless of the number of gates to be applied. 
For example, it is indeed proposed to use superconducting materials for classical computation, which may allow for nearly dissipation-free classical computation in the future (see, e.g., Ref.~\cite{Mukhanov2011}).
Thus, unlike the upper bound of the control cost, a lower bound of the control cost can be, in principle, arbitrarily close to zero even if the runtime of the computation is long.

Initialization cost, i.e., the required energy for resetting the state of the computer to a fixed initial state, is the key to our proof of the energy-consumption advantage of quantum computation over classical computation;  this cost can provide a fundamentally \textit{nonzero} lower bound of the energy consumption of classical computation, unlike the above energetic and control costs.
As described in Fig.~\ref{fig:thermo_diagram}, after solving a computational problem, the (quantum or classical) computer needs to reset all the (qu)bits to conduct the next task, by erasing the final state of the (qu)bits into a fixed initial state, i.e., $\ket{00\cdots 0}$.
However, in the thermodynamic analysis of the energy consumption, we need to be careful about the cost of this initialization due to Nernst's unattainability principle~\cite{Nernst1906}; it is known that the ideal initialization of a qubit exactly into a pure state would require divergent resource costs, e.g., infinite energy consumption or infinitely long runtime~\cite{Taranto2021}.
To avoid an unbounded energy consumption in the initialization, we newly construct a finite-step Landauer-erasure protocol for erasing the states of the qubits to a pure state approximately within a desired finite precision, which allows us to achieve a finite energy consumption for all the erasure steps required for the initialization after conducting the computation in Fig.~\ref{fig:thermo_diagram}.

Apart from this explicit protocol construction for showing an achievable upper bound of energy consumption, we discover a way to prove a nonzero lower bound of the initialization cost by applying two fundamental concepts in physics: the energy-conservation law in thermodynamics, and Landauer's principle~\cite{Landauer1961} that relates the irreversible process to energy consumption.
Landauer's principle shows that we need more energy consumption to initialize more information stored as states of (qu)bits, which is measured by the entropy of the states.
In particular, to solve Simon's problem~\cite{Simon1994,Simon1997}, we see that any classical algorithm needs to make an exponentially large number of queries to obtain an exponentially long bit string obtained from the input system.
In this case, our analysis shows that the initialization cost required for erasing this input information, in terms of entropy, also becomes exponentially large, which is one of our key techniques for proving the exponential energy-consumption advantage of quantum computation.

Consequently, with our formulation of the computation as the thermodynamic cycle in Fig.~\ref{fig:thermo_diagram}, we provide an unprecedented detailed analysis of all three types of costs, establishing the novel techniques for showing upper and lower bounds of the energy consumption of quantum and classical computation.
Our analysis shows that the upper bound on the energy consumption of quantum computation is governed by the gate and query complexities and the heat dissipation achievable through finite-step Landauer erasure.
Therefore, to ensure an upper bound of energy consumption within this framework, it suffices to feasibly bound each of these factors.
In contrast, the provably nonzero lower bound on the energy consumption of classical computation arises from the minimum cost of Landauer erasure, quantified by entropy.
Although we analyze this lower bound using techniques originally developed for studying query complexity, the lower bound is ultimately determined by the Landauer cost---the entropic quantity that depends both on the query complexity and the probability distribution of the oracle.
These techniques make it possible to rigorously analyze the energy-consumption advantage of quantum computation with a solid theoretical foundation.

\paragraph*{Impact.}

In today's world, where sustainability is a key concern, energy efficiency has become an essential performance metric for modern computers, along with computational speed and memory consumption. Our theoretical framework provides a new pathway for examining the relevance of quantum computation in terms of energy efficiency.
Our study provides a fundamental framework and techniques for exploring a novel quantum advantage in computation in terms of energy consumption.
The framework is designed in such a way that a quantum advantage in query complexity can be employed to prove the energy-consumption advantage of quantum computation over classical computation.
These results also clarify the physical meaning of the quantum advantage in the query-complexity setting in terms of energy consumption.

The query-complexity setting is a widely studied setting of computation by itself, but also, from a broader perspective, quantum computers may have promising applications in learning properties of physical dynamics described by an unknown quantum-mechanical map~\cite{Cerezo2022,PhysRevLett.126.190505,aharonov2022quantum,chen2022exponential,huang2022quantum, huang2023learning}.
In these learning problems, the physical dynamics can be considered a black-box oracle whose properties are to be learned, as in the query-complexity setting.
In other words, the query-complexity setting can also be regarded as a variant of these learning problems of unknown physical dynamics represented by the oracles.
Furthermore, we also clarify how to implement and demonstrate the query-complexity setting in an experimental setup.
The framework and techniques developed in this work serve as a theoretical foundation to realize the energy-consumption advantage of quantum computation in query-complexity settings with these physically well-motivated applications.

\paragraph*{Organization of this article.}
The rest of this article is organized as follows.
\begin{itemize}
    \item\label{item:contribution_1} In Sec.~\ref{sec:setting}, we formulate a framework within which the energy consumption of quantum and classical computation can be rigorously studied.
    We will describe its core idea based on Fig.~\ref{fig:thermo_diagram}, and further detail of the framework will be illustrated in Fig.~\ref{fig:generic_algo}.
    \item\label{item:contribution_2}
    In Sec.~\ref{sec:results}, we show general upper and lower bounds of quantum and classical computation within the framework of Sec.~\ref{sec:setting}. In Sec.~\ref{sec:quantum_upperbound}, we derive a general upper bound on the energy consumption that is achievable by the quantum computation in our framework. For this analysis, extending upon the existing asymptotic results~\cite{Reeb2014,Proesmans2020,Taranto2021,VanVu2022,Rolandi2022} on Landauer erasure~\cite{Landauer1961}, we show the achievability of finite-fidelity and finite-step Landauer erasure (Theorem~\ref{thm:ThmErasureBound}). Apart from the finite Landauer erasure, we derive the upper bound of achievable energy consumption using complexity-theoretic considerations and also taking into account overheads from quantum error correction, so as to cover all the energy consumption in the framework (Theorem~\ref{thm:ThmQuantum}).
    Moreover, in Sec.~\ref{sec:classical_lowerbound}, we develop a technique for obtaining an implementation-independent lower bound on the energy consumption of the classical computation in our framework. The lower bound is derived using energy conservation and the Landauer-erasure bound (Theorem~\ref{thm:ThmClassical}).
    \item\label{item:contribution_4} 
    In Sec.~\ref{sec:example}, we show an explicit example where the energy consumption of quantum and classical computation for solving a computational problem is exponentially separated. To prove this exponential energy-consumption advantage of quantum computation rigorously, we apply the above upper and lower bounds to Simon's problem~\cite{Simon1994, Simon1997} (Corollaries~\ref{cor:CorQuantum} and~\ref{cor:CorClassical}).
    \item\label{item:contribution_5} In Sec.~\ref{sec:experiment}, we propose a setting to demonstrate the energy-consumption advantage of quantum computation in quantum experiments.
    \item\label{item:contribution_6}
    In Sec.~\ref{sec:discussion}, we provide a discussion and outlook based on our findings.
\end{itemize}

\section{\label{sec:setting}Formulation of computational framework and energy consumption}

In this section, we present the setting of our analysis.
We start with introducing the thermodynamic model of the computation in Sec.~\ref{sec:termodynamic_model_of_computation}.
In Sec.~\ref{sec:computational_setting}, we formulate and elaborate on the framework used for analyzing the energy consumption of the computation (Fig.~\ref{fig:generic_algo}).
In Sec.~\ref{sec:energetic_costs}, we introduce all the relevant costs that contribute to the energy consumption of the computation in this framework.

\subsection{\label{sec:termodynamic_model_of_computation}Thermodynamic model of computation}

Computation is a physical process in a closed physical system that receives a given input and returns the corresponding output of a mathematical function. 
Various models can be considered as the physical processes for the computation.
Conventionally, classical computation uses the Turing machine, or equivalently a classical logic circuit, which works based on classical mechanics; quantum computation can be modeled by a quantum circuit based on the law of quantum mechanics~\cite{Watrous2009,Nielsen2010}.
Under restrictions on some computational resources, the differences between classical and quantum computation may appear.
For example, in terms of running time, quantum computation is considered to be advantageous over classical computation~\cite{doi:10.1126/science.aar3106}.
Another type of quantum advantage can be shown by query complexity~\cite{BUHRMAN200221}, where we have a black-box function, or an \textit{oracle}, that can be called multiple times during executing the computation.
In this query-complexity setting,  the required number of queries to the oracle for achieving the computational task is counted as a computational resource of interest.
This setting may concern the fundamental structure of computation rather than space and time resources.
The advantage of quantum computation over classical computation can also be seen in the required number of queries for solving a computational task.

Progressing beyond the studies of these various computational resources, we here formulate and analyze the fundamental requirements of energy consumption in classical and quantum computation, by introducing the model illustrated in Fig.~\ref{fig:thermo_diagram} for our thermodynamic analysis of computation.
To clarify the motivation for introducing this model, we start this section by going through the main difficulties in proving a potential advantage in the energy consumption of quantum computation over classical computation.
Then, in the next step, we will provide additional details to Fig.~\ref{fig:thermo_diagram} and formally define what we refer to as the energy consumption of a computation.

\paragraph*{Challenges in studying energy-consumption advantage of quantum computation.}

To analyze the energy-consumption advantage of quantum computation over classical computation, there are two inequalities that have to be shown: first, an achievable upper bound for the energy consumption of performing a quantum algorithm, and second, a lower bound on the energy consumption for \textit{all} classical algorithms for solving the same problem.

As for the former, previous works mostly investigate energy consumption in implementing a single quantum gate~\cite{Ozawa2002,GeaBanacloche2006,Ikonen2017,Stevens2022,Chiribella2021,Yang2022,PhysRevLett.121.110403,PhysRevResearch.2.043374}, but the analysis of the quantum advantage requires an upper bound of energy consumption of the overall quantum computation composed of many quantum operations.
Such an analysis has been challenging because, to account for the energy consumption of quantum computation, we need to take into account all the operations included in the computation, e.g., not only the gates but also the cost of initializing qubits and performing quantum error correction.
A challenge here arises since, unlike quantum gates, some other quantum operations, such as measurements and initialization, may consume an infinitely large amount of energy as we increase the accuracy in their implementation~\cite{Guryanova2020,Taranto2021}.
Thus, we need to formulate the framework of quantum computation properly to avoid the operations requiring infinite energy consumption and to establish finite achievability results of energy consumption for all the operations used for the quantum computation within the framework.

The latter is even more challenging since it is, in general, hard to derive a lower bound of energy consumption of overall computation from the lower bound of each individual operation; after all, we may be able to perform multiple operations in the computation collectively to save energy consumption.
In the first place, the energetic cost of performing reversible operations, be it on a quantum or classical computer, can be either positive or negative.
To see this fact, consider the following illustrative example: what is the lower bound of performing a bit-flip operation?
The answer may depend on the implementation and the initial state.
Let us choose, for example, a two-level system with states $0$ and $1$ and energies $E_0<E_1$.
If the bit initially starts in the state $0$, then we need to invest a positive amount of energy $E_1-E_0>0$ to obtain $1$ by the bit-flip gate.
By contrast, at best, one can also gain energy $E_1-E_0$ from the bit-flip gate, if the bit initially starts in the state $1$ and is flipped into $0$; in other words, the energetic cost of performing a reversible operation can be negative in general.
Thus, for the initial input state $1$, the operation could even extract work because energy-efficient gate implementations are possible in principle, especially for macroscopic physical systems.\footnote{Two concrete examples would be the Japanese \textit{Soroban} and the \textit{Abacus}.
Suppose that these devices are standing upright in the earth's gravitational field, and the external agent keeps the locations of the marbles of these devices fixed by some means. The locations of the marbles are used for representing states of bits, and bit-flip gates can be applied by physically moving the marbles. In this case, the information-bearing degrees of freedom carry potential energy. Moving up a marble requires an investment of work. On the other hand, moving down the marble yields a gain of energy, i.e., the extraction of the work.}
Apart from this energetic cost, the implementation of the bit-flip gate may also require an additional control cost arising from heat dissipation, which is caused by, e.g., friction and electrical resistance.
The control cost may be positive in reality, but in the limit of energy-efficient implementation,  it is hard to rule out the possibility that the control cost may be negligibly small; in particular, it is unknown whether the infimum of the control cost over any possible implementation of computation can still be lower bounded by a strictly positive constant gapped away from zero.

As far as the energetic and control costs are concerned, it is challenging to rule out the possibility that an energy-efficient implementation would achieve zero work cost (or arbitrarily small work cost) per gate even in the worst case over all possible input states.
For example, in the limit of realizing $0$ and $1$ with degenerate energy levels $E_0=E_1$, the bit-flip gate does not come at a fundamentally positive energetic cost.
Another example may exist in the limit of ideal implementation where friction, electrical resistance, and other sources of energy loss are negligible, and if this is the case, the control cost also becomes negligibly small.
As long as one considers these costs of an individual gate, it is hard to argue that each gate requires a positive work cost.
Thus, the analysis requires a novel technique for deriving a nonzero lower bound of energy consumption of the overall computation, which needs to be applied independently of the detail of the implementation; in our case, we derive such a nonzero lower bound based on taking into account the initialization cost, energy consumption for erasing the state of the computer at the end of the computation into the fixed initial state $0$.

\paragraph*{Model.}
In particular, to establish the framework for studying energy consumption, we view quantum and classical computation as thermodynamic processes on their respective physical implementation platforms, as shown in Fig.~\ref{fig:thermo_diagram}.
We here describe our model in detail.
The agent uses a computer to carry out the main part of the computation to solve a decision problem.
Performing the operations that make up an algorithm comes at a work cost for this external agent, which we summarize as $\mathcal{W}_\text{gates}$.
This cost includes, on the one hand, energetic costs arising from the change of energy of the internal states of the computer and, on the other hand, control costs caused by energetic losses in implementing these operations due to heat dissipation.

As in a conventional setting of studying query complexity to separate the complexity of quantum and classical computation~\cite{Watrous2009,Nielsen2010,BUHRMAN200221}, we use an oracle as an input model, which provides the input to the computer via oracle queries.
In our model, the oracle is a black box outside the computer whose internals are unknown to the agent performing the computation, as we will define more precisely in Sec.~\ref{sec:computational_setting}.
There is some transfer of energy between the oracle and the computer, namely $\Delta E^{(\text{in})}$, due to the states the oracle generates and inputs into the computer.

On the other hand, working on decision problems means that the output is a two-level system where the decision is stored as either $0$ or $1$.
Generating this single-bit output comes at an energy exchange of $\Delta E^{(\mathrm{out})}=O(1)$ and is practically negligible compared to the growing sizes of the computer, the input, and the thermal environment surrounding around the computer, which may usually scale polynomially (or can scale even exponentially in our framework) as the problem size increases.

Throughout the computation, the energy corresponding to the control cost may flow into the environment, which in parts contributes to the heat $\mathcal Q$ dissipated to the environment.
The other contribution to $\mathcal Q$ arises from reinitializing the internal state of the computer at the end of the computation into the fixed initial state as it was before performing the computation.
The resetting of the memory after the computation is necessary because we demand that the computation is cyclic on the computer so that this computer can be used for solving another task after conducting the current computation.
In contrast to the control cost that may approach zero in the limit of energy-efficient implementation, the energy consumption in reinitializing the computer may be nonzero due to Landauer's principle~\cite{Landauer1961} if part of the computation is irreversible, as is the case for the oracle queries in classical computation.

With this model, we define the \textit{energy consumption} $\mathcal W$ of the computer as the sum of all the contributions of the external agent, the input, and the output in Fig.~\ref{fig:thermo_diagram}, as presented in the following definition.
The energetic contributions to $\mathcal W$ will be detailed further for our framework (Fig.~\ref{fig:generic_algo}) introduced in Sec.~\ref{sec:energetic_costs}.
\begin{restatable}[Energy consumption]{definition}{DefEnergyConsumption}
\label{def:DefEnergyConsumption}
    We define the energy consumption of computation modeled in Fig.~\ref{fig:thermo_diagram} as
    \begin{align}\label{eq:energy_consumption_def}
        \mathcal W \coloneqq \mathcal W_\mathrm{gates} + \Delta E^{(\mathrm{in})} - \Delta E^{(\mathrm{out})}.
    \end{align}
\end{restatable}
The demand of the computation being cyclic is critical to energy conservation; i.e., in the limit of closing this thermodynamic cycle ideally without error\footnote{One might think that with a nonzero error $\varepsilon$ in reinitialization after a round of computation, it would be possible that the final state after the reinitialization might have lower energy than the initial state of the computer's memory since the initial state in our framework is not assumed to be a ground state; thus, one could extract energy using the $\varepsilon$-difference between the initial and final states after one round of the computation.
This reasoning, however, does not capture the fact that this energy is borrowed from the computer's memory only temporarily in this single round.
After all, the initial state of the computer in the second round of the computation is then no longer exactly equal to the initial state of the first round, but also, in this second round, we still need to close the thermodynamic cycle by approximately reinitializing the computer's memory state.
When averaging over many rounds of the computation using the same computer's memory, the energy extraction must perish.
To capture this, our analysis deals with the energy conservation on average in the limit of $\varepsilon\to 0$.}, the energy conservation demands
\begin{align}\label{eq:energy_conservation}
    \mathcal W = \mathcal Q.
\end{align}
The energy consumption $\mathcal W$ is the sum of the non-dissipative energy exchanges between external systems and the computer, which can be understood as the work cost of performing the computation.
On one hand, we will analyze all the contributions to $\mathcal{W}$ to identify an upper bound of energy consumption of quantum computation.
On the other hand, the formulation~\eqref{eq:energy_conservation} of energy conservation will be the basis for proving the implementation-independent lower bound on the energy consumption of classical computation by analyzing $\mathcal Q$.

\subsection{\label{sec:computational_setting}Computational framework}

Based on the model in the query-complexity setting (Fig.~\ref{fig:thermo_diagram}) that we conceptualized in the previous section, we here proceed to describe the explicit computational framework that represents the computation in terms of (quantum and classical) circuits, so that we can analyze all the contributions to the energy consumption therein.
The significance of the computational framework that we will formulate here is that it makes it possible to quantitatively study bounds of energy consumption in the computation, especially, achievable upper bounds for quantum computation and fundamental lower bounds for classical computation.
In the following, we discuss our core idea to obtain such bounds by presenting the definition of classical and quantum oracles as in the conventional query-complexity setting (Fig.~\ref{fig:oracle_QvsC}), followed by formulating the full detail of the computational framework (Fig.~\ref{fig:generic_algo}) for our analysis.

\paragraph*{Oracle in the classical case.}
To derive a fundamental nonzero lower bound of energy consumption of classical computation, our idea here is to formulate the framework in such a way that we can study the energy-consumption requirement for solving the computational task as a whole rather than implementing each gate.
Conventionally, any classical computation can be rewritten as a classical circuit in a reversible way~\cite{bennett1973logical}.
Formulating a computational task as a decision problem reduces the number of bits to output at the end of the computation to one bit.
All other bits are to be reset to their initial state at the end of the computation.
But if the original computation is written as a reversible circuit, it is possible in principle to uncompute the original circuit and get all the energy back that was invested into the computer to obtain the output, up to negligible $O(1)$ energy consumption to output the single bit, as established by Bennett in Refs.~\cite{bennett1973logical,Bennett1982}.
A nonzero lower bound on the energy consumption can, therefore, arise only if uncomputation is impossible due to some additional restriction in the setting.
A setting where this happens is when the computation is structured with oracles as in the query-complexity setting.

\begin{figure}
    \centering
    \includegraphics[width=3.4in]{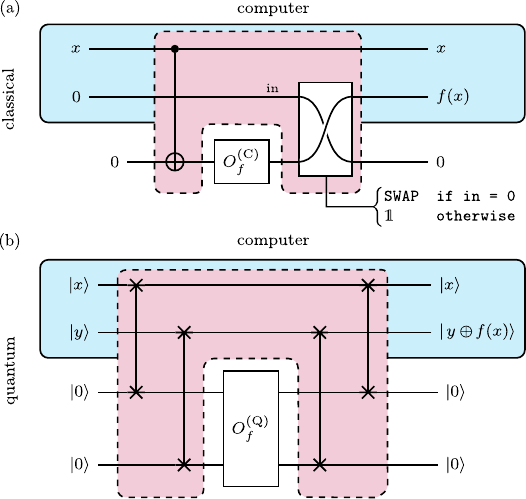}
    \caption{A classical oracle (a) and a quantum oracle (b) in our computational framework of classical and quantum computation, respectively.
    (a) In the classical case, the computer sends a copy of the input bitstring $x$ to the oracle.
    The oracle then allocates an output memory region in the computer (initialized in $00\cdots 0$ by the computer). In the allocated memory, the oracle writes out the corresponding output bitstring $f(x)$.
    As a result, each query to the oracle increases the information about the input-output relation $x\mapsto f(x)$ of the function $f$, which is to be stored in the computer until erased irreversibly.
    The oracle does nothing if there is no more memory in the computer to allocate.
    (b) In the quantum case, by convention, the oracle is implemented as a unitary map as defined in~\eqref{eq:oracle_Q}.
    The computer performs swap gates to set the input state into the input register for the oracle.
    After the oracle transforms the input state into the output state, the computer performs swap gates to receive this output state.}
    \label{fig:oracle_QvsC}
\end{figure}

\begin{figure*}
    \centering
    \includegraphics[width=7.0in]{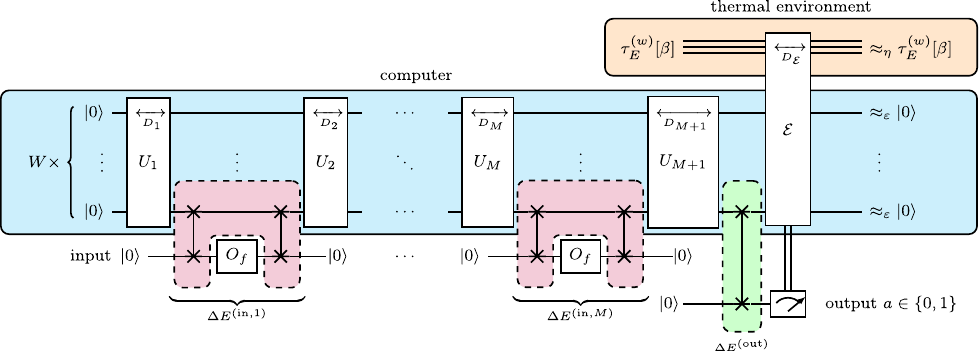}
    \caption{
    A circuit representing our framework for implementing the thermodynamic cycle of computation in Fig.~\ref{fig:thermo_diagram}. The circuit in the figure represents, ordered from top to bottom, the thermal environment, the $W$ (qu)bits of the (quantum or classical) computer, the outsourced oracle operation $O_f$ for input, and finally the output (qu)bit. The computation starts with all $W$ (qu)bits initialized in a fixed state $\ket{0}^{\otimes W}$.
    Then a sequence of $M+1$ operations $U_k$ is performed with oracle operations $O_f$ spaced in-between. Each operation $U_k$ has a depth of $D_k$ with respect to some gate set.
    Interaction between the computer and the oracle (and that between the computer and the output) is performed as described in Fig.~\ref{fig:oracle_QvsC} (the figure describes the quantum oracle, i.e., the case of Fig.~\ref{fig:oracle_QvsC}(b)), while the internal implementation of the oracle is ignored by conversion of the query-complexity setting. A single-(qu)bit measurement in the basis $\{\ket{a}:a=0,1\}$ of the output register is performed outside the computer to obtain the output $a\in\{0,1\}$ of the computation for a decision problem.
    At the end of the computation, the $W$ (qu)bits are reinitialized through the interaction $\mathcal E$ with an environment $E$ at inverse temperature $\beta$.
    In an idealized scenario, this would return each (qu)bit in the computer perfectly back to their initial state $\ket{0}$ so as to make the thermodynamic cycle closed.
    To capture a more realistic scenario, our analysis may allow for some non-zero infidelity $0<\varepsilon\ll 1$, by which each (qu)bit's state after reinitialization can deviate from $\ket{0}$ by $\varepsilon$ at the end of a single round of the computation.
    The required depth for this reinitialization is denoted by $D_\mathcal{E}$.
    We also require that each set of qubits in the environment for initializing the $w$th qubit ($1\leq w\leq W$) should be disturbed at most by $\eta$ from its initial thermal state $\tau_E^{(w)}[\beta]$, as analyzed in Sec.~\ref{sec:quantum_upperbound}.}
    \label{fig:generic_algo}
\end{figure*}

In this setting, the input to the computation is given via queries to an oracle that computes some function
\begin{align}
\label{eq:f_oracle}
    f:\{0,1\}^N\to\{0,1\}^{W_{f}},
\end{align}
where we let $N$ be the number of bits of $f$'s input (which is the size of the computational problems in the query-complexity setting), and
\begin{align}
\label{eq:W_f}
    W_{f}=O(\mathrm{poly}(N))
\end{align}
is the number of bits of $f$'s output.
For example, Simon's problem~\cite{Simon1994,Simon1997} is defined for  $f:\{0,1\}^N\rightarrow \{0,1\}^N$, i.e.,
\begin{align}
    W_{f}=N.
\end{align}

In the classical case, 
for an input string $x\in\{0,1\}^N$, the classical oracle lets the computer know the corresponding output string $f(x)\in\{0,1\}^{W_{f}}$ of the function $f$, where the oracle internally computes a map
\begin{align}
\label{eq:oracle_C}
     x \xrightarrow{O_f^{\mathrm{(C)}}} f(x).
\end{align}
More precisely, as visualized in Fig.~\ref{fig:oracle_QvsC}(a), the classical oracle receives a copy of the $N$-bit input string $x$ from the computer, allocates an output memory region in the computer (which is initialized by the computer as $00\cdots 0$), and writes out the corresponding output bitstring $f(x)$ in the allocated memory region.
The classical algorithm is then structured in such a way that the oracle queries are carried out outside of the computer, and the agent carrying out the algorithm on the computer is ignorant about the internals of the oracle so as to compute and learn a property of $f$ from the information obtained via the queries.
Since the oracle is a black box, the computer's memory initially has no information about $f$.
Each call of the oracle for new $x$ replaces $0$ in an initialized memory region of the computer with $f(x)$ and thus adds the information about the input-output relation of $f$ to the computer's memory.
To remove this information about $f$ from the computer's memory, we have to erase the information in an irreversible way independently of $f$, using the Laudauer erasure.\footnote{Once the computer knows $f(x)$, it is impossible to completely uncompute the information of $f(x)$ unless the computer irreversibly erases $f(x)$ by the Landauer erasure. 
If we were to allow for calling the classical oracle for the same input $x$ twice, then the computer would obtain two copies of $f(x)$ in the computer's memory, and in this case, by taking the bit-wise XOR operation, one of the copies could be uncomputed; however, even in this case, it is still impossible to reversibly uncompute both copies completely as the circuit for uncomputing $f$ requires knowledge about $f$.
For simplicity of analysis, we assume that the different queries to the classical oracle use different inputs $x$ since it is useless to make more than one query for the same $x$.
}
This non-invertibility is what breaks Bennett's uncomputing argument~\cite{Bennett1982,bennett1973logical,Watrous2009,Baumeler2019} in the classical case, which makes it possible to establish a nonzero lower bound on the energy consumption for any classical algorithm to solve a computational problem, as will be shown in Sec.~\ref{sec:classical_lowerbound}.

\paragraph*{Oracle in the quantum case.}
As for the quantum case, as described in Fig.~\ref{fig:oracle_QvsC}(b), the quantum version of the oracle for $f$ in~\eqref{eq:f_oracle} conventionally implements a unitary map that transforms each computational-basis state as~\cite{BUHRMAN200221,Watrous2009,Nielsen2010}
\begin{align}
\label{eq:oracle_Q}
    \ket{x,y} \xrightarrow{O_f^{\mathrm{(Q)}} } \ket{x,y\oplus f(x)},
\end{align}
for all $x\in \{0,1\}^N$ and $y,f(x)\in\{0,1\}^{W_f}$ with $W_f$ in~\eqref{eq:W_f}, where $\oplus$ is the bit-wise XOR operation.
For the superposition states $\sum_{x,y}\alpha_{x,y}\ket{x,y}$, $O_f^{\mathrm{(Q)}}$ acts linearly on each term of $\ket{x,y}$ as~\eqref{eq:oracle_Q}.
Note that we do not need the non-invertibility of the quantum oracle to derive the achievable upper bound of energy consumption for quantum computation, as we will show in Sec.~\ref{sec:quantum_upperbound}.\footnote{
Since we follow the conventional definition~\eqref{eq:oracle_Q} of quantum oracles in the query-complexity setting, unlike the classical case, the input $f(x)$ can be uncomputed by calling the quantum oracle twice;
in particular, since $f(x)\oplus f(x) = 0$,
the quantum oracle is its own inverse (i.e., $O_f^{\mathrm{(Q)}}\circ O_f^{\mathrm{(Q)}}$ is an identity map) and can, in principle, be used to uncompute all the oracle outputs.
Due to this invertibility of the quantum oracle, it is a priori not clear whether there exists a fundamental lower bound on the energy consumption of the quantum computation in this setting.
In practice, however, one may encounter another situation where the action of the quantum oracle except for $y=0$ in~\eqref{eq:oracle_Q} is undefined.
Simon's problem can still be solved with such a quantum oracle queried with $y=0$.
Then, it may be operationally infeasible to implement the inverse $\qty(O_f^{(\mathrm{Q})})^{-1}$ by only calling $O_f^{(\mathrm{Q})}$ with $y=0$ with a reasonable number of queries.
Still, in light of recent findings showing protocols for reversing unknown unitary operations~\cite{PhysRevLett.126.150504,Yoshida2023universal,yoshida2023reversing,Chen2024}, it is currently unknown how to rigorously impose the non-invertibility to quantum oracles used in quantum algorithms.
Finding a provable lower bound for the energy consumption of quantum computation in a variant of the query-complexity setting is therefore left for future work.}

\paragraph*{Formulation of our computational framework based on query-complexity setting.}
Using these oracles in both classical and quantum cases,
we here present the formulation of our computational framework, as shown in Fig.~\ref{fig:generic_algo}.
While it is a standard setting in computational complexity theory~\cite{Watrous2009}, oracle-based computation comes with important assumptions that are worth being pointed out.
For one, even if the algorithm can be written as a polynomial-size circuit with a polynomial number of oracle queries in terms of the problem size $N$, assuming a black-box oracle generally does not guarantee that the internals of the oracle itself can be implemented efficiently in polynomial time in $N$, which may be a usual criticism in computational complexity theory for oracle-based computation~\cite{BUHRMAN200221}.
For another, similar to the time of implementing the oracle, the required energy consumption for implementing the internals of the oracle may not be polynomial in $N$ in general either.
For an actual implementation of the oracle, the bounds on the runtime and the energy consumption eventually depend on its physical details;
however, we will propose how to make this additional cost negligible in practice, which we will discuss along with explaining the relevant energy consumption of the oracles in Sec.~\ref{sec:energetic_costs}.

In the most general form, an oracle-based algorithm based on Fig.~\ref{fig:thermo_diagram} can be structured as the circuit shown in Fig.~\ref{fig:generic_algo}, where the reinitialization of the computer's state is performed right after the output in an irreversible way based on Landauer erasure~\cite{Landauer1961}.
In particular, given a family of decision problems with problem size $N$ for a function
$f:\{0,1\}^N\to\{0,1\}^{W_f}$ in~\eqref{eq:f_oracle},
the algorithm is described by the circuit with the number of oracle queries $M(N)$, the width $W(N)$, and the depth $D_k(N)$ for the $k$th part ($1\leq k\leq M+1$).
The function $f$ is unknown at the beginning of the computation, and the knowledge on $f$ is input to the computer via multiple queries to the oracle.
Usually, the algorithm's width and depth, as well as the number of queries, monotonically increase as the problem size $N$ grows.
For Simon's problem, we will explicitly write out their dependence on $N$ in Sec.~\ref{sec:example}.
In the following, we may omit the $N$ dependency and write $M(N)$, $W(N)$, and $D_k(N)$ as $M$, $W$, and $D_k$, respectively, for simplicity of presentation.

In Fig.~\ref{fig:generic_algo}, the computer is initialized in a fixed pure state $\ket{0}^{\otimes W}$ of $W$ (qu)bits,
where the number of (qu)bits should be at least the number of bits required for storing the input and the output of the $N$-bit function $f$ in~\eqref{eq:f_oracle}, i.e., with $W_f$ in~\eqref{eq:W_f},
\begin{equation}
\label{eq:W_bound}
    W\geq N+W_f.
\end{equation}
Note that in the classical case, storing information obtained from $M$ queries to the classical oracle requires $W\geq (N+W_f)M$, but our analysis of the lower bound of energy consumption of classical computation does not explicitly use such requirements since the lower bound will be derived from the amount of information (i.e. the entropy) rather than the number of qubits, as will be shown in Sec.~\ref{sec:classical_lowerbound}.
In both classical and quantum cases,~\eqref{eq:W_bound} holds true as the lower bound of $W$.

Then, a sequence of $M+1$ unitary circuits $U_k$ of computational depth
\begin{align}
    \label{eq:D_k}
    D_k\geq 0
\end{align}
are performed, with the computational depth defined with respect to the decomposition of those unitaries into a product of gates selected from a finite gate set, such as Clifford and $T$ gates, where each gate in the gate set acts on at most a constant number of (qu)bits.\footnote{The particular choice of the gate set for quantum computation is irrelevant for our main results due to the Solovay-Kitaev Theorem \cite{Kitaev1997,Nielsen2010}, as it would, at worst, contribute with a polylogarithmic correction to the circuit depth. For simplicity of the analysis, we assume that we use a gate set such that the swap gate is implementable within a constant depth; for example, if the gate set includes the CNOT gate, the swap gate is implementable with three CNOT gates~\cite{Nielsen2010}. For universal classical computation, the classical Toffoli gate can be used.}

Between each operation $U_k$, the oracle $O_f^{\mathrm{(C)}}$ in~\eqref{eq:oracle_C} for the classical and $O_f^{\mathrm{(Q)}}$ in~\eqref{eq:oracle_Q} for the quantum case is called in total $M$ times.
For simplicity of notation, we may omit the superscripts of the classical and quantum definitions of the oracles to simplity write $O_f$ when it is clear from the context.
As shown in Fig.~\ref{fig:oracle_QvsC}, upon each query to $O_f$, the input state into $O_f$ is pulled out of the computer to an input register to the oracle; then, this state of the input register goes through the oracle $O_f$, and the state output from $O_f$ is put into the computer.
In our analysis, the runtime of $O_f$ is assumed to be negligible by the convention of the query-complexity setting, so the computer does not have to perform identity gates to wait for the runtime of $O_f$.
We let
\begin{equation}
\label{eq:D_swap}
    D_\mathrm{in}=O(1)
\end{equation}
denote the depth of each part of the circuit to make a query to the oracle for the input (i.e., each of the red parts in Figs.~\ref{fig:oracle_QvsC} and~\ref{fig:generic_algo}).

At the end of the computation, another swap gate (i.e., a green part in Fig.~\ref{fig:generic_algo}) is also performed to pull out an output state from the computer to an output register, which is a single (qu)bit for a decision problem.
The depth of a circuit of this output part (i.e., the green part of Fig.~\ref{fig:generic_algo}) is denoted by
\begin{align}
\label{eq:D_out}
    D_\mathrm{out}=O(1).
\end{align}
Reading the output register by the measurement in basis $\{\ket{a}:a=0,1\}$ yields the single-bit output $a\in\{0,1\}$ of the algorithm.

Finally, the computer's memory is reinitialized in the part we label $\mathcal E$ in Fig.~\ref{fig:generic_algo}, using an environment in a thermal state at inverse temperature
\begin{equation}
\label{eq:beta}
    \beta>0.
\end{equation}
In particular, let $\rho'$ denote the $W$-qubit state in the computer just after the reinitialization; then, for each $w=1,\ldots,W$, we require that each single-qubit reduced state for $\rho'$ of the $w$th qubit should be $\varepsilon$-close to the initial state $\ket{0}$ of the computer, i.e.,
\begin{equation}
\label{eq:fidelity_condition}
    \bra{0}\tr_{\overline{w}}[\rho']\ket{0}\geq 1-\varepsilon,
\end{equation}
where $\tr_{\overline{w}}$ is the partial trace over all the $W$ qubits but the $w$th.
We also require that during reinitializing the state of the $w$th qubit, the environment, initially in the Gibbs state $\tau_E^{(w)}[\beta]$, should also be disturbed at most by a small constant
\begin{equation}
\label{eq:eta_condition}
    \eta>0,
\end{equation}
as will be detailed more in~\eqref{eq:environment_condition} of Sec.~\ref{sec:quantum_upperbound}.
The depth of this reinitialization part of the circuit is denoted by $D_\mathcal{E}$.

As a whole, the circuit depth of the circuit represented in Fig.~\ref{fig:generic_algo} is denoted by
\begin{align}
\label{eq:D}
    D\coloneqq\sum_{k=1}^{M+1}D_k+MD_\mathrm{in}+D_\mathrm{out}+D_\mathcal{E}=D_\mathrm{comp}+D_\mathcal{E},
\end{align}
where we collectively represent the circuit depths~\eqref{eq:D_k},~\eqref{eq:D_swap}, and~\eqref{eq:D_out} for the operations $U_1,\dots,U_{M+1}$, the input, and the output as
\begin{align}
\label{eq:D_U}
    D_\mathrm{comp} \coloneqq \sum_{k=1}^{M+1} D_k+MD_\mathrm{in}+D_\mathrm{out}.
\end{align}
For quantum computation, the operations $U_k$ are unitaries acting on the $W$ qubits in the computer, and the state of the qubits can be in a generic superposition.
For classical computation, the operations $U_k$ may only be reversible classical logic operations, and the states of bits are always dephased in their energy-eigenbasis; in other words, the classical computer's state is in a probabilistic mixture of bit strings, which we may represent as diagonal density operators to use the same notation as those in quantum computation.

\subsection{\label{sec:energetic_costs}Energy consumption}

Accounting for the energy consumption of quantum computation in full depth requires considering both microscopic and macroscopic contributions.
Macroscopic contributions would be the energetic costs coming from refrigerators, readout amplifiers, and other electronics not directly involved in qubit control, as studied in, e.g., Ref.~\cite{FellousAsiani2022}.
Under microscopic contributions, we understand all energetics directly relevant to performing operations on the qubits, which is what we are interested in for the scope of this article.
In Definition~\ref{def:DefEnergyConsumption}, we have already given a high-level definition of what we refer to as the energy consumption of the computation in our framework.
Let us subdivide it further and explain the individual contributions:
\begin{enumerate}[label=(\roman*)]
    \item\label{item:qubit_energy_cost} \textbf{Energetic cost:} Any non-trivial operation changing the state of a memory system with non-degenerate energy levels in the computer may cost some energy due to the change of the energy of the computer's state. Per gate $U$, this cost is given, in expectation, by
    \begin{equation}
        \Delta E^{(U)} = \tr[H(U\rho U^\dagger - \rho)],
    \end{equation}
    where $H$ is the Hamiltonian of the system, and $\rho$ is the state of the system before performing $U$.
    \item\label{item:control_cost} \textbf{Control cost:} Implementing a desired unitary operation $U$ on a target quantum system by some Hamiltonian dynamics generally requires an auxiliary system for the control to make the implementation energy-preserving as a whole. 
    In general, this implementation may change the state of the auxiliary system, and compensating for this change has to require some thermodynamic control cost that we label as
    \begin{equation}
    \label{eq:control_cost_positive}
        E_\mathrm{ctrl.}^{(U)}\geq0,
    \end{equation}
    which can be considered heat dissipation and is by definition a nonnegative quantity. Note that we here may not analyze model-specific control costs but just write the control cost as $E_\mathrm{ctrl.}^{(U)}$ by abstracting the details, to maintain the general applicability of our analysis. See also Appendix~\ref{appendix:model_details} for a discussion on model-specific instantiation of the control cost.
    \item\label{item:info_theoretic_cost} \textbf{Initialization cost: } Landauer erasure ($\mathcal{E}$ in Fig.~\ref{fig:generic_algo}) to reinitialize the computer's memory comes at an additional work cost, bounded by the heat dissipation $\mathcal Q_E$ into the environment necessary for erasure, as will be defined in more detail as~\eqref{eq:heat}. Alternatively, we can think of this initialization cost as the one from type~\ref{item:qubit_energy_cost}, but not on the memory register of the computer but on the environment.
\end{enumerate}

When it comes to deriving an achievability result for quantum computation, we account for the three types of energy costs listed above: \ref{item:qubit_energy_cost} energetic costs on the memory qubits, \ref{item:control_cost} control costs, and \ref{item:info_theoretic_cost} initialization costs that arise as heat dissipation into a thermal environment when states in the computer are irreversibly erased into a fixed pure state.
In the following, we write out the decomposition into these three costs individually for all the gates in the framework of Fig.~\ref{fig:generic_algo}.

The energy $\mathcal{W}_\text{gates}$ that the agent invests into performing the computation in Fig.~\ref{fig:thermo_diagram} can be decomposed into contributions from all the gates listed in Fig.~\ref{fig:generic_algo}; that is, all the operations for $U_k$, the operations for the input from the oracles $O_f$ and the output (the red and green parts of Fig.~\ref{fig:generic_algo}, respectively), and the erasure $\mathcal E$ come at some work cost, which we write as
\begin{equation}
    \mathcal{W}_\text{gates}=\sum_{k=1}^{M+1}\mathcal W_\mathrm{gate}^{(U_k)} + \sum_{k=1}^{M}\mathcal W_\mathrm{gate}^{(\mathrm{in},k)} + \mathcal W_\mathrm{gate}^{(\mathrm{out})}+ \mathcal W_\mathrm{gate}^{(\mathcal{E})}.
\end{equation}

In particular,
for each $k\in\{1,\ldots,M\}$,
the work cost of $U_k$
is given by
\begin{align}
\label{eq:W_gate_U_k_def}
    \mathcal{W}_\text{gate}^{(U_k)} = \Delta E^{(U_k)} + E_\text{ctrl.}^{(U_k)},
\end{align}
where $\Delta E^{(U_k)}=\tr[H\qty(U_k\rho_{k-1}U_k^\dagger - \rho_{k-1})]$ is the change of energy in the memory of the computer, $H$ is the Hamiltonian of the $W$ qubits in the computer, and $\rho_{k-1}$ is the memory state just before application of $U_k$.

As for the energy consumption for the $k$th query to the oracle for the input,
we let $\Delta E^{(\text{in},k)}$ denote the change of energy of states of the computer's memory before and after making the $k$th query,
where the states before and after each query are shown in Fig.~\ref{fig:oracle_QvsC}.
Note that the states (and thus the energy) outside the computer in Fig.~\ref{fig:oracle_QvsC} remain unchanged. 
Apart from this energetic cost,
the rest of the cost for the $k$th query is
\begin{align}
\label{eq:contro_cost_in_k}
    \mathcal W_\text{gate}^{(\text{in},k)} = E_\text{ctrl.}^{(\text{in},k)},
\end{align}
where $E_\text{ctrl.}^{(\text{in},k)}$ is the control cost of performing a red part of Figs.~\ref{fig:oracle_QvsC} and~\ref{fig:generic_algo} for the $k$th query.\footnote{
Due to the conceptual separation of computer and oracle, our theoretical analysis does not count the control cost of performing the oracle's internal operations as part of the computer's energy consumption; however, in practical cases, it is indeed possible to make this additional control cost negligible.
For example, while the computer is a universal machine, the oracle is a single-purpose device for the input and thus may be made up of specialized parts with negligible energy loss compared to the computer.
In Sec.~\ref{sec:experiment}, we also explicitly propose an experimental setup with clarification on how the oracle can be implemented efficiently in the case of Simon's problem~\cite{Simon1994,Simon1997}, using a cryptographic primitive.
In this setup, the non-invertibility of the classical oracle is justified by the computational hardness assumption on the security of the cryptographic primitive in practice, in place of the black-box assumption in our theoretical analysis.}
The total energy transfer from the oracle into the computer's memory amounts to the contribution $\Delta E^{(\text{in})}$ in Fig.~\ref{fig:thermo_diagram}, given by the sum of the energy costs $\Delta E^{(\text{in},k)}$ of all $M$ queries, i.e.,
\begin{equation}
    \Delta E^{(\text{in})}=\sum_{k=1}^{M}\Delta E^{(\text{in},k)}.
\end{equation}

Similarly, 
we let
\begin{align}
\label{eq:E_out_def}
    -\Delta E^{(\text{out})}\in\mathbb{R}
\end{align}
denote the change of energy in the computer's memory before and after the output (which can be negative if the output state has higher energy than the initial state $\ket{0}$).
The negative sign in~\eqref{eq:E_out_def} follows from that in~\eqref{eq:energy_consumption_def}, respecting the direction of the arrow in Fig.~\ref{fig:thermo_diagram}.
Apart from this energetic cost,
the rest of the cost of outputting the result of the algorithm is
\begin{equation}
\label{eq:W_gate_out}
    \mathcal W_\text{gate}^{(\text{out})}=E_\text{ctrl.}^{(\text{out})},
\end{equation}
where $E_\text{ctrl.}^{(\text{out})}$ is the control cost of performing a green part of Fig.~\ref{fig:generic_algo} for the output.\footnote{
In our framework of Fig.~\ref{fig:generic_algo}, where the computer never performs measurements, the measurement to read the output is performed outside the computer.
For decision problems, only a single (qu)bit of information is output.
The energy consumption of performing the measurement of this (qu)bit with finite precision is $O(1)$ and does not grow with the problem size.
Thus, the energy consumption of realizing this single-qubit measurement to a finite precision is negligible, which is not counted in our analysis to simplify the bounds.
}

Reinitializing the memory at the end of the computation also comes at a work cost given by
\begin{align}
\label{eq:W_gate_E_def}
    \mathcal W_\text{gate}^{(\mathcal E)} = \Delta E^{(\mathcal E)} + E_\text{ctrl.}^{(\mathcal E)} + \mathcal Q_E.
\end{align}
The three terms account for the energy change of the qubits $\Delta E^{(\mathcal E)}=\tr[H(\rho'-\rho)]$, the control cost $E_\text{ctrl.}^{(\mathcal E)}$, and the initialization cost $\mathcal Q_E$ of the heat dissipation into the environment, where $\rho'$ is the $W$-qubit state after erasure and should satisfy~\eqref{eq:fidelity_condition}, and $\rho$ the state just before erasure.

Control costs in (quantum) computation arise because performing a unitary operation $U$ on a system usually requires external control.
The system itself generically evolves according to its Hamiltonian, but if one desires to perform some non-trivial operation $U$ on this target system, additional degrees of freedom of an auxiliary system should be needed.
For example, by means of some external control field, an interaction Hamiltonian $H_\text{int}$ with the target system can be turned on for a duration of time $t$, which then generates the unitary $U=e^{-iH_\text{int}t}$.
In a simple model that we give in Appendix~\ref{appendix:model_details} as an illustration, the state of the auxiliary system for the control degrades through the interaction.
Reinitializing the state of the auxiliary system then comes at some positive thermodynamic cost, whose precise form depends on the physical model of the actual system of interest.
This energy is dissipated into some environment $E'$, which we do not specify further for the sake of the generality of our analysis.
Rather than specifying the model, we work with generic positive constants for the control costs that give the overall control dissipation after summing up, i.e.,
\begin{equation}
\label{eq:Q_E_prime}
    \mathcal Q_{E'}=\sum_{k=1}^{M+1}E_\mathrm{ctrl.}^{(U_k)}+\sum_{k=1}^{M}E_\mathrm{ctrl.}^{(\text{in},k)}+E_\text{ctrl.}^\text{(out)}+ E_\mathrm{ctrl.}^{\mathcal E}\geq 0.
\end{equation}
The environment $E'$ into which this energy is dissipated is not explicitly illustrated in our computational framework of Fig.~\ref{fig:generic_algo}; there, only the initialization cost $\mathcal Q_E$ appears.
Adding these two terms together gives
\begin{equation}
\label{eq:Q}
    \mathcal Q = \mathcal Q_E + \mathcal Q_{E'},
\end{equation}
as in the thermodynamic diagram of Fig.~\ref{fig:thermo_diagram}.

With these definitions at hand, we can make the sanity check to verify that the energy conservation~\eqref{eq:energy_conservation} is true in the limit where the error $\varepsilon$ in~\eqref{eq:fidelity_condition} is small enough to close the thermodynamic cycle of the computation.
In particular, writing out~\eqref{eq:energy_consumption_def} explicitly gives
\begin{widetext}
\begin{align}
\label{eq:proof_energy_conservation}
    \mathcal W &= \left(\sum_{k=1}^{M+1}\mathcal W_\mathrm{gate}^{(U_k)} + \sum_{k=1}^{M}\mathcal W_\mathrm{gate}^{(\mathrm{in},k)} + \mathcal W_\mathrm{gate}^{(\mathrm{out})}+ \mathcal W_\mathrm{gate}^{(\mathcal{E})}\right)+ \left(\sum_{k=1}^{M}\Delta E^{(\text{in},k)}\right)-\Delta E^{(\text{out})} \\ 
    \label{eq:proof_energy_conservation_2}
    &= \left(\sum_{k=1}^{M+1}\Delta E^{(U_k)} + \sum_{k=1}^{M}\Delta E^{(\text{in},k)}- \Delta E^{(\text{out})} + \Delta E^{(\mathcal E)}\right)+\left(\sum_{k=1}^{M+1}E_\mathrm{ctrl.}^{(U_k)}+\sum_{k=1}^{M}E_\mathrm{ctrl.}^{(\text{in},k)}+E_\text{ctrl.}^\text{(out)}+ E_\mathrm{ctrl.}^{(\mathcal E)}\right) + \mathcal Q_E\\
    &\to 0 + \mathcal Q_{E'} + \mathcal Q_E = \mathcal Q,
\end{align}
\end{widetext}
in the limit of $\varepsilon\to 0$.
In~\eqref{eq:proof_energy_conservation}, following the definition~\eqref{eq:energy_consumption_def} of the energy consumption in Definition~\ref{def:DefEnergyConsumption}, we have represented the computer's energy consumption $\mathcal W$ as the sum of all work costs $\mathcal W_\mathrm{gates}^{(\cdots)}$ and the energy exchanges with the input $\Delta E^{(\text{in})}$ and output $-\Delta E^{(\mathrm{out})}$.
Then, in the first term of~\eqref{eq:proof_energy_conservation_2}, given $\varepsilon\rightarrow 0$, the qubit energy changes of type~\ref{item:qubit_energy_cost} together sum to zero because the computation together with reinitialization is a cyclic process on the computer.
What is left is the second term, i.e., a sum of all control costs~\ref{item:control_cost}, and the third term, i.e., the initialization cost~\ref{item:info_theoretic_cost}. Together, they yield energy conservation $\mathcal W = \mathcal Q$.
This energy conservation is crucial for our analysis of the upper and lower bounds of energy consumption of computation in Sec.~\ref{sec:classical_lowerbound}.

We also remark that, along the way of the analysis of the achievable upper bound of the energy consumption of quantum computation in Sec.~\ref{sec:quantum_upperbound}, it is also important to check that the work costs (which may include both the energetic cost and the control cost as in~\eqref{eq:W_gate_U_k_def} and~\eqref{eq:W_gate_E_def}) are feasibly bounded in every step of quantum computation, whereas the energetic costs cancel out in evaluating the net energy consumption $\mathcal W$ as shown in~\eqref{eq:proof_energy_conservation_2}. 
To see the importance, we recall, for example, that the work cost $\mathcal W_\mathrm{gate}^{(U_k)}$ is the cost that the agent needs to invest at each step $U_k$ of the computation.
For the implementability of each step of the computation, the work cost for each step should be bounded; otherwise, the agent would need to invest too much energy to run an intermediate step of the quantum algorithm.
For this reason, in Sec.~\ref{sec:quantum_upperbound}, we will provide achievable upper bounds of both the work costs and the control costs along the way of our analysis while using the control costs to present theorems on the upper bounds of the net energy consumption.

\section{\label{sec:results}General bounds on energy consumption of computation}

In this section, we derive general upper and lower bounds of energy consumption of conducting quantum and classical computation, respectively, in the framework formulated in Sec.~\ref{sec:setting}.
First, in Sec.~\ref{sec:quantum_upperbound}, we provide a detailed analysis of the achievable energy consumption, accounting in particular for finite-fidelity and finite-step Landauer erasure for reinitialization and quantum error correction.
Second, in Sec.~\ref{sec:classical_lowerbound}, we derive a fundamental lower bound on the energy consumption.

\subsection{\label{sec:quantum_upperbound}Upper bound on energy consumption for quantum computation}

In this section, we derive a general upper bound
\begin{equation}
\mathcal W\leq  \mathcal W^{(\text{Q})}
\end{equation}
of energy consumption that is achievable by quantum computation within the framework of Fig.~\ref{fig:generic_algo}.
The general idea for estimating $\mathcal W^{(\text{Q})}$ is to break down the budget~\eqref{eq:proof_energy_conservation} of energy consumption into the contributions from the individual elementary gates and bound each contribution from above.
The achievable upper bound of the energy consumption of quantum computation is determined not only by the query complexity; the total volume (width times depth) of the quantum circuit for implementing the quantum algorithm is also crucial.\footnote{
Indeed, there are edge cases where even if the query complexity scales polynomially, the quantum algorithm may require an exponential runtime to process the polynomial amount of information obtained from the queries.
In such a case, the quantum algorithm may require exponentially large energy consumption due to the exponential runtime.}

In the following analysis, we begin with introducing the maximum work cost $\mathcal W_\text{max,el.}$ per qubit for implementing an elementary gate in the gate set.
Using $\mathcal W_\text{max,el.}$, we will bound the work costs and the energy transfers between the oracle and the computer, i.e., each term in~\eqref{eq:proof_energy_conservation}, as $\mathcal W_\text{max,el.}$ times the volume (i.e., width times depth) of the part of the circuit in Fig.~\ref{fig:generic_algo} corresponding to each term.
At the same time, to estimate the overall energy consumption $\mathcal W$,
we will also bound the control costs and the initialization cost appearing in~\eqref{eq:proof_energy_conservation_2}.
Our analysis of these bounds consisted of the following three steps,
where we do a preliminary idealized estimate without considering gate errors in the first two steps, and then we account \textit{a posteriori} for the overheads of quantum error correction in the third step.
\begin{enumerate}[label=(\arabic*)]
    \item\label{item:qstep_1} \textit{Contributions from gates, input, and output.} We estimate the contributions $\mathcal W_\text{gate}^{(U_k)}$, $\mathcal W_\text{gate}^{(\text{in},k)}$, $\mathcal W_\text{gate}^{(\text{out})}$, $\Delta E^{(\text{in},k)}$, and $\Delta E^{(\text{out})}$ for generic input sizes, widths, and depths.
    \item\label{item:qstep_2} \textit{Contribution from reinitialization by finite Laudauer erasure.} We bound the cost $\mathcal W_\text{gate}^{\mathcal E}$ of reinitializing the qubits, for which we derive a bound for finite Landauer erasure.
    \item\label{item:qstep_3} \textit{Contribution from quantum error correction.} Lastly, we correct for the space and time overhead coming from quantum error correction.
\end{enumerate}

\paragraph*{Maximum work cost for implementing an elementary gate.}
We here analyze the work cost per qubit for performing each elementary quantum gate in the gate set, which is to be bounded by a constant independent of the size of the problem the gates are used to solve.

The work cost of each gate is determined as follows.
In degenerate quantum computing, that is, when all computational basis states are energy-degenerate, unitary operations would be free in terms of the energetic cost~\cite{Brandao2013,Lostaglio2019,Liu2019}.
Any implementation of a quantum computer on a physical platform inherently requires some energy splitting of the qubits in order to control their state.
Several works have been investigating the fundamental lower and achievable energy cost of performing unitary operations~\cite{Misra2016,Bakhshinezhad2019,Chiribella2021}, often with the goal to estimate \textit{control costs} to achieve a certain target fidelity in implementing a given unitary gate.
Physical realizations of the quantum gate are never perfect, but the requirement of the threshold theorem for fault-tolerant quantum computation (FTQC) is that the physical error rate of the elementary gates should be below a certain threshold constant~\cite{10.1145/258533.258579,10.1137/S0097539799359385,548464,Kitaev_1997,doi:10.1098/rspa.1998.0166,10.5555/2011665.2011666,10.1007/11786986_6,PhysRevA.71.012336,doi:10.1063/1.1499754,PhysRevLett.109.180502,gottesman2014faulttolerant,Yamasaki2022}.
The existence of a finite threshold for FTQC implies that the work cost of performing a unitary operation within a better fidelity than the threshold can also be bounded.
We can, therefore, take the maximum work cost per qubit over the gates in the universal gate set as
\begin{align}
\label{eq:max_elementary_cost}
    \mathcal W_\text{max,el.} \coloneqq E_\mathrm{qubit} + E_\text{ctrl.},
\end{align}
which splits into a contribution from the single-qubit energy change $E_\mathrm{qubit}$ and the remaining contribution $E_\text{ctrl.}$ representing the control cost per qubit.
The universal gate set includes multiqubit gates such as CNOT gates, and for a multiqubit gate, the per-qubit work cost in~\eqref{eq:max_elementary_cost} represents the total work cost divided by the number of qubits involved in the gate.
We will explain $E_\mathrm{qubit}$ and $E_\text{ctrl.}$ in the following.

In~\eqref{eq:max_elementary_cost}, $E_\mathrm{qubit}$ is the maximum single-(qu)bit energy over all the (qu)bits in the computer, where
\begin{equation}
\label{eq:E_qubit}
    E_\mathrm{qubit}\geq0.
\end{equation}
To be more specific, for the $w$th (qu)bit out of the $W$ (qu)bits ($w\in\{1,2,\ldots, W\}$), let $H_\mathrm{qubit}^{(w)}> 0$ be the single-(qu)bit Hamiltonian characterizing the energy of the $w$th (qu)bit, scaled appropriately so that its ground-state energy should be non-negative; in this case, we can take
\begin{equation}
    E_\mathrm{qubit}=\max_{w}\qty{\norm{H_\mathrm{qubit}^{(w)}}_\infty}, 
\end{equation}
where $\|\cdots\|_\infty$ is the operator norm.
To ensure the generality of our analysis, we will keep $E_\mathrm{qubit}$ as a general variable, which may depend on the circuit size and scale, at most, polynomially with $N$.
As for the control cost $E_\text{ctrl.}$ in~\eqref{eq:max_elementary_cost}, our analysis also keeps $E_\text{ctrl.}$ as a general variable, which may also grow, at most, polynomially with $N$.
To summarize, the assumption on $E_\mathrm{qubit}$ and $E_\mathrm{ctrl.}$ for our analysis is
\begin{align}
    E_\mathrm{qubit}&=O(\mathrm{poly}(N)),\label{eq:E_qubit_polyN}\\
    E_\mathrm{ctrl.}&=O(\mathrm{poly}(N)),\label{eq:E_ctrl_polyN}
\end{align}
while these may have much smaller, constant order $O(1)$ in conventional settings.
This general setting allows us to account for the possibility that as the problem size grows, qubit crosstalk in a large quantum device may lead to a higher control cost per qubit.
As shown in Ref.~\cite{Chiribella2021} (also see Appendix~\ref{appendix:model_details}), the variable $E_{\rm ctrl.}$ is conventionally related to $E_\mathrm{qubit}$ as
\begin{equation}
\label{eq:E_ctrl}
     E_\text{ctrl.}=O(E_\mathrm{qubit}).
\end{equation}

As a whole, we have a maximal work cost $\mathcal W_\text{max., el.}$ per qubit in~\eqref{eq:max_elementary_cost} to perform any elementary gate in the gate set, regardless of the initial state of the qubits before applying the gate.
This upper bound $\mathcal W_\text{max,el.}$ may be a loose worst-case estimate to cover the case where an elementary gate rotates the qubits it acts on from the ground state to the excited state with the energy change of $E_\mathrm{qubit}$.
Note that some operations, e.g., a gate transforming the excited state to the ground state, may allow us to extract energy; also, the overall energetic cost sums up to zero as in~\eqref{eq:proof_energy_conservation_2}.
Still, it is important to bound the worst-case contributions of the work cost in terms of $\mathcal{W}_\text{max., el.}$, rather than only bounding the contributions of $E_\mathrm{ctrl.}$ in~\eqref{eq:max_elementary_cost}, for implementability of every step of the computation.

\paragraph*{Step \ref{item:qstep_1}: Contributions from gates, input, and output.}
Using $\mathcal W_\text{max,el.}$ in~\eqref{eq:max_elementary_cost},  we estimate the energy consumption for each part of the circuit in Fig.~\ref{fig:generic_algo}.

To estimate the cost of the unitaries $U_k$, we use the fact that the unitary acts on $W$ qubits, and also each of these unitaries has depth $D_k$.
Thus, the cost of $\mathcal W_\text{gate}^{(U_k)}$ applying $U_k$ in~\eqref{eq:W_gate_U_k_def} is bounded by
\begin{align}
     \mathcal W_\text{gate}^{(U_k)}&\leq \mathcal W_\text{max., el.} \times W D_k\\
     &=O((E_\mathrm{qubit}+E_\mathrm{ctrl.})\times WD_k),\label{eq:W_gate_U_k}
\end{align}
where $W D_k$ is the volume of this part of the circuit, and the second line follows from~\eqref{eq:max_elementary_cost}.
Due to~\eqref{eq:W_gate_U_k_def}, $\mathcal W_\text{gate}^{(U_k)}$ in~\eqref{eq:W_gate_U_k} is decomposed into
\begin{align}
\label{eq:E_U_k}
    \Delta E^{(U_k)}&=O(E_\mathrm{qubit}\times WD_k),\\
\label{eq:E_ctrl_U_k}
    E_\text{ctrl.}^{(U_k)}&=O(E_\mathrm{ctrl.}\times WD_k).
\end{align}

In each red part of calling the oracle between $U_k$ and $U_{k+1}$ ($k=1,\ldots,M$) in Fig.~\ref{fig:generic_algo},
we account for the control cost in swapping the states between the computer and the input register for the quantum oracle back and forth as in Fig.~\ref{fig:oracle_QvsC}(b).
In this part, i.e., the $k$th query to the oracle with input size $N+W_f$ as in~\eqref{eq:oracle_Q},
we have $N+W_f$ auxiliary qubits for the input register in addition to the $W$ qubits in the computer; i.e., the width of this part of the circuit in Fig.~\ref{fig:generic_algo} is $W+N+W_f$.
Therefore, the control cost $W_\text{gate}^{(\text{in},k)}$ of the $k$th-query part for the input in~\eqref{eq:contro_cost_in_k} is given by
\begin{align}
    \mathcal W_\text{gate}^{(\text{in},k)} &\leq E_\mathrm{ctrl.}\times (W+N+W_f) D_\text{in}\\
    &=O(E_\mathrm{ctrl.}\times WD_\text{in}), \label{eq:W_gate_in_k}
\end{align}
where $D_\text{in}$ is given by~\eqref{eq:D_swap}, and we use~\eqref{eq:W_bound} in the last line.
Also, using $E_\mathrm{qubit}$ in~\eqref{eq:E_qubit}, we bound the energetic cost for the $k$th query by
\begin{equation}
\label{eq:E_in_k}
    \Delta E^{(\mathrm{in},k)}\leq E_\mathrm{qubit}(N+W_f) =O(E_\mathrm{qubit} \times W),
\end{equation}
where we use~\eqref{eq:W_bound}.

In the same way, in the green part of Fig.~\ref{fig:generic_algo}, the control cost $\mathcal W_\text{gate}^{(\text{out})}$ of outputting the result of the computation in~\eqref{eq:W_gate_out} is at most
\begin{align}
    \mathcal{W}_\text{gate}^{(\text{out})} &\leq E_\mathrm{ctrl.}\times (W+1)D_\text{out}\\
    &= O(E_\mathrm{ctrl.}\times WD_\text{out}), \label{eq:W_gate_out_k}
\end{align}%
where $D_\mathrm{out}$ is given by~\eqref{eq:D_out}, the factor $W+1$ is composed of the $W$ qubits in the computer and another single qubit for the output register.
Moreover, using the upper bound $E_\mathrm{qubit}$ of the single-qubit energy in~\eqref{eq:E_qubit}, we bound the energetic cost of the output by
\begin{align}
\label{eq:E_out}
    |\Delta E^{(\text{out})}|&\leq E_\mathrm{qubit}.
\end{align}

\begin{widetext}
As a whole, due to~\eqref{eq:W_gate_U_k},~\eqref{eq:W_gate_in_k},~\eqref{eq:E_in_k},~\eqref{eq:W_gate_out_k}, and~\eqref{eq:E_out}, the work cost and the energy transfer between the computer and the oracle in the part of the circuit in Fig.~\ref{fig:generic_algo} for the computation (i.e., the parts except for $\mathcal{E}$ in Fig.~\ref{fig:generic_algo}) is at most
\begin{align}
    &\sum_{k=1}^{M+1}\mathcal W_\text{gate}^{(U_k)}+\sum_{k=1}^M \mathcal W_\text{gate}^{(\text{in},k)}+\mathcal{W}_\text{gate}^{(\text{out})}+\sum_{k=1}^{M}\Delta E^{(\mathrm{in},k)}-\Delta E^{(\mathrm{out},k)}\nonumber\\
    &=O\qty((E_\mathrm{qubit}+E_\mathrm{ctrl.})\times W\qty(\sum_{k=1}^{M+1}D_k+MD_\mathrm{in}+D_\mathrm{out}))\\
    &=O(\mathcal{W}_\text{max,el.}\times WD_\mathrm{comp}),
\end{align}
where $D_\mathrm{comp}$ is given by~\eqref{eq:D_U}, and $\mathcal{W}_\text{max,el.}$ by~\eqref{eq:max_elementary_cost}.
In words, the work cost scales at most with the \textit{volume} of the circuit for implementing the algorithm, i.e., the product of total width and depth of the circuit for implementing the algorithm (complexity-theoretic contribution) times the parameter $\mathcal W_\text{max., el.}$ in~\eqref{eq:max_elementary_cost} representing the maximum work cost per qubit of gates (energetic contribution).
Also, as shown in~\eqref{eq:proof_energy_conservation_2}, 
the relevant contribution to the energy consumption $\mathcal{W}$ is the control cost, which is a part of the work cost and bounded, due to~\eqref{eq:E_ctrl_U_k},~\eqref{eq:W_gate_in_k}, and~\eqref{eq:W_gate_out_k}, by
\begin{align}
    \sum_{k=1}^{M+1}E_\text{ctrl.}^{(U_k)}+\sum_{k=1}^M  E_\text{ctrl.}^{(\text{in},k)}+E_\text{ctrl.}^{(\text{out})} &=O\qty(E_\mathrm{ctrl.}\times W\qty(\sum_{k=1}^{M+1}D_k+MD_\mathrm{in}+D_\mathrm{out}))\\
    \label{eq:work_cost_total}
    &=O(E_\mathrm{ctrl.}\times WD_\mathrm{comp}).
\end{align}
\end{widetext}

\paragraph*{Step \ref{item:qstep_2}: Contribution from reinitialization by finite Laudauer erasure.}
To close the thermodynamic cycle on the computer's memory, the qubits in the computer are reinitialized by the end of the computation as shown by the part $\mathcal{E}$ of the circuit in Fig.~\ref{fig:generic_algo}.
The research field of algorithmic cooling~\cite{Park2016,Alhambra2019,Serafini2020,Soldati2021} studies algorithms that cool down a mixed state of qubits in a quantum computer into a pure state to reset the qubits' state, but these results do not straightforwardly apply to our framework due to the difference in settings from ours.
Protocols for Landauer erasure~\cite{Landauer1961} can be used for the reinitialization in our framework. 
Existing results on Laudauer erasure use either a sequence of unitary operations as in Refs.~\cite{Reeb2014,Taranto2021} or finite-time relaxation~\cite{Proesmans2020,VanVu2022,Rolandi2022}, but these works mostly provide lower bounds of the heat dissipation.
Reference~\cite{Reeb2014} also provides an upper bound of heat dissipation, but the protocol to achieve this upper bound requires an infinite number of steps.
By contrast, what we need is an achievable upper bound of the energy consumption for erasing a given mixed state to a pure state up to finite precision $\varepsilon$ within a finite number of steps.
To obtain such a bound, we work on protocols using a sequence of unitary operations as in Ref.~\cite{Reeb2014} (also called a collision model~\cite{Taranto2024}), which are more suitable to obtain a finite bound of steps than those using finite-time relaxation.

Our goal here is to find an upper bound of the achievable cost $\mathcal W_\mathrm{gate}^{\mathcal E}$ in this finite setting of reinitialization for our framework.
We note that in Sec.~\ref{sec:classical_lowerbound}, we will also clarify a lower bound of the required energy for this initialization by a Landauer-erasure protocol, which will be given by the required amount of heat $\mathcal Q_E$ dissipated into the environment in the initialization (i.e., the initialization cost as defined in Sec.~\ref{sec:energetic_costs}). 
By contrast, we here bound $\mathcal{W}_\mathrm{gate}^{\mathcal E}$ from above, not from below, especially in the setting where we take into account all possible contributions including the control cost as well as the energetic cost.
In principle, a part of the achievable cost of the Landauer-erasure protocol may be made as close to its lower bound $\mathcal Q_E$ as possible~\cite{Reeb2014}, which is in part true, but only if we focus on the energetic cost and ignore the control cost.
Problematically, exactly achieving $\mathcal Q_E$ may require infinite time steps, and the control cost required for these infinite time steps may also diverge; indeed, based on the third law of thermodynamics in the formulation of Nernst's unattainability principle~\cite{Nernst1906}, it is generally agreed upon that cooling a (quantum) state to absolute zero and thereby erasing its previous state inevitably comes at divergent resource costs in some form~\cite{Wilming2017,Masanes2017,Freitas2018,Taranto2021}.

To obtain the upper bound including the control cost in addition to the energetic cost, we need a finite upper bound on the number of steps (depths) for achieving finite infidelity $\varepsilon$ in the Laundauer-erasure protocol, which is where Theorem~\ref{thm:ThmErasureBound} below steps in. 
The setting of the Landauer-erasure protocol starts with the initial state $\rho_S$ of the target system $S$ of interest and an environment available in a thermal state at inverse temperature $\beta>0$ 
\begin{equation}
\label{eq:tau_E_beta}
    \tau_E[\beta]\coloneqq\frac{e^{-\beta H_E}}{\tr[e^{-\beta H_E}]},
\end{equation}
where $H_E$ is the Hamiltonian of the environment.
Note that $S$ will be each qubit in the computer in our analysis, but we here present a general result on a $d$-dimensional target system.
The goal of the Landauer-erasure protocol is to transform the state
\begin{equation}
    \rho_S\mapsto \rho'_S = \tr_E[\rho^\prime_{SE}]= \tr_E[U(\rho_S\otimes\tau_E[\beta]) U^\dagger]
\end{equation}
by some unitary $U$,
so that it should hold that
\begin{equation}
\label{eq:fidelity_condition_erasure}
    \braket{0}{\rho'_S|0}\geq 1-\varepsilon,
\end{equation}
where $\ket{0}$ can be any pure state of the target system in the same way as the requirement~\eqref{eq:fidelity_condition} in our framework.
In the Landauer-erasure protocol, energy is dissipated into the environment $E$ as heat
\begin{equation}
\label{eq:heat}
    \mathcal Q_E\coloneqq\tr[H_E(\rho'_E-\rho_E)].
\end{equation}
where we write
\begin{equation}
     \rho^\prime_E=\tr_S[\rho^\prime_{SE}].
\end{equation}
One can always write this increase in the environment's energy as~\cite{Reeb2014}
\begin{equation}
\label{eq:reeb_wolf}
    \beta\mathcal Q_E = \Delta S + D(E'\|E)+I(E':S'),
\end{equation}
where 
$E, E'$ are the environment before and after the Landauer erasure, $S'$ is the target system after the erasure, 
\begin{equation}
\label{eq:Landauer_bound}
    \Delta S \coloneqq S(\rho_S)-S(\rho'_S)
\end{equation}
is called Landauer's bound~\cite{Landauer1961},
\begin{equation}
  S(\rho)\coloneqq-\tr[\rho\ln\rho]  
\end{equation}
is the von Neumann entropy, 
\begin{equation}
  D(E'\|E)\coloneqq\tr[\rho^\prime_{E}(\ln\rho^\prime_{E}-\ln\tau_{E}[\beta])]  
\end{equation}
is the quantum relative entropy, and
\begin{equation}
    I(E^\prime: S^\prime)\coloneqq S(\rho^\prime_E)+S(\rho^\prime_S)-S(\rho^\prime_{SE})
\end{equation}
is the quantum mutual information.
Using~\eqref{eq:reeb_wolf} together with $I(E^\prime: S^\prime)\geq 0$, we can bound the relative entropy for the states of the environment after the erasure versus before as
\begin{equation}
\label{eq:D_Q_bound}
    D(E'\|E)\leq \beta\mathcal Q_E -\Delta S.
\end{equation}
Due to $D(E'\|E)\geq 0$, it always holds that
\begin{equation}
    \mathcal{Q}_E\geq\frac{\Delta S}{\beta}.
\end{equation}
The quantity $\beta\mathcal Q_E -\Delta S\geq 0$ on the right-hand side of~\eqref{eq:D_Q_bound} represents the excess of the protocol's heat $\mathcal Q_E$ dissipated to the environment over that given by the Laudauer's bound $\Delta S$,
which is also an upper bound of how different the environment's initial state $\tau_E[\beta]$ is from its final state $\rho_E'$ after the erasure; as the condition on the environment to be specified by the constant $\eta$ in~\eqref{eq:eta_condition}, we require
\begin{equation}
\label{eq:D_Q_bound_2}
    \beta\mathcal Q_E -\Delta S\leq \eta.
\end{equation}
With these notations, we obtain the following finite bounds.

\begin{restatable}[Finite Landauer-erasure bound]{theorem}{ThmErasureBound}\label{thm:ThmErasureBound}
Given a quantum state $\rho_S$ of a  $d$-dimensional target system for any $d\geq 2$ and any constants $\varepsilon\in(0,1/2]$ and $\eta\in(0,1]$,
let $T$ be
\begin{align}
\label{eq:T_def_general}
    T = \left\lceil \frac{(e+1)}{e\eta}\ln\qty(\frac{(e+1)(d-1)^2}{\varepsilon\eta})\right\rceil=O\qty(\frac{1}{\eta}\log\frac{d^2}{\varepsilon\eta}),
\end{align}
where $\lceil{}\cdots{}\rceil$ is the ceiling function, and $e$ is the Euler constant.
Then, a $T$-step Landauer-erasure protocol can transform $\rho_S$ into a final state $\rho'_S$ with infidelity to a pure state $\ket{0}$ satisfying $\braket{0|\rho'_S}{0}\geq 1-\varepsilon$ in~\eqref{eq:fidelity_condition_erasure}
by using an environment at inverse temperature $\beta>0$, in such a way that the heat $\mathcal Q_E$ in~\eqref{eq:heat} dissipated into the environment differs from Laundauer's bound $\Delta S$ in~\eqref{eq:Landauer_bound} at most by
\begin{align}\label{eq:finite_erasure_bound}
    \beta\mathcal Q_E - \Delta S \leq \eta,
\end{align}
as required in~\eqref{eq:D_Q_bound_2}.
\end{restatable}
\begin{proof}
    The proof can be found in Appendix~\ref{appendix:proof_cooling}.
\end{proof}

For our analysis of the framework in Fig.~\ref{fig:generic_algo}, this result on the finite erasure is now applied qubit-wise, i.e., $d=2$ in Theorem~\ref{thm:ThmErasureBound}, to the computer's state after the agent has obtained output from the computer.
In particular, for each qubit of the computer, we run the erasure protocol in Theorem~\ref{thm:ThmErasureBound} in total $W$ times in parallel to erase all the $W$ qubits, using $W$ sets of $T$ auxiliary qubits of the environment $E$ for an appropriate choice of $H_E$ as shown in~\eqref{eq:H_E} of Appendix~\ref{appendix:proof_cooling}, where each of the $W$ sets is used for erasing one of the $W$ qubits in the computer.
In this case, as described in Appendix~\ref{appendix:proof_cooling}, the erasure protocol is composed of $T$ steps (labeled $t\in\{1,\ldots, T\}$) of swap gates, where the $t$th swap gate is applied between the qubit to be erased in the computer and the $t$th auxiliary qubit in the set of $T$ auxiliary qubits of the environment to erase this qubit of the computer.
Note that applying the qubit-wise erasure to each of the $W$ qubits in total $W$ times may be more costly as a whole than erasing the state of all $W$ qubits in the computer simultaneously~\cite{Buffoni2023,Rolandi2023a} but provides a simpler upper bound of the work cost; remarkably, it suffices to use this potentially suboptimal Landauer-erasure protocol to prove the exponential quantum advantage in Sec.~\ref{sec:example}.
On the other hand, for the analysis of the lower bound of the work cost in Sec.~\ref{sec:classical_lowerbound}, we will analyze the efficient erasure protocol applied to all qubits simultaneously at the end of the computation to obtain a general lower bound.
We also remark that since the analysis here deals with the qubit-wise erasure, we indeed do not have to apply the erasure collectively at the end of the computation but can also initialize each qubit in the middle of the computation as long as the initialized qubit is no longer used for the rest of the computation.
Later in our analysis of step~\ref{item:qstep_3}, we will make a minor modification to the protocol in Fig.~\ref{fig:generic_algo}, so that we may initialize each qubit in the middle for quantum error correction rather than erasing all the qubits at the end.
The following analysis of the upper bound of the work cost for the qubit-wise erasure is still applicable with this modification.

For the erasure protocol in Theorem~\ref{thm:ThmErasureBound} applied to our framework, the constant $\varepsilon>0$ is to be determined by the threshold for FTQC, and $\eta>0$ by the amount of thermal fluctuation of the environment.
In particular, FTQC can be performed if every qubit is reinitialized to a fixed pure state $\ket{0}$ within constant infidelity $\varepsilon$ satisfying the threshold theorem, i.e.,~\eqref{eq:fidelity_condition}.
Moreover, we demand that the environment changes little in this process of reinitializing each qubit in the computer compared to a constant amount of inherent thermal fluctuation of the environment.
To state this littleness quantitatively, we require for each single-qubit erasure that
\begin{equation}
\label{eq:Q_condition}
    \beta \mathcal Q_E^{(w)} -\Delta S^{(w)} \leq \eta,
\end{equation}
where $\mathcal{Q}_E^{(w)}$ and $\Delta S^{(w)}$ are $\mathcal{Q}_E$ and $\Delta S$ in Theorem~\ref{thm:ThmErasureBound} applied to erase the $w$th qubit in the computer ($1\leq w\leq W$).
Due to~\eqref{eq:D_Q_bound}, this inequality implies that, for each qubit of the computer that is erased separately, the quantum relative entropy between the states of the corresponding environment $E'^{(w)}$ after the erasure protocol and $E^{(w)}$ before (in the thermal state $\tau_E^{(w)}[\beta]$ as shown in Fig.~\ref{fig:generic_algo}) is bounded by
\begin{equation}
\label{eq:environment_condition}
    D(E'^{(w)}\|E^{(w)}) \leq \eta
\end{equation}
for some small constant $\eta>0$.
The choice of $\eta$ is independent of that of $\varepsilon$.

For each qubit, both requirements for $\varepsilon$ and $\eta$, i.e.,~\eqref{eq:fidelity_condition} and~\eqref{eq:Q_condition}, can be satisfied simultaneously if we choose the number of steps $T$ large enough, but how large $T$ should be was unclear in the existing asymptotic results~\cite{Reeb2014,Proesmans2020,Taranto2021,VanVu2022,Rolandi2022}.
By contrast, Theorem~\ref{thm:ThmErasureBound} shows in more detail that if we choose (for fixed dimension $d=2$)
\begin{equation}
\label{eq:T}
    T=O\qty(\frac{1}{\eta} \log\frac{1}{\varepsilon\eta}),
\end{equation}
we can satisfy the above requirements at the same time for any $\varepsilon\in(0,1/2]$ and $\eta\in(0,1]$.
Using a gate set that can implement a swap gate within a constant depth, the depth of the erasure part $\mathcal{E}$ in Fig.~\ref{fig:generic_algo} is bounded by
\begin{equation}
\label{eq:D_E}
    D_\mathcal{E}\leq O(T).
\end{equation}
Also, due to~\eqref{eq:Q_condition}, if we sum together the heat dissipation for all $W$ qubits, the protocol achieves
\begin{equation}
\label{eq:Q_E}
    \mathcal Q_E\leq \sum_{w=1}^{W}\frac{\Delta S^{(w)}+\eta}{\beta}\leq \frac{W}{\beta}\left(\ln 2 + \eta\right),
\end{equation}
where the last inequality follows from the upper-bound estimate $\Delta S^{(w)}\leq \ln2$ that holds for each qubit.

Apart from the energetic cost $\mathcal Q_E$ in the environment, 
for each of the $W$ qubits in the computer, the energetic cost per qubit per single-depth part of the circuit for the erasure is bounded by the maximum energetic cost $E_\mathrm{qubit}$ in~\eqref{eq:E_qubit}.
Thus, similar to step~\ref{item:qstep_1}, the energy change in the computer's memory due to the erasure is upper bounded by
\begin{equation}
\label{eq:delta_E_E}
    \Delta E^{(\mathcal{E})}\leq E_\mathrm{qubit}\times W D_\mathcal{E}.
\end{equation}

As for the control cost,
we can divide the control cost of erasure into a product of the maximum control cost per single-depth part of the circuit for the erasure and the depth $D_\mathcal{E}$ of the circuit in~\eqref{eq:D_E}, in the same way as step~\ref{item:qstep_1}.
To bound the maximum control cost per single-depth part of the erasure,
as in~\eqref{eq:tau_E_beta},
let $H_E$ denote the overall Hamiltonian $H_E$ of the environment used for the single-qubit erasure protocol, which decomposes into the sum over the single-qubit Hamiltonians $H_E^{(t)}$ of the $t$th auxiliary qubit of the environment used at the $t$th time step in the protocol ($1\leq t\leq T$).
The control cost per single-depth part of the circuit depends on the energy scale of the qubit to be controlled, i.e., $E_\mathrm{qubit}$ for a qubit in the computer and $\|H_E\|_\infty$ for the environment to reinitialize the qubit by the Landauer erasure (note that $H_E\geq 0$ holds for the Landauer-erasure protocol here, as shown in~\eqref{eq:H_E} of Appendix~\ref{appendix:proof_cooling}). 
To avoid fixing the specifics of the control system, we here recall the conventional scaling of the control cost $E_\text{ctrl.}$ of the qubit in the computer in~\eqref{eq:E_ctrl}; following this conventional scaling,
we here assume that the control cost in the environment per single-depth part of the circuit depends linearly on the energy scale of the environment Hamiltonian, i.e.,
\begin{equation}
    E_\text{ctrl.}^{(E)} = O(\|H_E\|_\infty),
\end{equation}
as shown for the conventional models in, e.g., Ref.~\cite{Chiribella2021} (see also Appendix~\ref{appendix:model_details}).
We calculate in Proposition~\ref{prop:MaxErasureEnergy} (Appendix~\ref{appendix:proof_cooling}) how the largest energy eigenvalue $\|H_E\|_\infty$ of the environment Hamiltonian scales as a function of $T\rightarrow \infty$ and $\varepsilon\rightarrow 0$, for fixed inverse temperature $\beta$.
This calculation allows us to estimate the per-depth control cost of all the $T$ auxiliary qubits in the environment for erasing each of the $W$ qubits in the computer (for fixed dimension $d=2$) by
\begin{equation}
\label{eq:environment_control_cost}
    E_\text{ctrl.}^{(E)}=O\qty(\frac{T}{\beta}\log\frac{1}{\varepsilon}).
\end{equation}
Consequently, the control cost of erasure for the $W$ qubits in the computer and the corresponding $W$ sets of the auxiliary qubits in the environments for the Laudauer erasure protocol is bounded by
\begin{align}
\label{eq:E_ctrl_E}
    E_\text{ctrl.}^{(\mathcal{E})} \leq \left(E_\text{ctrl.}+E_\text{ctrl.}^{(E)}\right)\times WD_\mathcal{E}.
\end{align}
If $E_\mathrm{ctrl.}$ is independent of $N$, then $E_\mathrm{ctrl.}$ may be negligible compared to the control cost $E_\text{ctrl.}^{(E)}$ of the environment, in the limit of $\eta,\varepsilon\rightarrow 0$ due to~\eqref{eq:T} and~\eqref{eq:environment_control_cost};
however, in the general case where the control cost $E_\mathrm{ctrl.}$ of the computer may grow in the problem size $N$ and is not negligible, we keep both $E_\mathrm{ctrl.}$ and $E_\text{ctrl.}^{(E)}$ separately in~\eqref{eq:E_ctrl_E}.

\begin{widetext}
As a whole, by combining~\eqref{eq:T} -- \eqref{eq:E_ctrl_E} together to bound the work cost for the erasure, we obtain
\begin{align}
\label{eq:W_gate_E}
    \mathcal W_\mathrm{gate}^{(\mathcal{E})}&= \Delta E^{(\mathcal{E})}+E_\text{ctrl.}^{(\mathcal{E})}+\mathcal Q_E \\
    &\leq\qty(E_\mathrm{qubit}+E_\mathrm{ctrl.}+E_\mathrm{ctrl.}^{(E)})\times WD_\mathcal{E}+\frac{W}{\beta}\qty(\ln 2+\eta)\\
        &= O\left(\left(E_\mathrm{qubit}+ E_\mathrm{ctrl.} + \frac{1}{\beta \eta}\right)\times\frac{W}{\eta}\mathrm{polylog}\frac{1}{\varepsilon\eta}\right)\\
        &=O\left(\left(\mathcal{W}_\mathrm{max,el.} + \frac{1}{\beta \eta}\right)\times\frac{W}{\eta}\mathrm{polylog}\frac{1}{\varepsilon\eta}\right),
\end{align}
where $\mathcal{W}_\mathrm{max,el.}$ is given by~\eqref{eq:max_elementary_cost}.
Also, as shown in~\eqref{eq:proof_energy_conservation_2}, 
the relevant contribution to the energy consumption $\mathcal{W}$ is the control cost, which is bounded, due to~\eqref{eq:T},~\eqref{eq:Q_E},~\eqref{eq:environment_control_cost}, and~\eqref{eq:E_ctrl_E}, by
\begin{align}
    E_\text{ctrl.}^{(\mathcal{E})}+\mathcal Q_E
        &\leq\qty(E_\mathrm{ctrl.}+E_\mathrm{ctrl.}^{(E)})\times WD_\mathcal{E}+\frac{W}{\beta}\qty(\ln 2+\eta)\\
\label{eq:W_gate_E_bound}
        &= O\left(\left(E_\mathrm{ctrl.} + \frac{1}{\beta \eta}\right)\times\frac{W}{\eta}\mathrm{polylog}\frac{1}{\varepsilon\eta}\right).
\end{align}

To summarize the argument on the achievable upper bound of the energy consumption of the quantum computation in Fig.~\ref{fig:generic_algo}, in the case without gate errors, the energy consumption can be given by Theorem~\ref{thm:ThmQuantumIdeal} shown in the following.
The derivation is based on the expressions in~\eqref{eq:work_cost_total} and~\eqref{eq:W_gate_E_bound};
in particular, we use~\eqref{eq:proof_energy_conservation_2} to make all the energetic costs (i.e., the terms involving $E_\mathrm{qubit}$ in~\eqref{eq:E_qubit}) cancel out, and thus, only the control and initialization costs remain.

\begin{restatable}[An upper bound of energy consumption in ideal case without gate errors]{theorem}{ThmQuantumIdeal}
\label{thm:ThmQuantumIdeal}    
Given constants $E_\mathrm{ctrl.}$ in~\eqref{eq:E_ctrl} and $\beta$ in~\eqref{eq:beta},
working in an idealized case without gate errors, computation in the framework of Fig.~\ref{fig:generic_algo} can be performed with an energy consumption $\mathcal W$ (Definition~\ref{def:DefEnergyConsumption}) bounded,
as the problem size $N$ goes to infinity, and the infidelities $\varepsilon,\eta>0$ vanish,
by
\begin{align}\label{eq:idealized_quantum_bound}
    \mathcal W\leq \mathcal W^{(\mathrm{Q})} &=
    O\left(E_\mathrm{ctrl.}\times W  D_\mathrm{comp}  +\left(E_\mathrm{ctrl.} + \frac{1}{\beta \eta}\right)\times\frac{W}{\eta}\mathrm{polylog}\qty(\frac{1}{\varepsilon\eta})\right),
\end{align}
where the width $W$ of the circuit in~\eqref{eq:W_bound} and the depth $D_\mathrm{comp}$ given by~\eqref{eq:D_U} are the parameters depending on the problem size $N$.
\end{restatable}
\end{widetext}
\begin{proof} The proof can be found in Appendix~\ref{appendix:detailed_energy_consumption}.
\end{proof}

\paragraph*{Step \ref{item:qstep_3}: Contribution from quantum error correction.}
The discussion so far has not taken into account imperfections in carrying out the unitary operations $U_k$ and swap gates for input and output, which are necessary for realizing quantum computation.
Any unitary gate carried out on an actual physical platform will inherently be imperfect, be it through unavoidable noise~\cite{G}, non-ideal control parameters~\cite{Jiang2022}, or timing errors~\cite{Ball2016,Xuereb2023}.
Quantum error correction is a way to suppress the effect of these errors given that the rate at which these errors occur is below a certain threshold value, a result that is known as the threshold theorem~\cite{10.1145/258533.258579,10.1137/S0097539799359385,548464,Kitaev_1997,doi:10.1098/rspa.1998.0166,10.5555/2011665.2011666,10.1007/11786986_6,PhysRevA.71.012336,doi:10.1063/1.1499754,PhysRevLett.109.180502,gottesman2014faulttolerant,Yamasaki2022}.
The main principle of quantum error correction is to use a quantum error-correcting code of multiple physical qubits to redundantly represent a logical qubit on which the unitary operations are carried out.
In this way, we, in principle, have protocols for FTQC\@.
However, such fault-tolerant protocols to simulate the given original circuit by running the corresponding fault-tolerant circuit may incur space and time overheads, i.e., the overheads in terms of the number of qubits and the circuit depth, respectively, of the fault-tolerant circuit compared to the original circuit.

In the computational model herein presented, the space is quantified through the number of qubits $W$ in~\eqref{eq:W_bound}, and the time through the circuit depth $D$ in~\eqref{eq:D_bound}, as shown in Fig.~\ref{fig:generic_algo}.
With a fault-tolerant protocol using measurements and classical post-processing, one could convert the original circuit in Fig.~\ref{fig:generic_algo} into a fault-tolerant circuit with only polylogarithmic overheads in width and depth; in particular, to simulate a given $W$-width $D$-depth circuit, the fault-tolerant circuit can have~\cite{G}
\begin{align}
    W_\text{FT} &= W \times \mathrm{polylog} (WD), \label{eq:polylog_w_overhead} \\
    D_\text{FT} &= D \times \mathrm{polylog} (WD), \label{eq:polylog_d_overhead}
\end{align}
where $W_\text{FT}$ and $D_\text{FT}$ are the width and depth of the fault-tolerant circuit, respectively.
This choice is not unique but depends on the construction of the fault-tolerant protocol to achieve FTQC\@.
For example, advances have been made to achieve a constant space overhead $W_\text{FT}=W\times O(1)$ and a quasi-polylogarithmic time overhead $D_\text{FT} = D \times \exp(\mathrm{polylog}(\log(WD)))$~\cite{Yamasaki2022};
more recently, Ref.~\cite{tamiya2024polylogtimeconstantspaceoverheadfaulttolerantquantum} has proven that a fault-tolerant protocol can achieve a constant space overhead $W_\text{FT}=W\times O(1)$ and a polylogarithmic time overhead $D_\text{FT} = D \times \mathrm{polylog}(WD)$.

To analyze the upper bound of the achievable quantum cost $\mathcal W^{(\text{Q})}$ of energy consumption with the fault-tolerant protocol in~\eqref{eq:polylog_w_overhead} and \eqref{eq:polylog_d_overhead}, one has to modify the analysis in steps~\ref{item:qstep_1} and~\ref{item:qstep_2}, so as to allow initializing qubits in the middle of the computation rather than at the end and to include costs of measurements and classical post-processing present in the fault-tolerant protocol.
The qubit initialization in the middle of computation is indispensable for performing quantum error correction to achieve FTQC; in the presence of noise, if the initialization were allowed only at the beginning or at the end, the class of problems that noisy quantum computation can solve would be significantly limited~\cite{aharonov1996limitations}.
Regarding the initialization, as discussed in step~\ref{item:qstep_2}, the same upper bound on the work cost of each qubit-wise erasure remains to hold even if we initialize each qubit in the middle rather than at the end.
As for the measurement, an ideal projective measurement would come at a divergent resource cost as explored in Ref.~\cite{Guryanova2020}, but for the quantum error correction, measurements with a finite fidelity above a certain constant threshold are sufficient.
The costs of these finite-fidelity measurements may still be non-negligible, and their exact value is still subject to research~\cite{Deffner2016,Debarba2019,Stevens2022}; however, in principle, this cost can also be bounded by a constant that is independent of the algorithm's other parameters such as the problem size $N$, similar to the cost of each gate in~\eqref{eq:max_elementary_cost}.
The cost of classical post-processing can also be evaluated by writing the classical computation in terms of reversible classical logic circuits and bounding the cost of each classical logic gate by a constant, similar to~\eqref{eq:max_elementary_cost}.
In this way, the increase of the work cost arising from quantum error correction can be taken into account in principle by setting $\varepsilon$ in step~\ref{item:qstep_2} as a constant below the threshold and multiplying polylogarithmic factors to $W$, $D_k$, and $D_\mathcal{E}$ as shown in~\eqref{eq:polylog_w_overhead} and~\eqref{eq:polylog_d_overhead}.

On top of this, to solidify our analysis, we further employ the fact that measurements and classical post-processing are not fundamentally necessary for FTQC\@;
in particular, instead of performing the syndrome measurement and classical post-processing for decoding in FTQC, it is possible to implement the decoding by quantum circuits and perform the correction by the controlled Pauli gates in place of applying Pauli gates conditioned on the output of the classical post-processing~\cite{Aharonov1997,Aharonov2008}.
In the fault-tolerant protocol with measurements, each measured qubit can be reinitialized and reused in the rest of the computation, where a single qubit may be initialized multiple times in implementing the computation.
By contrast, we here allow initializing each qubit only once in the middle of the computation, so as to apply the same upper bound of the work cost as that obtained in step~\ref{item:qstep_2}; that is, the fault-tolerant protocol without measurements here never traces out each of these qubits, and furthermore, instead of reinitializing the same qubit multiple times, another auxiliary qubit is initialized by the finite erasure protocol in the middle of the computation.   
The fault-tolerant circuit with measurements has the depth of $D \times \mathrm{polylog} (WD)$, and thus, each of the $W \times \mathrm{polylog} (WD)$ qubits in this fault-tolerant circuit may have at most $D \times \mathrm{polylog} (WD)$ measurements.
To simulate these measurements and classical post-processing in the above way, it suffices to use at most $D \times \mathrm{polylog} (WD)$ auxiliary qubits for each of the $W \times \mathrm{polylog} (WD)$ qubits in the protocol without the measurements.
Consequently, sacrificing the efficiency in the circuit width $W$ (i.e., introducing a polynomial factor in $D$ to $W$), we implement the original circuit in Fig.~\ref{fig:generic_algo} by a fault-tolerant protocol without measurements, which achieves
\begin{align}
\label{eq:width_FTQC}
    W_\text{FT} &= W D\times\mathrm{polylog}(WD),\\
\label{eq:depth_FTQC}
    D_\text{FT} &= D \times \mathrm{polylog} (WD),
\end{align}
in simulating a $W$-width $D$-depth original circuit.
While whether we should completely avoid the measurements in implementing quantum computation may be of questionable practical relevance, these results show that the measurement and post-processing costs are not a fundamental restriction toward bounding the energy consumption of quantum computing.

As a whole, we conclude this section by summarizing all the costs in one theorem including those of quantum error correction.
Recall that the depth of the circuit in Fig.~\ref{fig:generic_algo} is, according to~\eqref{eq:D},
\begin{align}
    D&=D_\mathrm{comp}+D_\mathcal{E}\\
    &=O\left(D_\mathrm{comp}+\frac{1}{\eta}\mathrm{log}\left(\frac{1}{\varepsilon\eta}\right)\right)\\
\label{eq:D_bound}
    &=\widetilde{O}\left(D_\mathrm{comp} + \frac{1}{\eta}\right),
\end{align}
where the second line follows from~\eqref{eq:T} and~\eqref{eq:D_E}, and with the big-$\widetilde{O}$ notation, we may ignore polylogarithmic factors.
Thus, with $D$ in~\eqref{eq:D_bound}, adding $D$ auxiliary qubits for each of the $W$ qubits in Fig.~\ref{fig:generic_algo}, up to polylogarithmic factors, we can use the fault-tolerant protocol without measurements in~\eqref{eq:width_FTQC} and~\eqref{eq:depth_FTQC} to obtain the following theorem.

\begin{widetext}
\begin{restatable}[An upper bound of energy consumption including the cost of quantum error correction]{theorem}{ThmQuantum}
\label{thm:ThmQuantum}
Given constants $E_\mathrm{ctrl.}$ in~\eqref{eq:E_ctrl} and $\beta$ in~\eqref{eq:beta},
quantum computation in the framework of Fig.~\ref{fig:generic_algo} can be performed in a fault-tolerant way with an energy consumption $\mathcal W^{(\mathrm{Q})}$ (Definition~\ref{def:DefEnergyConsumption}) bounded,
as the problem size $N$ goes to infinity, and the infidelities $\varepsilon,\eta>0$ vanish,
by
\begin{align}
    \mathcal W\leq \mathcal W^{(\mathrm{Q})} &=  \widetilde{O}\qty(E_\mathrm{ctrl.}\times W\left(D_\mathrm{comp}+\frac{1}{\eta}\right)D_\mathrm{comp}+\left(E_\mathrm{ctrl.}+\frac{1}{\beta\eta}\right)\times \frac{W\left(D_\mathrm{comp}+\frac{1}{\eta}\right)}{\eta}),
    \label{eq:quantum_bound_erasure}
\end{align}
where the width $W$ of the circuit in~\eqref{eq:W_bound} and the depth $D_\mathrm{comp}$ given by~\eqref{eq:D_U} are the parameters depending on the problem size $N$, and with $\widetilde{O}$, we may ignore polylogarithmic factors in $W$, $D$, $1/\varepsilon$, and $1/\eta$.
\end{restatable}
\end{widetext}

\subsection{\label{sec:classical_lowerbound}Lower bound on energy consumption of classical computation}
In this section, we derive a general lower bound
\begin{equation}
\mathcal W\geq \mathcal W^{(\text{C})}    
\end{equation}
of energy consumption required for classical computation in the framework of Fig.~\ref{fig:generic_algo}.
The techniques for finding an achievable upper bound in Sec.~\ref{sec:quantum_upperbound} are not applicable to finding a fundamental lower bound of energy consumption.
After all, in the limit of energy-efficient implementation, a single elementary gate may be performed at as close to zero control cost as possible, as we have pointed out in Sec.~\ref{sec:termodynamic_model_of_computation}.
Rather, the essential technique here for deriving the nonzero fundamental lower bound is to use the non-invertibility of the classical oracle.
Due to the non-invertibility of the classical oracle, to reset the computer's memory at the end of the computation, we need to perform the Landauer erasure of information obtained from the queries to the oracle, which inevitably consumes energy.

Using energy conservation in~\eqref{eq:energy_conservation} together with the heat dissipation in the Landauer erasure, we obtain a general lower bound on the energy consumption of classical computation in the framework of Fig.~\ref{fig:generic_algo}, as shown by Theorem~\ref{thm:ThmClassical} below.
In this theorem, we identify an entropic quantity $\frac{1}{\beta}S(C|a)$ as a lower bound of energy consumption, which falls into the category of the initialization cost that is derived in an information-theoretic way and is thus implementation-independent.\footnote{
In our analysis of the lower bound of energy consumption of classical computation, we do not only analyze the scaling but provide the lower bound $\frac{1}{\beta}S(C|a)$ explicitly, i.e., including the constant factors, which is important for explicitly providing the goals of the experimental demonstrations shown by Table~\ref{table:1} of Sec.~\ref{sec:experiment}.
The agent's knowledge of the $1$-bit information $a$ of the result of the computation may be asymptotically negligible, but our analysis explicitly takes $a$ into account to provide the lower bound including the constant factors.
}
Even more notably, since we work on the query-complexity setting with the classical oracle, we can employ the techniques for deriving the lower bound of query complexity to prove rigorously that the required energy consumption for any classical algorithm to solve the problem can also be strictly gapped away from zero,  even without complexity-theoretical assumptions on the hardness of the problems, such as the hardness of integer factoring for classical computation.
Yet we remark that, as we will see in Corollary~\ref{cor:CorClassical} of Sec.~\ref{sec:classical_lowerbound}, the query complexity and the entropic lower bound of the energy consumption do not necessarily coincide in general up to constant factors. 

\begin{restatable}[A lower bound of energy consumption]{theorem}{ThmClassical}
\label{thm:ThmClassical}
If the classical computer in Fig.~\ref{fig:generic_algo} is allowed to interact with an environment at inverse temperature $\beta>0$ for memory reinitialization,
any computation in Fig.~\ref{fig:generic_algo} requires an average energy consumption of at least
\begin{align}
    \label{eq:work_cost_classical}
    \mathcal W \geq \mathcal W^{(\mathrm{C})} = \frac{S(C|a)}{\beta} ,
\end{align}
in the limit of closing the thermodynamic cycle $\varepsilon\rightarrow 0$,
where $S(C|a)\coloneqq\sum_{a=0,1}p(a)S(\rho_a^C)$ is the conditional entropy of the classical-quantum state for the computer's state $\rho_a^C$ before the erasure $\mathcal{E}$ conditioned on the measurement outcome $a\in\{0,1\}$ of the output register, $p(a)$ is the probability distribution of the measurement outcome, and $S(\rho)\coloneqq-\tr[\rho\ln\rho]$. 
\end{restatable}

To prove the lower bound in Theorem~\ref{thm:ThmClassical}, we have already introduced most of the essential techniques in formulating the framework of computation in Sec.~\ref{sec:setting}.
What remains for completing the proof is to assemble these techniques, which yields the following concise proof.

\begin{proof}[Proof of Theorem~\ref{thm:ThmClassical}]
We prove Theorem~\ref{thm:ThmClassical} by starting with some general energetic considerations along the lines of the thermodynamic diagram in Fig.~\ref{fig:thermo_diagram}.
As introduced in Definition~\ref{def:DefEnergyConsumption},
$\mathcal{W}$ is the energy consumption of the computer including reinitialization at the end.
Together with the heat dissipation into the environment $\mathcal Q$, we have energy conservation in~\eqref{eq:energy_conservation}, i.e.,
\begin{align}
    \mathcal W =\mathcal Q=\mathcal Q_E +\mathcal Q_{E'}, 
\end{align}
where the second equality is given by~\eqref{eq:Q}.
Recall that the heat dissipation $Q_{E'}$ due to control costs is by definition non-negative as shown in~\eqref{eq:Q_E_prime}, and thus, the initialization cost $\mathcal{Q}_E$ provides a lower bound
\begin{equation}
\label{eq:energy_conservation_inequality}
\mathcal W\geq\mathcal{Q}_E.
\end{equation}
As shown in~\eqref{eq:reeb_wolf}, Landauer's principle~\cite{Landauer1961,Reeb2014,rio2011thermodynamic} guarantees that, for reinitialization of the computer's memory using a heat bath at inverse temperature $\beta>0$, the heat dissipation $\mathcal Q_E$ into the environment is bounded by
\begin{equation}
\beta\mathcal{Q}_E= \Delta S+D(E'||E)+I(E':S')\geq \Delta S,
\end{equation}
where we use $D(E'||E)\geq 0$ and $I(E':S')\geq 0$.
Note that $\Delta S$ is the change in entropy of the state of all the $W$ bits in the computer to be erased; as discussed in Sec.~\ref{sec:quantum_upperbound}, repeating the single-bit Landauer-erasure protocol $W$ times in total (analyzed for deriving the upper bound) may be more costly as a whole than the erasure of all the $W$ bits collectively~\cite{Buffoni2023,Rolandi2023a}.
Thus, the simultaneous erasure of all $W$ bits is analyzed here to derive the general lower bound based on Landauer's principle.

To write down $\Delta S$ more explicitly,
we recall the fact that the agent has access to the output $a\in\{0,1\}$ stored in the output register before performing the erasure $\mathcal{E}$ in Fig.~\ref{fig:generic_algo}, where $a$ may, in general, be correlated with the computer's state $\rho_a^C$ before the erasure conditioned on $a$.
Even though the single-bit information in the output may not affect the asymptotic scaling, it is still necessary to take it into account to ensure that we bound the required energy consumption from below.
Thus, $\Delta S$ on average is written as the conditional entropy
\begin{equation}
    \Delta S\geq S(C|a)-S(C'|a)\to S(C|a)\quad\text{as $\varepsilon\to 0$,}
\end{equation}
where $S(C'|a)$ is the conditional entropy of the classical-quantum state for the computer's state after the erasure $\mathcal{E}$ conditioned on $a$, and we have $S(C'|a)\to 0$ as $\varepsilon\to 0$ since the computer's state after the erasure in~\eqref{eq:fidelity_condition_erasure} is a fixed pure state in the limit.
These bounds yield
\begin{align}\label{eq:work_entropy_ineq}
    \mathcal Q_E \geq \frac{\Delta S}{\beta} \geq \frac{S(C|a)}{\beta},
\end{align}
which, together with~\eqref{eq:energy_conservation_inequality},  implies the energy-consumption inequality
\begin{equation}
    \mathcal W \geq \frac{S(C|a)}{\beta}
\end{equation}
to prove Theorem~\ref{thm:ThmClassical}.
\end{proof}

One may wonder under which circumstances the bound provided by Theorem~\ref{thm:ThmClassical} on the energy consumption is strictly gapped away from zero, i.e., where the non-zero entropy comes from.
The origin lies in the query-complexity setting, where usually, the oracle $O_f$ is drawn randomly from some set of functions $f$ specified by the problem to be solved.
Then, the agent performs computation to learn some property of $f$ via the queries to $O_f$, but it is not necessarily required to specify the function $f$ itself in full detail.
Conditioned on the algorithm's output $a$, we have a probability distribution $p(f|a)$ according to which the function $f$ (i.e., the oracle $O_f$) is chosen.
A nonzero entropy of this probability distribution $p(f|a)$ leaves a nonzero entropy $S(C|a)$ in the computer's memory.
Then, to erase this information obtained from the (randomly chosen) oracle, nonzero energy consumption is inevitably required due to Laudauer's principle.
This analysis indicates that the entropic lower bound $S(C|a)/\beta$ of energy consumption in Theorem~\ref{thm:ThmClassical} is determined not only by the lower bound of query complexity but also by the probability distribution $p(f|a)$ for the oracle.

In Sec.~\ref{sec:example}, we will derive how the entropy arises and where the connection to $p(f|a)$ lies. Then, we employ the bound in Theorem~\ref{thm:ThmClassical} to prove the exponential energy-consumption advantage of quantum computation over classical computation in solving a computational task, in particular, Simon's problem~\cite{Simon1994, Simon1997}.

\section{\label{sec:example}Exponential quantum advantage in energy consumption}
In this section, using the upper and lower bounds on energy consumption shown in Sec.~\ref{sec:results}, we prove that the energy consumption of quantum computation can be exponentially smaller than that of classical computation within our framework.
In Sec.~\ref{sec:simon}, we recall the standard formulation for Simon's problem given in Refs.~\cite{Simon1994,Simon1997}.
In Sec.~\ref{sec:quantum_simon}, we apply Theorems~\ref{thm:ThmQuantumIdeal} and~\ref{thm:ThmQuantum} to show that the energy consumption of quantum computation to solve Simon's problem scales polynomially with the problem size,
and in Sec.~\ref{sec:classical_simon}, we use Theorem~\ref{thm:ThmClassical} to prove that any classical algorithm for solving the same problem requires an energy consumption that scales exponentially in the problem size.

\subsection{\label{sec:simon}Definition of Simon's problem}

Simon's problem~\cite{Simon1994, Simon1997} is an exemplary case where the query complexity in quantum computation is exponentially separated from any classical computation.
The structure of the quantum algorithm to solve Simon's problem is closely related to the ones to solve the non-oracle-based (i.e., non-relativized) problems of the discrete logarithm and the integer factorization~\cite{Shor1997}.
The issue with the latter non-relativized problems is that lower bounds on the complexity of classical algorithms to solve them are notoriously hard to show in general because they mainly boil down to the open question of whether there exists a classical polynomial-time algorithm to solve the integer-factorization problem~\cite{Arora2009,Boudot2022}.
Studying Simon's problem allows for a robust analysis, where lower bounds on the complexity of classical algorithms do not depend on the conjectured difficulty of a computational task.
Our analysis, therefore, will lay the groundwork for studying the other problems on a case-by-case basis.

Simon's problem \cite{Simon1994, Simon1997} may be posed as the computational task to search for an unknown bit string or, in the way we will define it in Problem~\ref{problem:SimonsProblem}, as the task of deciding whether an unknown function $f:\{0,1\}^N\rightarrow\{0,1\}^N$ with $N$-bit input and $N$-bit output is 1-to-1 or 2-to-1 mapping.
The computer solving the problem has access to an oracle $O_f$, to obtain the output of $f$ for a given input $x$.
For the \textit{classical algorithm}, the oracle is defined as a map in~\eqref{eq:oracle_C} and Fig.~\ref{fig:oracle_QvsC}(a), i.e.,
\begin{align}
    \label{eq:oracle_C_simon}
    O_f^{(\mathrm{C})}(x,0) = (x,f(x)),
\end{align}
for all inputs $x\in\{0,1\}^N$ on the first $N$ bits, with the second $N$ bits initialized in $0$.
When it comes to the \textit{quantum oracle}, on the other hand,
the conventional prescription of the oracle is a unitary map that transforms each computational-basis state $\ket{x,y}$ as
\begin{align}
    \label{eq:oracle_Q_simon}
    O_f^{(\mathrm{Q})}(\ket{x,y}) = \ket{x,y\oplus f(x)},
\end{align}
as shown in~\eqref{eq:oracle_Q} and Fig.~\ref{fig:oracle_QvsC}(b).
For superposition states of $\sum_{x,y}\alpha_{x,y}\ket{x,y}$, the quantum oracle acts linearly on each term of $\ket{x,y}$.

Without further ado, we present the definition of Simon's problem as follows.

\begin{restatable}[Simon's problem \cite{Simon1994, Simon1997}]{problem}{SimonsProblem}
\label{problem:SimonsProblem}
Let $N>0$ be the problem size and $O_f$ the oracle for a function $f:\{0,1\}^N\to\{0,1\}^N$ constructed in the following way: choose an $N$-bit string $s\in\{0,1\}^N\setminus\{0\}$ and a single-bit state $b\in\{0,1\}$ uniformly at random.
If $b=0$, then $f$ is a uniformly randomly chosen 1-to-1 function.
If $b=1$, then $f$ is a uniformly randomly chosen 2-to-1 function satisfying, for all $x\in\{0,1\}^N$,  
\begin{align}
\label{eq:simons_promise}
    f(x)=f(x\oplus s),
\end{align}
where $\oplus$ is the bitwise exclusive OR\@.
The goal is to output $a\in\{0,1\}$ that estimates $b$ correctly, i.e., $a=b$, with a high probability greater than $2/3$ on average, using as few calls to the oracle $O_f$ as possible, where the average success probability is taken over the choices of $s$, $b$, and $f$.
\end{restatable}

In the following, we will apply the techniques developed in Sec.~\ref{sec:results} to show that there exists a quantum algorithm for solving Problem~\ref{problem:SimonsProblem} whose energy consumption is exponentially smaller than that required for any classical algorithm.
In particular, we use Theorem~\ref{thm:ThmQuantum} to show that a quantum computer can solve Simon's problem at an energy consumption $\mathcal{W}\leq\mathcal{W}^{(\text{Q})} = O(\mathrm{poly}(N))$ as the problem size $N$ grows (Corollary~\ref{cor:CorQuantum}), and we apply Theorem~\ref{thm:ThmClassical} to show that all classical algorithms solving the same problem consume the energy of at least
$\mathcal{W}\geq\mathcal{W}^{(\text{C})}=\Omega (2^{N/2} N)$
(Corollary~\ref{cor:CorClassical}).

\subsection{\label{sec:quantum_simon}Quantum upper bound on energy consumption for solving Simon's problem}

In this section, we analyze the energy consumption of a quantum algorithm that solves Simon's problem as defined in Problem~\ref{problem:SimonsProblem} in Sec.~\ref{sec:simon}.
In particular, as a corollary of Theorem~\ref{thm:ThmQuantum}, we show in Corollary~\ref{cor:CorQuantum} by summing up all the contributions of the upper bound of energy consumption in Theorem~\ref{thm:ThmQuantum}, that a fault-tolerant quantum computer can solve Simon's problem with only a polynomial amount of energy consumption in the problem size $N$.
We also remark that the result shows that the polynomial amount of energy consumption in $N$ is achievable for quantum computation solving Simon's problem even if the parameter $E_\mathrm{ctrl.}$ for the control cost in~\eqref{eq:E_ctrl} grows polynomially in $N$.

\begin{restatable}[Upper bound of energy consumption for quantum algorithm to solve Simon's problem]{corollary}{CorQuantum}
\label{cor:CorQuantum}
Let $N>0$ be the problem size and the oracle $O_f$ be chosen according to the probability distribution for Simon's problem in Problem~\ref{problem:SimonsProblem}.
Then, there exists a fault-tolerant quantum algorithm for solving Simon's problem with energy consumption bounded by
\begin{align}
    \label{eq:work_cost_quantum_asymptotic}
    \mathcal{W}\leq \mathcal W^{(\mathrm{Q})} = O\left(E_\mathrm{ctrl.}\times \mathrm{poly}(N)\right).
\end{align}
\end{restatable}

\begin{proof}
In deriving the general upper bound of energy consumption of quantum computation in Theorem~\ref{thm:ThmQuantum} of Sec.~\ref{sec:quantum_upperbound},  we have accounted for all the contributions to the work cost of performing a general quantum computation according to the framework in Fig.~\ref{fig:generic_algo}.
For the proof here, the remaining point to be done is to translate the algorithm presented in Refs.~\cite{Simon1994, Simon1997} (for introductory explanations, see also Refs.~\cite{lecture_note_cleve,Wolf2019}) to our computational framework, so as to calculate the bound in $W$ and
$D_\mathrm{comp}$ (i.e., $D_k$, and $M$ in the definition~\eqref{eq:D_U} of $D_\mathrm{comp}$) accordingly. 

In the quantum algorithm for solving Simon's problem in Refs.~\cite{Simon1994, Simon1997},
a quantum subroutine \textit{Fourier-twice} is repeatedly used.
This quantum subroutine acts on $2N$ qubits initialized as $\ket{0}^{\otimes 2N}$ with the operations in the following order: (i) $N$ parallel Hadamard gates act on the first $N$ bits, i.e., $V=H^{\otimes N}\otimes\mathds{1}^{\otimes N}$ with circuit depth $O(1)$, then, (ii) the oracle $O_f$ acts on all $2N$ qubits simultaneously and finally, (iii) $N$ parallel Hadamard gates $V$ are again applied on the first $N$ qubits.
In the overall quantum algorithm, this subroutine is carried out $O(N)$ times; that is, the oracle $O_f$ (called once in (ii) per performing this subroutine) is called in total $M$ times with
\begin{equation}
\label{eq:M_upper}
    M=O(N).
\end{equation}
In fact, it has been shown that this is also the best a quantum algorithm can do; i.e., any quantum algorithm solving Simon's problem has to query $O_f$ at least $O(N)$ times~\cite{Koiran2005}.

In the algorithm as structured in Fig.~\ref{fig:generic_algo}, $M$ sets of $2N$ qubits are initialized at first, i.e.,
\begin{equation}
\label{eq:W_upper}
    W=2N\times M=O(N^2), 
\end{equation}
and the subroutine is sequentially applied to the $m$th set of the $2N$ qubits for $1\leq m\leq M$.
In this case, $U_1$ is of the form
\begin{equation}
    U_1=V\otimes \mathds{1}^{\otimes 2N(M-1)},
\end{equation}
and all remaining $U_k$ for $2\leq k\leq M$ are of the form 
\begin{align}
    U_k=\mathds{1}^{\otimes 2N(k-2)}\otimes V \otimes V \otimes \mathds{1}^{\otimes 2N (M-k)},
\end{align}
where $\mathds{1}^{\otimes 0}=1$.
The last unitary $U_{M+1}$ equals the operation $\mathds{1}^{\otimes 2N(M-1)}\otimes V$ followed by a quantum implementation $U_\mathrm{post-processing}$ of the classical post-processing of the resulting state obtained from $M$ repetitions of the subroutine, i.e.,
\begin{equation}
    U_{M+1}=U_\mathrm{post-processing}(\mathds{1}^{\otimes 2N(M-1)}\otimes V),
\end{equation}
where $U_{M+1}$ can be performed within time complexity $O(N^3)$ \cite{Wolf2019}.
Therefore, using, e.g., the Clifford+$T$ universal gate set, each of these unitaries can be implemented efficiently with circuit depths, for all $1\leq k \leq M$,
\begin{equation}
\label{eq:D_k_upper}
    D_k=O(1),
\end{equation}
and
\begin{equation}
\label{eq:D_M_upper}
    D_{M+1}=O(N^3).
\end{equation}
Note that the classical post-processing in $U_{M+1}$ can be represented by a reversible classical logic circuit, e.g., written in terms of the classical Toffoli gate, and by replacing each classical Toffoli gate with the quantum Toffoli gate (and further rewrite each quantum Toffoli gate in Clifford and $T$ gates if we use the Clifford+$T$ universal gate set), we obtain the quantum circuit of $U_\mathrm{post-processing}$ of depth $O(N^3)$.

Therefore, in the ideal case without gate error, inserting the explicit numbers in~\eqref{eq:M_upper},~\eqref{eq:W_upper},~\eqref{eq:D_k_upper}, and~\eqref{eq:D_M_upper} into the upper bound of the energy consumption in Theorem~\ref{thm:ThmQuantumIdeal}, we find that the energy consumption in the quantum algorithm in Fig.~\ref{fig:generic_algo} to solve Simon's problem is bounded by
\begin{equation}
    \mathcal{W}\leq\mathcal{W}^{(\mathrm{Q})}=O\left(E_\mathrm{ctrl.} N^5\right),
\end{equation}
where the dominant part is the one for implementing $U_{M+1}$ by a circuit of volume $WD_{M+1}=O(N^5)$, and $\beta$, $\eta$, and $\varepsilon$ are set as constants.
Even in the case where we implement this quantum algorithm by a fault-tolerant protocol without measurement in Theorem~\ref{thm:ThmQuantum}, the energy consumption in implementing the quantum algorithm to solve Simon's problem in a fault-tolerant way is bounded by
\begin{equation}
    \mathcal{W}\leq\mathcal{W}^{(\mathrm{Q})}=\widetilde{O}\left( E_\mathrm{ctrl.}N^8\right),
\end{equation}
where the dominant part is the one for implementing $U_{M+1}$ by a circuit of volume $\widetilde{O}(WD_{M+1}D_{M+1})=\widetilde{O}(N^8)$, and $\widetilde{O}$ may ignore polylogarithmic factors.
In any of these cases,  we obtain
\begin{align}
    \mathcal{W}\leq\mathcal{W}^{(\text{Q})} = O( E_\mathrm{ctrl.}\times \mathrm{poly}(N)),
\end{align}
which scales polynomially in $N$ so long as the parameter $E_\mathrm{ctrl.}$ for the control cost also scales polynomially in $N$.
\end{proof}

\subsection{\label{sec:classical_simon}Classical lower bound on energy consumption for solving Simon's problem}

In this section, we use the lower bound of energy consumption in Theorem~\ref{thm:ThmClassical} to prove that any classical computation that can solve Simon's problem (Problem~\ref{problem:SimonsProblem}) should require an exponentially large energy consumption in the problem size $N$, no matter how energy-efficiently the algorithm is implemented, as shown in Corollary~\ref{cor:CorClassical} below.

The essential idea for deriving this lower bound of the energy consumption is to establish a connection between the conventional argument for a lower bound of query complexity of classical algorithms for solving Simon's problem~\cite{Simon1994, Simon1997} and the entropic lower bound of the energy consumption that we derived in Theorem~\ref{thm:ThmClassical}.
In Problem~\ref{problem:SimonsProblem}, $s\in\{0,1\}^N\setminus\{0\}$, $b\in\{0,1\}$, and $f$ are chosen uniformly at random, and the classical algorithm makes $M$ queries to the oracle $O_f$ using $M$ different classical inputs $x_1,\ldots,x_M\in\{0,1\}^N$ to estimate $b$, i.e., to output $a\in\{0,1\}$ satisfying $a=b$, with a high probability greater than $2/3$.
We write these $M$ inputs to $O_f$ collectively as
\begin{equation}
\label{eq:vec_x}
    \vec x\coloneqq(x_1,\ldots,x_M).
\end{equation}
Note that a random guess of $b$ would only achieve the average success probability of $1/2$, but Problem~\ref{problem:SimonsProblem} requires a strictly better average success probability $2/3=1/2+\Delta$ for $\Delta=1/6$.

We here provide a necessary condition on $M$ for achieving the average success probability $2/3$ in Simon's problem as required in Problem~\ref{problem:SimonsProblem}.
In Refs.~\cite{Simon1994, Simon1997}, it has been shown that, in order to have the average success probability $1/2+2^{-N/2}$, it is necessary to make $M$ queries with at least $M\geq2^{N/4}$;
later in Refs.~\cite{lecture_note_cleve,Wolf2019}, this lower bound has been improved by proving that a larger lower bound $M=\Omega(2^{N/2})$ should be necessary for achieving the average success probability $1/2+\Omega(1)$, but the constant factors of the lower bounds were not explicitly presented.
Refining these analyses, we here show that, for any problem size $N$ and required average success probability $1/2+\Delta$, we have a lower bound of the query complexity of Simon's problem with an explicit constant factor given by $\left\lceil\sqrt{\frac{2\Delta}{1+\Delta}}\times2^{N/2}\right\rceil$, where $\Delta\in(0,1/2)$ is an arbitrary constant in this parameter region and is suitably tunable for our analysis.
In particular, we prove the following proposition, where the proof is given in Appendix~\ref{appendix:proofs_classical_bound}.

\begin{restatable}[Refined analysis of lower bound of query complexity of  Simon's problem for classical algorithms]{proposition}{SimonBound}
\label{prp:simon}
    For any parameter $\Delta\in(0,1/2)$,
    if we have
    \begin{equation}
    \label{eq:M_N_delta}
        M<M_{N,\Delta}\coloneqq\left\lceil\sqrt{\frac{2\Delta}{1+\Delta}}\times2^{N/2}\right\rceil,
    \end{equation}
    the average success probability of any probabilistic classical algorithm with at most $M$ queries for Simon's problem in Problem~\ref{problem:SimonsProblem} is upper bounded by
    \begin{equation}
    \label{eq:7_12}
        \frac{1}{2}+\Delta.
    \end{equation}
\end{restatable}

Due to this proposition, the required number $M$ of the oracle queries for the classical algorithms to solve Simon's problem is lower bounded, for any $N$ and any choice of $\Delta\in(0,1/2)$ in Proposition~\ref{prp:simon}, by
\begin{equation}
\label{eq:M_geq_2N4}
    M\geq M_{N,\Delta}=\Omega(2^{N/2})
\end{equation}
with a high probability, where we quantify this probability more precisely in Lemma~\ref{lemma:QueryNumber} of Appendix~\ref{appendix:proofs_classical_bound}.

In the following corollary on the lower bound of energy consumption, we argue that the above query-complexity bound can be used for evaluating a dominant contribution to the entropy of the state of the computer.
Thus, as the corollary of Theorem~\ref{thm:ThmClassical}, we have the lower bound of the energy consumption of classical computation.
Our lower bound of the energy consumption for any classical algorithm for solving Simon's problem scales $\Omega(2^{N/2}N)$ (with explicit constant factors derived in~\eqref{eq:lower_bound} and~\eqref{eq:tight_lower_bound} in the proof), which is different from that of the query complexity $\Omega(2^{N/2})$ for the same problem (i.e., those in Refs.~\cite{lecture_note_cleve,Wolf2019} and Proposition~\ref{prp:simon}).
Indeed, as we will show later in Proposition~\ref{prp:simon_achievability}, the upper bound of query complexity of Simon's problem for classical computation is also $O(2^{N/2})$ and thus is strictly smaller than the lower bound of energy consumption $\Omega(2^{N/2}N)$.
As we have argued in Sec.~\ref{sec:classical_lowerbound}, this difference arises because the query complexity only counts the number of queries, but the lower bound of the energy consumption of classical computation is determined by the amount of information (i.e., entropy) obtained from the queries, depending also on the probability distribution of the oracle.\footnote{ 
If one divides the entropy by the number of queries, one may be able to define the entropy per query, but it still remains unclear whether one could directly evaluate a lower bound of the entropy per query without analyzing the overall entropy for all the queries.
In the case of Simon's problem, the entropy per query scales linearly with the problem size $N$; however, this scaling depends on the probability distributions of the oracles used in the computational problems and thus may differ in general.
Therefore, we need to evaluate a lower bound of the overall entropy, not just the query-complexity lower bound, to derive a lower bound of energy consumption.}
In this regard, the difference in the scaling of the energy consumption and the query complexity of classical computation solving Simon's problem quantitatively illustrates the fact that the energy consumption and the number of queries are different computational resources despite the resemblance in the proof techniques.

\begin{restatable}[Lower bound of energy consumption for all classical algorithms to solve Simon's problem]{corollary}{CorClassical}
\label{cor:CorClassical}
Let $N>0$ be the problem size of Simon's problem in Problem~\ref{problem:SimonsProblem}.
Any probabilistic and adaptive classical algorithm to solve Simon's problem has an energy consumption of at least
\begin{align}
    \label{eq:simons_work_cost_classical}
    \mathcal W\geq \mathcal{W}^{(\mathrm{C})} = \Omega\qty(\frac{2^{N/2}N}{\beta}).
\end{align}
\end{restatable}

\begin{proof}
As formulated in Sec.~\ref{sec:setting}, the crucial assumption for our analysis is that the knowledge on the input-output relations for $f$ obtained from the oracle queries in the classical computation remains in the computer, up to possible reversible transformations $U_k$, until erased with the thermodynamical work cost.
As for the choice of the input sequence $\vec x$ in~\eqref{eq:vec_x}, the classical algorithm can be probabilistic and adaptive in general.
For our analysis to cover these general cases,
we consider the classical algorithms where the classical input $x_{k+1}$ to the oracle $O_f$ is sampled probabilistically and can depend on the previous inputs $x_1,\dots,x_k$ and the oracle outputs $f(x_1),\dots,f(x_k)$.
Moreover, after each query to the oracle, the algorithm may probabilistically decide to stop or to continue further querying the oracle.
In particular, let $q_1(x_1)$ denote the probability distribution on the choice of $x_1$, and for each $k>1$,  conditioned on the previous oracle calls with input $x_\ell$ and output function value $f(x_\ell)$ for $1\leq \ell <k$,
we write the conditional probability distribution on the next choice of $x_k$ as $q_{k}(x_k | \{(x_\ell,f(x_\ell)): 1\leq\ell< k\})$. 
The algorithm uses $M$ different inputs, and thus,
if $x_k$ has already appeared previously in $x_1,\ldots,x_{k-1}$, i.e., $x_k \in\{x_\ell : 1\leq \ell < k\}$,
then we have $q_k\big(x_k | \{(x_\ell,f(x_\ell)) : 1\leq\ell< k\}\big) = 0$.
For a given sequence $((x_1,f(x_1),\dots,(x_k,f(x_k))$ of $k$ queries, the algorithm stops with probability
\begin{align}\label{eq:q_stop_def}
    q(\mathrm{stop}|\{(x_\ell,f(x_\ell)) : 1\leq \ell \leq k\}),
\end{align}
or continues for the $(k+1)$th query with probability
\begin{align}\label{eq:q_continue_def}
    q(&\mathrm{continue}|\{(x_\ell,f(x_\ell)) : 1\leq \ell \leq k\})\nonumber \\
     &\coloneqq 1-q(\mathrm{stop}|\{(x_\ell,f(x_\ell)) : 1\leq \ell \leq k\}).
\end{align}
Since we require that $x_1,\ldots,x_k\in\{0,1\}^N$ should be $k$ different inputs among $2^N$ elements, it holds by construction that
\begin{equation}
\label{eq:q_continue_0}
    q(\mathrm{continue}|\{(x_\ell,f(x_\ell)):1\leq\ell\leq 2^N\})=0,
\end{equation}
which is the requirement that the algorithm must eventually terminate with at most $2^N$ queries.
The overall strategy of the classical algorithm is determined by the strategy for choosing $\vec x$, i.e., $q_1,q_2,\ldots$, and that for stopping, i.e., $q(\mathrm{stop}| \dots)$ and  $q(\mathrm{continue}| \dots)$.
Given this general formulation of strategies in choosing the input sequence $\vec x$ to the oracle, we can write the computer's memory state before outputting $a$ in Fig.~\ref{fig:generic_algo}, conditioned on having queried the function oracle exactly $M$ times, up to reversible transformation, as
\begin{widetext}
    \begin{align}
     \rho_{f,M}=\frac{1}{p_f(M)}\sum_{\vec x}&\bigg(q\left(\mathrm{stop}|\{(x_\ell,f(x_\ell))\}_{1\leq\ell\leq M}\right)q_M(x_{M}|\{(x_\ell,f(x_\ell))\}_{1\leq\ell\leq M-1}) \nonumber \\
     &\qquad \prod_{k=1}^{M-1}q(\mathrm{continue}|\{(x_\ell,f(x_\ell))\}_{1\leq\ell\leq k})q_k(x_{k}|\{(x_\ell,f(x_\ell))\}_{1\leq\ell\leq k-1})\bigg)\rho_{f,\vec x}.
    \label{eq:rho_f_M}
    \end{align}
\end{widetext}
where
$p_f(M)$ is the probability of the strategy yielding exactly $M$ queries defined via normalization $\tr[\rho_{f,M}]=1$,
we may write $q_1(x_1)$ as $q_{k}(x_k | \ldots)$ with $k=1$ for simplicity of notation,
\begin{equation}
\label{eq:rho_f_x}
    \rho_{f,\vec x}\coloneqq\bigotimes_{k=1}^{M}\ket{x_k,f(x_k)}\bra{x_k,f(x_k)} 
\end{equation}
is a $(2N\times M)$-bit classical state, i.e., a diagonal density operator, representing the $M$ different input-output relations, and we here write only the relevant $(2N\times M)$-bit part of the state depending on $f$ and $\vec x$ out of the $W$ bits in the computer.

In general, for fixed $s$, $b$, and $f$ in Problem~\ref{problem:SimonsProblem}, the required number of queries $M$ for the algorithm to succeed in estimating $b$ may change probabilistically, depending on the strategy of the algorithm, i.e., $q_1,q_2,\ldots$, $q(\mathrm{stop}| \dots)$, and $q(\mathrm{continue}| \dots)$.
We write the average state before the output, in expectation taken over all possible number of queries $M\geq 1$, as
\begin{align}
\label{eq:rho_f}
    \rho_f=\sum_{M\geq 1}p_f(M)\rho_{f,M}.
\end{align}
For the agent before seeing any output from the computer, this state is further averaged over the uniform probability distributions $p(s)$, $p(b)$, $p(f|s,b)$ as described in Problem~\ref{problem:SimonsProblem}, i.e., up to reversible transformation,
\begin{align}
\label{eq:p}
    \rho^C&=\sum_{s,b}\sum_{f}p(s)p(b)p(f|s,b)\rho_{f}\\
    &=\sum_{M\geq 1}\sum_{s,b}\sum_{f}p(s)p(b)p(f|s,b)p_f(M)\rho_{f,M}\\
    &=\sum_{M\geq 1}p(M)\rho_{M},
\label{eq:rho_a}
\end{align}
where we write in the last line
\begin{align}
    \label{eq:p_M}
    p(M)&\coloneqq\sum_{s,b}\sum_{f}p(s)p(b)p(f|s,b)p_f(M),\\
    \label{eq:rho_M}
    \rho_{M}&\coloneqq\frac{1}{p(M)}\sum_{s,b}\sum_{f}p(s)p(b)p(f|s,b)p_f(M)\rho_{f,M}.
\end{align}
In this notation, $p(M)$ represents the probability of the algorithm stopping exactly after $M$ queries on average over the choices of $s$, $b$, and $f$, and $\rho_M$ is the state of the computer's memory before the output, conditioned on the algorithm stopping exactly after $M$ queries, on average over $s$, $b$, and $f$.

To output $a$ in Fig.~\ref{fig:generic_algo}, the agent uses the output register prepared in $\ket{0}$ and performs a swap gate to pull out one of the $W$ bits from the computer to the output register, where the state of the bit in the computer on which the swap gate acts becomes $\ket{0}$ after these output operations.
After the swap gate, the $(W+1)$-bit state of the computer and the output register is equivalent to $\rho^C\otimes\ket{0}\bra{0}$ up to reversible transformation. 
A measurement of the output register is then performed to obtain the single-bit output $a\in\{0,1\}$ with probability $p_\mathrm{out}(a)$.
For the agent who knows the output $a$ of the computation, 
before performing the erasure $\mathcal{E}$ conditioned on $a$, 
the computer's state is in a $W$-qubit mixed state $\rho_a^C$, which is a classical state represented by a diagonal density operator obtained by measuring one of the $W+1$ qubits in the basis $\{\ket{a}:a=0,1\}$.
As a result, for any agent erasing the computer's state, the probabilistic state of the computer's memory has the conditional entropy $S(C|a)$ of the state
\begin{align}
    \rho^{Ca}\coloneqq\sum_{a}p_\mathrm{out}(a)\ket{a}\bra{a}\otimes\rho_a^C.
\end{align}
Since the entropy is invariant under the reversible transformation, we have
\begin{align}
     S\qty(\rho^C\otimes\ket{0}\bra{0})=S\qty(\rho^{Ca}).
\end{align}
Thus, by definition of the conditional entropy, it holds that
\begin{align}
    S(C|a)&=S\qty(\rho^C\otimes\ket{0}\bra{0})-S(a)\\
    &\geq S\qty(\rho^C)-\ln 2\\
\label{eq:SCa_geq_SC}
    &\geq \sum_{M\geq 1}p(M)S(\rho_M)-\ln 2,
\end{align}
where the $\ln2$ comes from $S(a)\leq\ln 2$ for a single-bit state, and the last line is due to~\eqref{eq:rho_a} and the concavity of the entropy~\cite{Nielsen2010}.

With the bound~\eqref{eq:SCa_geq_SC}, we evaluate the entropic lower bound in Theorem~\ref{thm:ThmClassical},
by taking into account the assumption that the algorithm should stop successfully.
We are proving the entropic lower bound for \textit{all} classical algorithms with success probability greater than $2/3$,
and importantly, our analysis takes into account all possible classical algorithms that try to minimize energy consumption at their best.
For example, if the input sequence $\vec x=(x_1,\ldots,x_M)$ obtained from $M$ queries includes some pair $(x_k,x_m)$ satisfying $f(x_k)=f(x_m)$, then the algorithm can know that $f$ is a $2$-to-$1$ function with $s=x_k\oplus x_m$.
But even in such a case, the algorithm does not have to stop immediately but is still allowed to continue queries and computation, e.g., for trying to uncompute some part of the computer's memory so as to reduce the overall energy consumption at best.
On the other hand, as long as the overall success probability is greater than $2/3$, the algorithm may stop even before making the $M_{N,\Delta}$ queries with some small probabilities to save the energy consumption (we allow such stops even if the input sequence $\vec x$ with $M < M_{N,\Delta}$ does not contain any pair satisfying $f(x_k)=f(x_m)$).
Our proof covers all these cases, to find the fundamental lower bound on energy consumption for all possible classical algorithms to solve Simon's problem~\ref{problem:SimonsProblem}.

To arrive at a lower bound on the entropy of the computer's state, we have to constrain the possible probability distributions $p(M)$ and the entropy $S(\rho_M)$ in~\eqref{eq:SCa_geq_SC}.
To bound $p(M)$ in~\eqref{eq:SCa_geq_SC}, we use the requirement in Problem~\ref{problem:SimonsProblem} that the overall success probability of the algorithm should be greater than $2/3$.
Indeed, as we show in Lemma~\ref{lemma:QueryNumber} in Appendix~\ref{appendix:proofs_classical_bound}, if the success probability is greater than or equal to $1/2 + \Delta$ for any parameter $\Delta\in(0,1/2)$,
the probability of making at least $M_{N,\Delta}$ queries is bounded from Proposition~\ref{prp:simon} by
\begin{align}
\label{eq:sum_pM_bound}
    \sum_{M\geq M_{N,\Delta}} p(M) \geq \frac{1-6\Delta}{3-6\Delta}.
\end{align}
The lower bound on the right-hand side depends on the choice of $\Delta$.
In Problem~\ref{problem:SimonsProblem}, we demand that the algorithm for solving this problem must have a success probability greater than $2/3$.
Any such algorithm will also have a success probability greater than $1/2+\Delta$ if the choice of $\Delta$ is in the range
\begin{equation}
\label{eq:Delta_range}
    \Delta\in\left(0,\frac{1}{6}\right),
\end{equation}
i.e.,
\begin{equation}
    \frac{1}{2}+\Delta<\frac{2}{3}.
\end{equation}
With $\Delta$ in this range, the lower bound on the right-hand side of~\eqref{eq:sum_pM_bound} applies to all the algorithms solving Problem~\ref{problem:SimonsProblem} with a success probability greater than $2/3$; also, in this range, the right-hand side of~\eqref{eq:sum_pM_bound} indeed becomes positive.

As for the bound of $S(\rho_M)$ in~\eqref{eq:SCa_geq_SC},
in Proposition~\ref{prop:ClassicalEntropy} of Appendix~\ref{appendix:proofs_classical_bound}, we show explicitly that, for any $M$, the entropy of $\rho_M$ is bounded by
\begin{align}
\label{eq:S_rho_M_bound}
    S(\rho_M) \geq \frac{1}{2}\ln \frac{2^N!}{(2^N-M)!},
\end{align}
regardless of whether $\rho_M$ consists of good or bad input sequences $\vec x$;
in particular, since the right-hand side of~\eqref{eq:S_rho_M_bound} is monotonously increasing in $M=1,2,\ldots$, we can write, for all
$M\geq M_{N,\Delta}$,
\begin{align}
    S(\rho_M) \geq  \frac{1}{2}\ln\frac{2^N!}{\left(2^N-M_{N,\Delta}\right)!}.
\end{align}

\begin{widetext}
Inserting~\eqref{eq:sum_pM_bound} and~\eqref{eq:S_rho_M_bound} into~\eqref{eq:SCa_geq_SC}, we obtain
\begin{align}
\sum_{M\geq 1}p(M) S(\rho_M)
&= \sum_{M<M_{N,\Delta}}p(M)S(\rho_M)+\sum_{M\geq M_{N,\Delta}}p(M)S(\rho_M)\\
    &\geq 0 + \frac{1}{2}\ln\frac{2^N!}{\left(2^N-M_{N,\Delta}\right)!}\times\sum_{M\geq M_{N,\Delta}}p(M)\\
    &\geq \frac{1}{2}\ln\frac{2^N!}{\left(2^N-M_{N,\Delta}\right)!}\times\frac{1-6\Delta}{3-6\Delta}\\
    &= \frac{1-6\Delta}{6-12\Delta}\ln\frac{2^N!}{\left(2^N-M_{N,\Delta}\right)!}.
    \label{eq:SrhoC_bound_factorial}
\end{align}
Furthermore, as we show in Lemma~\ref{prp:fraction_evaluation} of Appendix~\ref{appendix:proofs_classical_bound}, for $M_{N,\Delta}$ in~\eqref{eq:M_N_delta}, a variant of Stirling's approximation~\cite{10.2307/2308012} yields
\begin{equation}
\label{eq:tighter_stirling_03_main}
    \ln\frac{2^N!}{\left(2^N-M_{N,\Delta}\right)!}\geq\sqrt{\frac{2\Delta}{1+\Delta}} 2^{N/2}N \ln 2 - 3.
\end{equation}
Combining the results from~\eqref{eq:SCa_geq_SC}, \eqref{eq:SrhoC_bound_factorial} and~\eqref{eq:tighter_stirling_03_main}, for all classical algorithms that solve Simon's problem in our framework, we have
\begin{align}
     S(C|a)\geq\frac{1-6\Delta}{6-12\Delta}\left(\sqrt{\frac{2\Delta}{1+\Delta}} 2^{N/2}N \ln 2 - 3\right)-\ln 2=\Omega(2^{N/2}N),
\end{align}
which is gapped away from zero independently of how energy-efficiently the classical algorithm is implemented, as long as the algorithm solves Simon's problem.

Therefore, for any choice of the constant parameter $\Delta\in(0,1/6)$ in~\eqref{eq:Delta_range}, applying Theorem~\ref{thm:ThmClassical} yields a lower bound of the energy consumption
\begin{align}
\label{eq:lower_bound}
     \mathcal{W}\geq\frac{1}{\beta}S(C|a)\geq\frac{1}{\beta}\times\left(\frac{1-6\Delta}{6-12\Delta}\left(\sqrt{\frac{2\Delta}{1+\Delta}} 2^{N/2}N \ln 2 - 3\right)-\ln 2\right)=\Omega\qty(\frac{2^{N/2}N}{\beta}).
\end{align}
In particular, the maximum of the constant factor
\begin{equation}
    \frac{1-6\Delta}{6-12\Delta}\sqrt{\frac{2\Delta}{1+\Delta}}
\end{equation}
of the leading order in~\eqref{eq:lower_bound} is achieved with the choice of
\begin{equation}
    \Delta=2-\frac{\sqrt{15}}{2}=0.0635\cdots\in\left(0,\frac{1}{6}\right),
\end{equation}
and thus, we can explicitly obtain the tightest lower bound among the choices of $\Delta\in(0,1/6)$ as
\begin{equation}
\label{eq:tight_lower_bound}
    \mathcal{W}\geq\frac{1}{\beta}\times\left(\frac{3\sqrt{15}-11}{6\sqrt{15}-18}\left(\sqrt{\frac{8-2\sqrt{15}}{6-\sqrt{15}}}\times 2^{N/2}N \ln 2 - 3\right)-\ln 2\right).
\end{equation}
\end{widetext}
These lower bounds conclude the proof.
\end{proof}

\section{\label{sec:experiment}Proposal on experimental demonstration}
In this section, we propose an experimental scenario for demonstrating the energy-consumption advantage of quantum computation over classical computation in solving Simon's problem analyzed in Sec.~\ref{sec:example}.
In the following, we begin by presenting a schematic setup for measuring the energy consumption of quantum computation in experiments.
Then, we clarify the explicit polynomial-time way of instantiating the oracle for Simon's problem to implement this setup.
Finally, we also show approaches to estimate the energy consumption for classical computation to solve the same problem, which yields criteria to be compared with the results of the quantum experiments.

\begin{figure}
    \centering
    \includegraphics[width=3.4in]{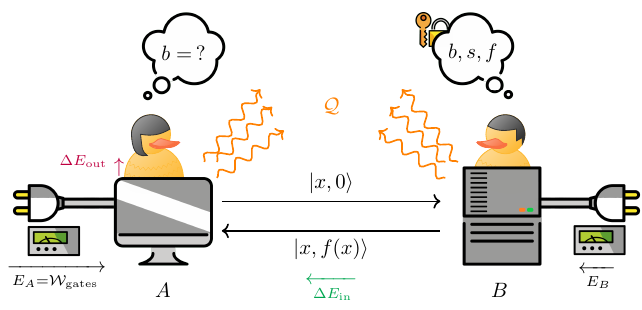}
    \caption{
    A schematic depiction of the setup to experimentally demonstrate the exponential energy-consumption advantage of quantum computation over classical computation based on solving Simon's problem.
    The setup is comprised of two parties: $B$ who keeps a bit $b\in\{0,1\}$ and an $N$-bit string $s\in\{0,1\}^{N}\setminus\{0\}$ as $B$'s secret, and $A$ who tries to estimate $B$'s secret $b$, as in Simon's problem in Problem~\ref{problem:SimonsProblem}.
    With $b$ and $s$ chosen uniformly at random, $B$ initially constructs a function $f$ in such a way that $B$ can implement the oracle $O_f$ feasibly within a polynomial time for quantum algorithms to solve Simon's problem, yet $A$ cannot distinguish $f$ from the uniformly random choice as required for Simon's problem in Problem~\ref{problem:SimonsProblem}, using a cryptographic primitive as described in the main text.
    To solve Simon's problem, $A$ runs the quantum algorithm for solving Simon's problem on $A$'s quantum computer, and $B$ acts as the oracle $O_f$.
    Throughout the computation, $B$'s internal implementation of $O_f$ is never directly revealed to $A$; to estimate $B$'s secret $b$, $A$ can query $B$ by sending a superposition of states $\ket{x,0}$ to $B$, and $B$ internally performs $O_f$ and sends back the corresponding superposition of $\ket{x,f(x)}$ to $A$.
    These queries provide the only information $A$ can access for guessing $b$, as in Simon's problem.
    What we care about in this experiment is the energy consumption.
    By measuring $A$'s energy consumption $E_A$ and $B$'s energy consumption $E_B$ in solving Simon's problem, we obtain an achievable energy consumption for the implemented quantum computation from their sum $E_A+E_B$, as shown in~\eqref{eq:measured_upper_bound_energy_consumption}.
    In the main text, we also clarify how to estimate the energy consumption for classical computation to solve the same problem, using extrapolation or the fundamental lower bound (Table~\ref{table:1}).
    The experimental demonstration of the energy-consumption advantage of quantum computation will be successful if $E_A+E_B$ is smaller than the estimate of energy consumption for classical computation.
    }
    \label{fig:exp_sketch}
\end{figure}

\paragraph*{Setup for measuring the energy consumption of quantum computation in experiments.}

We propose a schematic setup depicted in Fig.~\ref{fig:exp_sketch} to capture our framework in Fig.~\ref{fig:generic_algo} and measure the energy consumption of quantum computation therein.
A challenge in the experimental realization of the framework in Fig.~\ref{fig:generic_algo} is the implementation of the oracle as required in Sec.~\ref{sec:computational_setting}.
To address this challenge, the setup in Fig.~\ref{fig:exp_sketch} consists of two parties $A$ and $B$ with their quantum computers, in a setup similar to that of a client and a server: $A$ is where the main computation is carried out, and $B$ is where the oracle is instantiated without revealing its internal to $A$.
Quantum communication between $A$ and $B$ in Fig.~\ref{fig:exp_sketch} implements the query to the oracle in Fig.~\ref{fig:generic_algo}.
For Simon's problem, i.e., Problem~\ref{problem:SimonsProblem}, $B$ selects a bit $b\in\{0,1\}$ and a bit string $s\in\{0,1\}^N\setminus\{0\}$ of length $N$ uniformly at random.
Based on $b$ and $s$, $B$ constructs an $N$-bit function $f:\{0,1\}^N\rightarrow\{0,1\}^N$ in such a way that $B$ can internally implement the oracle $O_f$ for $f$.
The particular choice of $b$, $s$, and $f$
is kept as $B$'s secret so that $A$ should not be able to distinguish $B$'s construction of $f$ from the uniformly random choice of $f$ as in Problem~\ref{problem:SimonsProblem}.

In this two-party setup, the task is for $A$ to estimate the bit $b$ that $B$ has internally.
To achieve this task, $A$ performs the quantum algorithm for Simon's algorithm, i.e., that analyzed in Sec.~\ref{sec:quantum_simon}.
One way to implement the initialization $\mathcal{E}$ in Fig.~\ref{fig:generic_algo} is to use our finite-step Laudauer-erasure protocol for Theorem~\ref{thm:ThmErasureBound}, but one can also use other initialization techniques available in the experiments.
Without the finite-step Laudauer-erasure protocol, the upper bound of the achievable initialization and control costs may not be theoretically guaranteed, but the experiment will still be successful as long as the energy consumption of the implementation of the quantum algorithm is below an estimate of that of the classical algorithms.
Each query to the oracle in the quantum algorithm, implemented by a swap gate in Figs.~\ref{fig:oracle_QvsC}(b) and~\ref{fig:generic_algo}, is replaced with that implemented by (quantum) communication between $A$ and $B$, as shown in Fig.~\ref{fig:exp_sketch}.
Throughout the experiment, $B$ never directly reveals $b$ to $A$, but $B$ only receives the input state for the query from $A$ and returns the corresponding output state of the oracle $O_f$ to $A$ by using $B$'s internal implementation of $O_f$.
In other words, in implementing this experiment in a lab, imagine that the party $A$ has some program to run its computation, and also the party $B$ has some program to compute the function $f$ of the oracle; then, in the lab, we require that $A$'s program should not directly read $b$ in $B$'s internal program but should estimate $b$ through the queries.
Thus, in the setup of Fig.~\ref{fig:exp_sketch}, $A$ can estimate $B$'s secret $b$ only from the input-output relations of $O_f$ obtained from the queries implemented by communication, in the same way as estimating $b$ from those obtained from the oracle queries in Fig.~\ref{fig:generic_algo}.

To measure the achievable energy consumption of quantum computation in this setup, i.e., $\mathcal{W}_\mathrm{gates}+\Delta E^{(\mathrm{in})}-\Delta E^{(\mathrm{out})}$ in Definition~\ref{def:DefEnergyConsumption}, it suffices to measure upper bounds of $\mathcal{W}_\mathrm{gates}-\Delta E^{(\mathrm{out})}$ and $\Delta E^{(\mathrm{in})}$, so as to sum them up.
For this purpose, we keep track of the amount of energy supply that is used for the implementation of $A$'s computation, and that for the implementation on $B$'s side.
For example, these amounts can be obtained by reading the difference of the electricity meters before and after conducting the computation, which shows the electricity consumption.
Also, on the output system in the experiment, we set the state $0$ as a low-energy state and $1$ as a high-energy state, so that it should hold that
\begin{equation}
    \Delta E^{(\mathrm{out})}\geq 0
\end{equation}
in the experiment.
With this setup, the energy consumption $E_A$ on $A$'s side is, by construction,
\begin{equation}
    E_A=\mathcal{W}_\mathrm{gates}\geq\mathcal{W}_\mathrm{gates}-\Delta E^{(\mathrm{out})}.
\end{equation}
Similarly, the energy consumption $E_B$ on $B$'s side is the sum of the energetic cost $\Delta E^{(\mathrm{in})}$ and the control cost of implementing the oracle $O_f$.
As explained in Sec.~\ref{sec:computational_setting}, our theoretical analysis has ignored the control cost of implementing the oracle by convention, but the control cost is nonnegative by definition.
Therefore, we have
\begin{equation}
    E_B\geq\Delta E^{(\mathrm{in})}.
\end{equation}
Note that the implementation of communication between $A$ and $B$ may also incur a non-negative control cost, but we can include this control cost in $E_A$ or $E_B$ (e.g., by using the energy supply from $A$ and $B$ as the power source for the communication devices).
As a whole, the sum of the measured energy consumption
\begin{equation}
\label{eq:measured_upper_bound_energy_consumption}
    E_A+E_B
\end{equation}
serves as an upper bound of the energy consumption $\mathcal{W}_\mathrm{gates}+\Delta E^{(\mathrm{in})}-\Delta E^{(\mathrm{out})}$.

The experimental demonstration of the advantage of quantum computation will be successful if the measured upper bound of the energy consumption $E_A+E_B$ in~\eqref{eq:measured_upper_bound_energy_consumption} is smaller than an estimate of the energy consumption for classical computation.
Based on our results in Corollaries~\ref{cor:CorQuantum} and~\ref{cor:CorClassical}, a separation of $\mathcal W^{\text(Q)}$ and $\mathcal W^{\text(C)}$ will emerge.
If the experiment achieves
\begin{equation}
\label{eq:E_A_E_B_condition}
    E_A+E_B=O(\mathrm{poly}(N)),
\end{equation}
then the advantage can be demonstrated for large $N$ due to the exponentially large energy consumption for the classical computation to solve the same problem.
In the following, we will propose an explicit scheme according to which $B$ can construct $O_f$ as required in our framework yet still within a polynomial energy consumption to achieve~\eqref{eq:E_A_E_B_condition}, followed by also clarifying how to explicitly evaluate the lower bound of the energy consumption for the classical computation.

\paragraph*{Construction of oracle for Simon's problem.}
To realize~\eqref{eq:E_A_E_B_condition}, it is necessary that the control cost of $B$'s implementation of $O_f$ should be bounded by $O(\mathrm{poly}(N))$ so that we should have $E_B=O(\mathrm{poly}(N))$.
Then, it is essential that $B$ can compute $f$ within a polynomial time (i.e., by a polynomial-size circuit) so that the oracle $O_f$ can also be implemented in a polynomial time (and a polynomial control cost).
The formulation of Simon's problem in Problem~\ref{problem:SimonsProblem}, however, allows for random $1$-to-$1$ and $2$-to-$1$ functions $f$ on $N$ bits, which may not be computable within a polynomial time but take an exponential time in general.
Note that in Simon's problem, the choice of $f$ is random, but once $f$ is chosen, the oracle needs to output the same $f(x)$ for fixed input $x$, which is different from producing random outputs conditioned on $x$.
The challenge here arises from the fact that, in general, implementation of randomly chosen Boolean functions almost always requires an exponential-size circuit~\cite{6771698}.
In constructing $f$ in the experiment, we, therefore, have to accept an approximation to the uniformly random choice of $f$ required for Simon's problem.

To overcome the challenge,
we propose a construction of $f$ that can be implemented within a polynomial time by $B$ yet cannot be distinguished from the uniformly random choice of $f$ for $A$.
To instantiate the random functions $f:\{0,1\}^N\to\{0,1\}^N$ in Simon's problem, we use a pseudorandom permutation, a well-established primitive of cryptography~\cite{goldwasser1996lecture}.
For some key bit string $k\in\{0,1\}^L$ of length $L$, a pseudorandom permutation $f_k$ can be computed by a block-cipher algorithm such as Advanced Encryption Standard (AES)~\cite{Dworkin2023,Daemen:2002:DRA}, using a classical circuit of a polynomial size in $N$ and $L$.
For a sufficiently large fixed $L$, the key $k\in\{0,1\}^L$ is selected uniformly at random by $B$ and kept as $B$'s secret in the same way as $b$ and $s$, so that $k$ should be kept unknown to $A$.
A crucial property of the pseudorandom permutation is that whatever polynomial-time algorithm $A$ performs in the experiment, $A$ cannot distinguish $f_k$ from a uniformly randomly drawn $1$-to-$1$ function unless $A$ knows $k$.
Conventionally, this computational hardness assumption is justified for a large key length $L$ because no known polynomial-time algorithm can distinguish $f_k$ from a truly random choice of $1$-to-$1$ functions.
Note that multiple block-cipher algorithms such as AES are believed to be secure against attacks even by quantum computers as well as classical computers as we increase the key length $L$~\cite{Bonnetain2019,Jaques2020,Jang2022},
while the detailed security analysis of the choice of $L$ for the computational hardness assumption is beyond the scope of this paper and is left for future work.

In place of the uniformly random $1$-to-$1$ function $f_\text{$1$-to-$1$}$, we propose to use the one-way permutation $f_k$ as its replacement $\tilde{f}_\text{$1$-to-$1$}$ in the experiment, i.e.,
\begin{equation}
\label{eq:instantiation_one_to_one}
    \tilde{f}_\text{$1$-to-$1$}(x)\coloneqq f_k(x).
\end{equation}
To construct a replacement $\tilde{f}_\text{$2$-to-$1$}$ for the uniformly random $2$-to-$1$ function $f_\text{$2$-to-$1$}$ satisfying $f_\text{$2$-to-$1$}(x)=f_\text{$2$-to-$1$}(x\oplus s)$,
we use the fact that for any $a<b$, the probability of having
\begin{equation}
    f_\text{$1$-to-$1$}(x)=a,\,f_\text{$1$-to-$1$}(x\oplus s)=b
\end{equation}
and that of
\begin{equation}
    f_\text{$1$-to-$1$}(x\oplus s)=a,\,f_\text{$1$-to-$1$}(x)=b
\end{equation}
are exactly the same for the uniformly random choice of $f_\text{$1$-to-$1$}$.
Thus, using a uniformly random choice of $f_\text{$1$-to-$1$}$, we can also construct a uniformly random $2$-to-$1$ function $f_\text{$2$-to-$1$}$ by
\begin{equation}
    f_\text{$2$-to-$1$}(x)\coloneqq f_\text{$1$-to-$1$}(\min\{x,x\oplus s\}),
\end{equation}
where the order between $x$ and $x\oplus s$ is defined as that in terms of the binary integers.
By replacing $f_\text{$1$-to-$1$}$ with $\tilde{f}_\text{$1$-to-$1$}=f_k$ in~\eqref{eq:instantiation_one_to_one}, we propose to use
\begin{equation}
\label{eq:instantiation_two_to_one}
    \tilde{f}_\text{$2$-to-$1$}(x)\coloneqq f_k(\min\{x,x\oplus s\}).
\end{equation}
With the definitions in~\eqref{eq:instantiation_one_to_one} and~\eqref{eq:instantiation_two_to_one}, we have an instantiation of the function $f$ for Simon's problem that is indistinguishable from the uniformly random choice for $A$ under the computational hardness assumption on the block cipher and is implementable within a polynomial time for $B$.
Note that a polynomial-size quantum circuit for implementing the oracle $O_f$ for the quantum computation can be obtained from the reversible classical logic circuit for the instantiation of $f$, e.g., written in terms of the classical Toffoli gate, by replacing each classical Toffoli gate with the quantum Toffoli gate (and further rewrite each quantum Toffoli gate into Clifford and $T$ gates if we use the Clifford+$T$ universal gate set).

Therefore, combining this instantiation of $f$ with the quantum algorithm analyzed in Corollary~\ref{cor:CorQuantum}, we have an explicit construction of the setup to demonstrate the polynomial amount of energy consumption in quantum computation for solving Simon's problem, as described in~\eqref{eq:E_A_E_B_condition}.

\paragraph*{Estimation of energy consumption of classical computation.}

Finally, we clarify how to estimate the energy consumption of classical computation for solving Simon's problem to be compared with the results of the quantum experiments.
With the above implementation, the quantum experiment can demonstrate the energy-consumption advantage of quantum computation if the measured upper bound~\eqref{eq:measured_upper_bound_energy_consumption} of the energy consumption for the quantum computation is smaller than an estimate for the classical computation.  
However, classical algorithms for solving the same problem cannot be realized due to the exponential runtime and the exponential energy consumption, and thus, the energy consumption of the classical computation cannot be measured directly unlike that of quantum computation; after all, the main point of our analysis is that it is fundamentally infeasible for any classical algorithm to solve this problem while the quantum algorithm can.
To overcome this difficulty, we here propose two approaches for the estimation: one is an extrapolation-based approach, which is more suitable for the near-term demonstration in comparison with the existing best-effort classical algorithm, and the other is a fundamental approach based on our fundamental lower bound in Corollary~\ref{cor:CorClassical}, which may be more demanding yet can demonstrate the advantage for any possible classical algorithm.

The extrapolation-based approach works for a particular choice of classical algorithm.
For example, for fixed $N$, we consider a classical algorithm with querying $M$ different inputs $\vec x=(x_1,\ldots,x_M)$ uniformly at random to the oracle for Simon's problem.
The algorithm works as follows:
\begin{itemize}
    \item if at least one pair $(x_k,x_m)$ satisfying $f(x_k)=f(x_m)$ for some $x_k\neq x_m$ is found in the $M$ queries, the algorithm has an evidence that $f$ is a $2$-to-$1$ function and thus outputs $a=1$; 
    \item  otherwise, the algorithm considers $f$ to be $1$-to-$1$ and outputs $a=0$.
\end{itemize}
For fixed $N$, we can conduct this classical algorithm for any $M\in\{1,\ldots,2^N\}$, but the success probability of finding $(x_k,x_m)$ satisyfing $f(x_k)=f(x_m)$ may be low if $M$ is set to be too small.
By contrast, Refs.~\cite{lecture_note_cleve,Wolf2019} have shown that there exists $M=O(2^{N/2})$ such that this algorithm solves Simon's problem from the uniformly random $M$ queries, i.e., achieves $a=b$ with a high probability as required in Problem~\ref{problem:SimonsProblem}.
However, these existing analyses only show the existence of $M=O(2^{N/2})$ without explicitly clarifying the constant factors therein, and thus, it was unknown how large we should set $M$ in this classical algorithm explicitly.
To establish an explicit criterion on the quantum advantage, we here provide the following proposition with the explicit constant factor.

\begin{restatable}[Refined analysis of upper bound of query complexity of  Simon's problem for a classical algorithm]{proposition}{SimonUpperBound}
\label{prp:simon_achievability}
Fix any $N$.
For any $\delta>0$,
if we have
\begin{equation}
\label{eq:M_choice}
    M=\left\lceil\frac{1}{2}\sqrt{\left(8\ln \frac{1}{\delta}\right)2^{N}+1}+\frac{1}{2}\right\rceil=O(2^{N/2}),
\end{equation}
the classical algorithm that makes the uniformly random $M$ (or more) queries for solving Simon's problem in Problem~\ref{problem:SimonsProblem} achieves the following with a high probability greater than $1-\delta$:
\begin{itemize}
    \item if there exists any $s\in\{0,1\}^N\setminus\{0\}$ such that $f$ is a $2$-to-$1$ function satisfying $f(x)=f(x\oplus s)$, then this $M$-query algorithm outputs $a=1$; 
    \item  if $f$ is a $1$-to-$1$ function, this $M$-query algorithm outputs $a=0$.
\end{itemize}
\end{restatable}
\begin{proof}
The proof is presented in Appendix~\ref{appendix:proofs_classical_achievability}.
\end{proof}

For this classical algorithm, the extrapolation for estimating the energy consumption can be performed in the following procedure.
In the setup of Fig.~\ref{fig:exp_sketch}, we replace the quantum computer of $A$ with a classical computer to perform the classical algorithm and the communication between $A$ and $B$ with classical communication.
For a fixed choice of $N$, conduct the classical algorithm for several different choices of $M$ in the order of
\begin{equation}
\label{eq:M_poly}
    M=O(\mathrm{poly}(N)),
\end{equation}
and measure the energy consumption as in~\eqref{eq:measured_upper_bound_energy_consumption}.
From the measured energy consumption for several $M$, we estimate the leading order of the energy consumption and its constant factor as $M$ increases.
Using this estimation of the function of $M$ representing the energy consumption, perform an extrapolation to obtain an estimate of the energy consumption for $M$ in~\eqref{eq:M_choice} of Proposition~\ref{prp:simon_achievability} with setting $\delta=1/3$ as in Problem~\ref{problem:SimonsProblem}.
For each fixed choice of $N$, use this estimate as the energy consumption of this classical algorithm for solving Simon's problem, so we can compare it with the quantum case.

We note that the goal of this approach is to estimate the scaling and the constant factor of the energy consumption of the classical algorithm as $M$ grows.
In the proposed experiment with polynomial large $M$ in~\eqref{eq:M_poly}, the classical algorithm may not return the correct output; still, one should not try to conduct the classical algorithm for exponentially large $M$ in~\eqref{eq:M_choice} to estimate the energy consumption directly, which is infeasible.
After all, Proposition~\ref{prp:simon_achievability} already guarantees that the algorithm works for $M$ in~\eqref{eq:M_choice}, and the aim of the experiment is not the reproduction of this theoretical result on query complexity.
Rather, our proposal here is to perform only polynomial-time experiments to estimate the energy consumption for $M$ in~\eqref{eq:M_choice} by extrapolation, to obtain a criterion on the advantage in the energy consumption.

Next, the fundamental approach is based on the fundamental lower bound of energy consumption for any classical algorithm shown in Corollary~\ref{cor:CorClassical}.
Problematically, the above extrapolation-based approach can be used for demonstration of the advantage of quantum computation only over the existing classical algorithm used in the experiment, rather than any possible classical algorithm.
For example, apart from the probabilistic classical algorithm analyzed here, Ref.~\cite{CAI201883} provides a deterministic classical algorithm using $M=O(2^{N/2})$ queries to solve Simon's problem, but the bound that we have shown in Proposition~\ref{prp:simon_achievability} may not directly be applicable to this different algorithm; thus, as long as the extrapolation-based approach is used, one may need a case-by-case analysis to determine the constant factor for the extrapolation depending on the choice of classical algorithms.
By contrast, in the fundamental approach here, we will address this issue by using the fundamental lower bound applicable to any classical algorithm.

In particular, for any classical algorithm to solve Simon's problem, we have the explicit lower bound in~\eqref{eq:tight_lower_bound} in the proof of Corollary~\ref{cor:CorClassical}, i.e.,
\begin{align}
\label{eq:explicit_lower_bound}
    \mathcal W &\geq \mathcal W^{(\mathrm{C})} = \frac{1}{\beta} S(C|a)\nonumber\\
    &\geq\frac{1}{\beta}\times\left(\frac{3\sqrt{15}-11}{6\sqrt{15}-18}\times\right.\nonumber\\
    &\left.\left(\sqrt{\frac{8-2\sqrt{15}}{6-\sqrt{15}}}\times 2^{N/2}N \ln 2 - 3\right)-\ln 2\right),
\end{align}
where we have chosen the Boltzmann constant to be $k_B=1$.
With the room temperature in mind, we set
\begin{equation}
    \frac{1}{\beta}=300~\mathrm{K},
\end{equation}
and substituting $1/\beta$ in~\eqref{eq:explicit_lower_bound} with $k_B/\beta$ for $k_B\approx 1\times 10^{-23}$ J/K, we obtain the energy consumption in joule.
With this calculation, typical values of the fundamental lower bounds in~\eqref{eq:explicit_lower_bound} are summarized in Table~\ref{table:1}.
For comparison, we have the following facts~\cite{wiki}.
\begin{enumerate}
    \item Energy of using a $10$-watt flashlight for $1$ minute is $6\times 10^2$~J\@.
    \item Magnetic stored energy in the world's largest toroidal superconducting magnet for the ATLAS experiment at CERN, Geneva, is $1\times 10^{9}$~J\@.
    \item Yearly electricity consumption in the U.S.\ as of 2009 is $1\times 10^{19}$~J\@.
    \item Total energy from the Sun that strikes the face of the Earth each year is $5\times 10^{25}$~J\@.
\end{enumerate}
With an appropriate choice of $N$, if the quantum experiment demonstrates an energy consumption below the corresponding value in Table~\ref{table:1}, our analysis shows that no classical algorithm can outperform such an implementation of the quantum algorithm in terms of energy consumption.

\begin{table}[t]
\begin{tabularx}{\linewidth}{CC}
 \hline\hline
$N$ & $\mathcal W^{(C)}$~(J) \\
\hline
50 & $2\times 10^{-13}$ \\
100 & $1\times 10^{-5}$ \\
150 & $7\times 10^{2}$ \\
200 & $3\times 10^{10}$\\
250 & $1\times 10^{18}$\\
300 & $5\times 10^{25}$\\
 \hline\hline
\end{tabularx}
\caption{Fundamental lower bounds on the energy consumption of any classical algorithm for solving Simon's problem in Problem~\ref{problem:SimonsProblem} with various problem sizes $N$. For small values of $N\lesssim 100$ and room temperature $1/\beta=300\,\mathrm K$, the Boltzmann constant $k_B \approx 1\times 10^{-23}\, \mathrm{J/K}$ is dominant, keeping the lower bound of the energy consumption below $1\,\mathrm J$. For larger $N$, the lower bound of energy consumption rapidly becomes too large to realize in practice for any classical algorithm. Note that apart from this fundamental approach, in the main text, we also discuss an extrapolation-based approach to estimate the lower bound of energy consumption more tightly for some specific choice of the classical algorithm, which is expected to be more suitable for near-term experiments.}\label{table:1}
\end{table}

Although the values in the table may be demanding at least under the current technology, the fundamental approach here has a merit instead.
The conventional experiments for demonstrating quantum computational supremacy in terms of time complexity compare the quantum algorithm with the existing classical algorithm; in this case, even if a quantum algorithm in the experiments outperforms a classical algorithm, another faster classical algorithm may again outperform the implemented quantum algorithm as technology develops further.
In our case of energy-consumption advantage, the extrapolation-based approach also makes a similar type of comparison between the quantum algorithm and the existing classical algorithm.
By contrast, the fundamental approach here is applicable to any possible classical algorithm, which was hard to claim for the conventional analysis of quantum advantage in time complexity.
Therefore, our analysis establishes an ultimate goal of quantum technologies from a new perspective: outperform the fundamental classical lower bound of energy consumption by an experimental realization of a quantum algorithm.

Consequently, our proposal for these experimental demonstrations clarifies how to realize the query-complexity setting (formulated as a learning problem, where a party $A$ is to learn a property of the function $f$ through queries realized as quantum communication with another party $B$, and $B$ is to implement the oracle of $f$). 
Our proposal and analysis also set quantitative goals of quantum technologies for the demonstrations (e.g., the values in Table~\ref{table:1}).
On top of these theoretical developments, accomplishing a demonstration of the energy-consumption advantage of quantum computation over classical computation would be a major milestone in the experiments.

\section{\label{sec:discussion}Discussion and outlook}

In this work, we have formulated a framework for studying the energy consumption of quantum and classical computation.
With this framework, we have derived general upper and lower bounds of the energy consumption of quantum and classical computation, respectively, in the query-complexity setting.

To derive the upper bound of the energy consumption of quantum computation,
our analysis provides a comprehensive account of all the contributions arising from the cost of each operation used for implementing the computation, of initializing the state of the computer at the end of the computer into the initial state, and of the overhead of quantum error correction,
with imperfections of the implementation taken into account.
In particular, to obtain an upper bound of the achievable initialization cost, we have developed a finite-step Landauer-erasure protocol to initialize a given state into a pure state within a finite number of time steps up to finite infidelity, progressing over the existing infinite-step protocols~\cite{Reeb2014,Taranto2021}.
Our analysis indicates that an upper bound of energy consumption of quantum computation is determined by the gate and query complexities and the heat dissipation achievable by the finite-step Landauer erasure in our framework.
These upper bounds rigorously clarify that, even though quantum computation might have larger per-gate energy consumption than classical computation, it is indeed true that polynomial-time quantum computation can be physically realized with a polynomial amount of energy consumption including all the above costs, serving as a widely applicable (yet potentially nontight) estimate of the achievable energy consumption of quantum computation.

As for the lower bound, we developed new techniques to prove a fundamentally nonzero lower bound that works for any classical algorithm regardless of how energy-efficiently the classical computation is implemented.
To establish this fundamental lower bound, we have employed two fundamental physical insights: the energy-conservation law and Landauer's principle.
In particular, working on a query-complexity setting, we develop a technique to argue that a lower bound on the query complexity of any classical algorithm solving a computational problem can be used to establish a lower bound on the Landauer cost---that is, the minimum heat dissipation required to erase the information obtained from the queries, measured in terms of entropy.
Consequently, due to the energy-conservation law, the energy consumption of any classical algorithm solving such a computational problem must be at least the Landauer cost.
Our analysis suggests that the lower bound on energy consumption is governed by this entropic quantity, which depends not only on the query complexity but also on the probability distribution of the oracle in our framework.
Owing to the generality of these techniques, our lower bound holds regardless of technological advances in the energy efficiency of the physical implementation of classical computation.

Using these upper and lower bounds, we rigorously prove that quantum computation can be exponentially advantageous in terms of energy consumption over classical computation for solving a computational problem, in particular, Simon's problem, as the problem size increases.
Since our formulation and analysis are based on the query-complexity setting of Simon's problem, our proof of this exponential quantum advantage is free from the complexity-theoretic assumptions on the hardness of the problem, such as the assumption on the conjectured hardness of integer factoring for classical computation.
We have also proposed a schematic setup of experiments to demonstrate this energy-consumption advantage of quantum computation, using a cryptographic primitive to instantiate and implement the oracle for Simon's problem feasibly.

These results bridge the gap between the complexity-theoretical techniques for studying query complexity and the physical notion of energy that has been studied previously in the field of quantum thermodynamics, opening an alternative way of studying the advantages of quantum computation.
Our framework and the general bounds on energy consumption can be used for studying a potentially broader class of problems; for example, similar to Simon's problem, it is interesting to extend our analysis to other problems in the query-complexity setting, such as the Abelian hidden subgroup problem~\cite{Ettinger2004}, the Bernstein-Vasirani problem~\cite{doi:10.1137/S0097539796300921}, and the problem with verifiable quantum advantage without structure~\cite{9996892}.
Another relevant question would be whether such an energy-consumption advantage of quantum computation could still be proven even if the assumption of a black-box oracle is relaxed and whether the energetic contributions from the oracle can be included in the overall bounds, unlike the conventional analysis in the query-complexity setting.
Also, while our analysis provides an achievable upper bound of energy consumption for quantum computation and a fundamental lower bound for classical computation, it is interesting to develop techniques for proving a lower bound of energy consumption of quantum computation.

Finally, through our analysis, we have proposed a new type of experiment to demonstrate the energy-consumption advantages of quantum computation for Simon's problem, offering an explicit criterion for implementing a quantum algorithm that outperforms any possible classical algorithm in energy consumption (Table~\ref{table:1}).
As a result, our work elevates the significance of Simon’s problem---and potentially a broader class of query-complexity problems---from both theoretical and experimental perspectives; this aligns with the pioneering theoretical studies on quantum computational supremacy~\cite{10.1145/1993636.1993682,10.5555/3135595.3135617,boixo2018characterizing,bouland2019complexity,doi:10.1098/rspa.2008.0443,doi:10.1098/rspa.2010.0301,PhysRevLett.117.080501}, which emphasized the importance of special sampling problems that now receive considerable attention in quantum experiments~\cite{arute2019quantum,doi:10.1126/science.abe8770,bluvstein2024logical}.
Remarkably, unlike the existing experiments on quantum computational supremacy that focus on time complexity, once the proposed quantum experiment satisfies our energy-consumption advantage criterion in Table~\ref{table:1}, our analysis guarantees that no classical algorithm will be able to surpass the quantum algorithm, regardless of technological advances in energy-efficient computation. 
To make this demonstration possible in principle, we have explicitly evaluated the lower bounds of query complexity and energy consumption for any classical algorithm to solve Simon's problem, providing clarification on the constant factors involved.
On the theoretical side, refining these constant factors in our analysis would further reduce the technological requirements for experimental demonstration.
Experimentally demonstrating this fundamental energy-consumption advantage of quantum computation over any classical computation would represent a significant milestone, and our developments have opened a route toward further studies in this direction with a solid theoretical foundation.

\begin{acknowledgments}
The authors thank Jake Xuereb, Parnam (Faraj) Bakhshinezhad, Philip Taranto, Satoshi Yoshida, Natsuto Isogai, Mio Murao, Pauli Erker, and Marcus Huber for the discussion.
Quantum circuits have been generated using the \texttt{quantikz} package~\cite{Kay2018}.
Figure~\ref{fig:exp_sketch} has been created with emojis designed by OpenMoji under the license CC BY-SA 4.0.
F.M. acknowledges funding by the European flagship on quantum technologies (`ASPECTS' consortium 101080167) and funding from the European Research Council (Consolidator grant `Cocoquest’ 101043705).
H.Y. acknowledges JST PRESTO Grant Number JPMJPR201A, JPMJPR23FC, JSPS KAKENHI Grant Number JP23K19970, and MEXT Quantum Leap Flagship Program (MEXT QLEAP) JPMXS0118069605, JPMXS0120351339\@.
\end{acknowledgments}
\clearpage
\newpage

\appendix

\section*{Appendices}

Appendices of the article ``Energy-Consumption Advantage of Quantum Computation'' are organized as follows.
In Appendix~\ref{appendix:proofs}, we provide technical details and proofs for the results in the main text.
In Appendix~\ref{appendix:model_details}, we illustrate the physics of control cost for a specific example and highlight further references on this topic.

\section{\label{appendix:proofs}Proofs}

In this appendix, we provide the proofs of the theorems, propositions, and lemmas present in the main text.
In Appendix~\ref{appendix:proof_cooling}, we prove the finite-step and finite-fidelity Landauer-erasure bound (Theorem~\ref{thm:ThmErasureBound} in the main text) together with an estimate of the relevant energy scales (Proposition~\ref{prop:MaxErasureEnergy}).
Then, in Appendix~\ref{appendix:detailed_energy_consumption}, we provide a line-by-line decomposition and proof of the energy-consumption upper bound (Theorem~\ref{thm:ThmQuantumIdeal} of the main text).
In Appendix~\ref{appendix:proofs_classical_bound}, we provide technical details for the entropic lower bound in Simon's problem (Proposition~\ref{prp:simon}, Lemma~\ref{lemma:QueryNumber}, Proposition~\ref{prop:ClassicalEntropy}, and Lemma~\ref{prp:fraction_evaluation}) used in the proof of Corollary~\ref{cor:CorClassical} in the main text.
Lastly, in Appendix~\ref{appendix:proofs_classical_achievability}, we provide a proof of the achievable upper bound of query complexity of Simon's problem for a classical algorithm that we use in our analysis (Proposition~\ref{prp:simon_achievability}).

\subsection{\label{appendix:proof_cooling}Proof on finite Landauer-erasure bound}

In this appendix, we present the details of the Laudauer-erasure protocol and prove a finite version of the work cost of erasing states of a finite-dimensional target quantum system for this protocol.
In particular, we will prove Theorem~\ref{thm:ThmErasureBound} in the main text and Proposition~\ref{prop:MaxErasureEnergy} shown below.

Regarding the protocol for the Landaure erasure, we consider a protocol based on the one in Ref.~\cite{Reeb2014}.
Given a quantum state $\rho_S$ of the  $d$-dimensional target system to be erased as in Theorem~\ref{thm:ThmErasureBound} in the main text, the Landauer-erasure protocol transforms $\rho_S$ into a final state $\rho_S'$ satisfying the pure-state fidelity $1-\varepsilon = \braket{0|\rho_S'}{0}$;
in particular, the protocol here uses
\begin{align}
\label{eq:final_state_cooling}
    \rho_S' = (1-\varepsilon)\ket{0}\bra{0}+\frac{\varepsilon}{d-1}(\mathds{1}-\ket{0}\bra{0}),
\end{align}
where $\mathds{1}$ is the identity operator acting on the $d$-dimenional space.
For the sake of analysis, we require $\varepsilon\leq1/2$ (to be used in~\eqref{eq:epsilon_bound_required_1},~\eqref{eq:epsilon_bound_required_2}, and~\eqref{eq:epsilon_bound_required_3}).
Importantly, this final state has full rank and therefore satisfies $\mathrm{supp}(\rho_S) \subseteq \mathrm{supp} (\rho_S')$ as required also in Ref.~\cite{Reeb2014}, where $\mathrm{supp}(\rho)$ represents the support of the operator $\rho$.
To define the protocol, one can consider any continuously differentiable path $\rho[u]$ through the space of density operators parameterized by $u\in[0,1]$ with endpoints $\rho(0)=\rho_S$ and $\rho(1)=\rho_S'$.
For any choice of $T$ points $0<u_1<u_2<\cdots <u_T=1$ in $[0,1]$, 
the environment is composed of $T$ $d$-dimensional subsystems and has the total Hamiltonian
\begin{align}
\label{eq:H_E}
    H_E \coloneqq \sum_{t=1}^T H_E^{(t)}\otimes\mathds{1}_{\text{rest}\neq t}\geq 0,
\end{align}
where the Hamiltonian of the $t$th subsystem ($t=1,\ldots,T$) is given by
\begin{equation}
\label{eq:H_E_k}
    H_E^{(t)}\coloneqq-\frac{1}{\beta}\ln \rho[u_t]\geq 0,
\end{equation}
and $\mathds{1}_{\text{rest}\neq t}$ is the identity operator acting on all the subsystems but the $t$th.
For each $t$,
we write the thermal state of the $t$th subsystem of the environment as
\begin{equation}
    \tau^{(t)}_E[\beta]\coloneqq \frac{e^{-\beta H_E^{(t)}}}{\tr[e^{-\beta H_E^{(t)}}]}=\rho[u_t],
\end{equation}
where the last equality holds by definition~\eqref{eq:H_E_k}.
Then, the thermal state of the environment $\tau_E[\beta]$ in~\eqref{eq:tau_E_beta} of the main text becomes
\begin{equation}
    \tau_E[\beta]=\bigotimes_{t=1}^T \tau^{(t)}_E[\beta]=\bigotimes_{t=1}^T \rho[u_t].
\end{equation}
The protocol is composed of $T$ steps; in the $t$th step for $t=1,\ldots, T$, the protocol swaps the target system's state with the thermal state $\tau_E^{(t)}[\beta]$ of the $t$th subsystem of the environment, so that the target system's state after the $t$th step should be $\rho[u_t]$ along the path $\rho[u]$.
For this protocol, Ref.~\cite{Reeb2014} shows that the heat dissipation into the environment is given by
\begin{align}\label{eq:heat_sum}
    \beta \mathcal Q_E = \sum_{t=1}^T \tr\Big[\big(\rho[u_t]-\rho[u_{t-1}]\big)\ln\left(\rho[u_t]\right)\Big].
\end{align}
For the sake of analysis, our protocol chooses a straight path $\rho[u]$ through the state space and equally spaced points on the path, i.e.,
\begin{align}
    \label{eq:rho_u}
    \rho[u] &= \rho_S + u(\rho_S'-\rho_S),\\
    \label{eq:u_t}
    u_t &= \frac{t}{T}.
\end{align}

With this protocol, we prove Theorem~\ref{thm:ThmErasureBound} in the main text, which is repeated again below.
In particular, we progress beyond the asymptotic bound shown in Ref.~\cite{Taranto2021} and analyze in detail how the number of steps an erasure protocol requires scales as a function of the final pure-state infidelity $\varepsilon$ and the difference of the actual heat dissipation $\eta$ to Landauer's ideal bound.

\ThmErasureBound*\begin{proof}
We will prove that, for a generic number of steps $T$, the protocol described above to transform $\rho_S$ into $\rho_S'$ in~\eqref{eq:final_state_cooling}, which is $\varepsilon$-close to the pure state $\ket{0}$, has a heat dissipation $\mathcal Q_E$ to the environment bounded by
\begin{align}
\label{eq:generic_bound}
    \beta\mathcal Q_E-\Delta S
    &\leq\frac{\ln\qty(\frac{e(d-1)^2 T}{\varepsilon})}{T},
\end{align}
so that setting $T$ as in the theorem statement should achieve $\beta\mathcal Q_E - \Delta S\leq \eta$.

\begin{figure}
    \centering
    \includegraphics[width=3.4in]{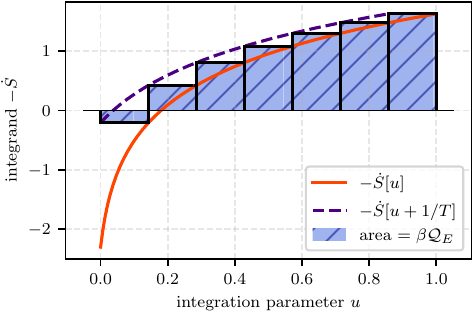}
    \caption{Illustration visualizing the arguments used for deriving the upper and lower bound on $\beta\mathcal Q_E$ in~\eqref{eq:QE_geq_S} and~\eqref{eq:QE_leq_Splus} for $T=7$. The sum on the right-hand side of~\eqref{eq:heat_sum} making up $\beta\mathcal Q_E$ is represented through the area below the step function in blue.
    This step function is lower bounded by the integral $\int_0^1 du(- \dot S[u])$ due to the monotonicity of the integrand (solid red curve), which yields to~\eqref{eq:QE_geq_S}.
    An upper bound is given by the shifted integral given by the dashed violet curve, together with a correction of the last step on the right of the interval of the integral, leading to~\eqref{eq:QE_leq_Splus}.}
    \label{fig:integral_sketch}
\end{figure}

To bound $\mathcal{Q}_E$ in~\eqref{eq:heat_sum}, we write,
for simplicity of notation,
\begin{equation}
    S[u] \coloneqq S(\rho[u])=-\tr[\rho[u]\ln(\rho[u])].
\end{equation}
By using the shorthand $\dot S[u] \coloneqq \frac{\partial}{\partial u} S[u]$, it holds that
\begin{align}
\dot S[u] &= - \tr[\dot \rho[u]\times (\ln(\rho[u])+\mathds{1})]\\
&=- \tr[(\rho_S'-\rho_S)\times (\ln(\rho[u])+\mathds{1})]\\
&=- \tr[(\rho_S'-\rho_S)\ln(\rho[u])].
\end{align}
Using this notation, each term on the right-hand side of~\eqref{eq:heat_sum} simplifies into
\begin{align}
    \tr[(\rho[u_t]-\rho[u_{t-1}])\ln(\rho[u_t])]&=\tr[\frac{\rho'_S-\rho_S}{T}\ln(\rho[u_t])]\\
    &=\frac{-\dot S[u_t]}{T}.
\end{align}
Moreover, $S[u]$ is a concave function; i.e., for any $u_1,u_2$, and $\lambda\in[0,1]$, we have
\begin{align}
    S[\lambda u_1 + (1-\lambda) u_2]
    &= S(\rho[\lambda u_1 + (1-\lambda) u_2]) \\
    &\hspace{-0.15cm}\stackrel{\text{\eqref{eq:rho_u}}}{=}\hspace{-0.15cm} S(\lambda \rho[u_1] + (1-\lambda)\rho[u_2]) \\
    &\geq \lambda S(\rho[u_1]) + (1-\lambda)S(\rho[u_2]) \\
    &= \lambda S[u_1] + (1-\lambda)S[u_2].
\end{align}
Thus, $-\dot S[u]$ is a monotonically non-decreasing function as $u$ increases.
In Fig.~\ref{fig:integral_sketch}, we illustrate an example of a monotonically non-decreasing function and how we use this property to derive upper and lower bounds of $\beta \mathcal Q_E$.
By evaluating the integral according to this figure,~\eqref{eq:heat_sum} can be bounded from below by
\begin{align}
    \beta\mathcal{Q}_E&=\sum_{t=1}^T\frac{-\dot S[u_t]}{T}\\
    \label{eq:upper_S}
    &\geq\int_0^1 du\,(-\dot S[u])\\
    &=S[0]-S[1]\\
    &=\Delta S\label{eq:QE_geq_S}
\end{align}
and from above by
\begin{align}
    \beta\mathcal{Q}_E&=\sum_{t=1}^T\frac{-\dot S[u_t]}{T}\\
    \label{eq:lower_S}
    &\leq\int_{1/T}^1 du\,\qty(-\dot S\qty[u])+\frac{-\dot S[1]}{T}\\
    &=\int_{0}^1 du\,\qty(-\dot S\qty[u])-\int_0^{1/T} du\,\qty(-\dot S\qty[u])+\frac{-\dot S[1]}{T}\\
    &\leq\Delta S+\qty|S[0]-S\qty[\frac{1}{T}]|+\frac{-\dot S[1]}{T}.\label{eq:QE_leq_Splus}
\end{align}
Therefore, we can rewrite~\eqref{eq:heat_sum} as
\begin{align}\label{eq:heat_dissipation_integral_estimate}
    0\leq\beta\mathcal Q_E - \Delta S \leq \qty|S[0]-S\qty[\frac{1}{T}]| - \frac{1}{T}\dot S[1].
\end{align}

We bound the right-hand side of~\eqref{eq:heat_dissipation_integral_estimate}.
The first term on the right-hand side can be estimated using the continuity bound of the quantum entropy in Refs.~\cite{Zhang2007,Audenaert2007,Petz2008,Winter2016}.
In particular, 
let
\begin{equation}
    \delta\coloneqq\frac{1}{2}\norm{\rho[0]-\rho\qty[\frac{1}{T}]}_1 
\end{equation}
denote the trace distance, and
\begin{equation}
    h(x)\coloneqq\begin{cases}-x\ln(x)-(1-x)\ln(1-x),&\text{$x\in(0,1)$},\\
    0&\text{$x=0,1$}
    \end{cases}
\end{equation}
the binary entropy.
Since we have
\begin{equation}
    \rho[0]-\rho\qty[\frac{1}{T}]=\frac{\rho_S'-\rho_S}{T},
\end{equation}
it holds that
\begin{align}
\label{eq:delta_bound_T}
    \delta=\frac{1}{2T}\|\rho_S'-\rho_S\|_1\leq \frac{1}{T}, 
\end{align}
where the last inequality follows from $\frac{1}{2}\|\rho'-\rho\|_1\in[0,1]$ for all density operators $\rho$ and $\rho'$.
Thus, for any $T\geq2$, it holds that
\begin{equation}
\label{eq:delta_bound}
    \delta\leq\frac{1}{2}\leq1-\frac{1}{d},
\end{equation}
where the last inequality holds due to $d\geq 2$.
For $\delta$ satisfying~\eqref{eq:delta_bound},
the continuity bound in Refs.~\cite{Zhang2007,Audenaert2007,Petz2008,Winter2016} leads to
\begin{align}
    \qty|S[0]-S\qty[\frac{1}{T}]|&=\qty|S(\rho[0])-S\qty(\rho\qty[\frac{1}{T}])|\\
\label{eq:continuity_bound}
    &\leq \delta \ln(d-1)+h(\delta).
\end{align}
For any $T\geq 2$, it also follows from~\eqref{eq:delta_bound_T} that
\begin{align}
    h(\delta)&\leq h\qty(\frac{1}{T})\\
\label{eq:h_bound}
    &\leq\frac{1+\ln T}{T},
\end{align}
where the last inequality follows from
\begin{equation}
    h(x)\leq x\left(1+\ln \frac{1}{x}\right)
\end{equation}
for all $x>0$.
Therefore, for any $T\geq 2$, we obtain from~\eqref{eq:delta_bound_T},~\eqref{eq:continuity_bound}, and~\eqref{eq:h_bound}
\begin{align}
\label{eq:entropy_difference_T}
    &\qty|S[0]-S\qty[\frac{1}{T}]|\leq\frac{\ln (d-1)}{T}+\frac{1+\ln T}{T}.
\end{align}
To bound the second term of the right-hand side of~\eqref{eq:heat_dissipation_integral_estimate}, we use $\rho_S'$ in~\eqref{eq:final_state_cooling} to evaluate
\begin{align}
   -\dot S[1]&=\tr[(\rho_S'-\rho_S)\ln\rho_S']\\
   &=-S(\rho_S')-\tr[\rho_S\ln\rho_S']\\
   &=-S(\rho_S')-\bra{0}\rho_S\ket{0}\ln(1-\varepsilon)\nonumber\\
   &\quad-\sum_{j=1}^{d-1}\bra{j}\rho_S\ket{j}\ln\qty(\frac{\varepsilon}{d-1})\\
   &\leq 0-\bra{0}\rho_S\ket{0}\ln(1-\varepsilon)-\sum_{j=1}^{d-1}\bra{j}\rho_S\ket{j}\ln\qty(\frac{\varepsilon}{d-1})\\
   &\leq-\sum_{j=0}^{d-1}\bra{j}\rho_S\ket{j}\ln\qty(\frac{\varepsilon}{d-1})\label{eq:bound_log_trick}\\
   &=\ln\qty(\frac{d-1}{\varepsilon}),
   \label{eq:bound_S_1}
\end{align}
where~\eqref{eq:bound_log_trick} follows from
\begin{equation}
\label{eq:epsilon_bound_required_1}
    -\ln(1-\varepsilon)\leq-\ln\qty(\frac{\varepsilon}{d-1})
\end{equation}
for all $\varepsilon\leq 1/2$ and $d\geq 2$, and~\eqref{eq:bound_S_1} uses normalization of $\tr[\rho_S]=1$ and the identity $-\log(x)=\log(1/x)$ for $x>0$.
As a whole, applying~\eqref{eq:entropy_difference_T} and~\eqref{eq:bound_S_1} to~\eqref{eq:heat_dissipation_integral_estimate} yields
\begin{align}
    \beta Q_E - \Delta S &\leq \frac{\ln(d-1)}{T}+\frac{1+\ln T}{T}+\frac{\ln\qty(\frac{d-1}{\varepsilon})}{T}\\
    &=\frac{\ln\qty(\frac{e(d-1)^2 T}{\varepsilon})}{T},
    \label{eq:Q_E_S_bound}
\end{align}
which provides the bound~\eqref{eq:generic_bound} on the heat dissipation $\mathcal{Q}_E$ for a generic number of steps $T\geq 2$.

To complete the proof, we set $T$ in~\eqref{eq:Q_E_S_bound} in such a way that we should obtain $\beta\mathcal{Q}_E - \Delta S\leq \eta$.
We write~\eqref{eq:Q_E_S_bound} as
\begin{equation}
    f(T)\coloneqq\frac{\ln\qty(\frac{e(d-1)^2 T}{\varepsilon})}{T}.
\end{equation}
Given the fact that $\ln x/x$ does not increase for all $x\geq e$,
we see that $f(T)$ is a monotonically non-increasing function for all $T\geq \varepsilon/(d-1)^2$, which always holds here due to $T\geq 1$, $\varepsilon\leq 1/2$, and $d\geq 2$.
An exact solution of
\begin{equation}
    f(T)=\eta
\end{equation}
with respect to $T$ would require the Lambert $W$ function, which cannot be expressed in a closed form through elementary functions~\cite{Chow1999,Valluri2000}.
But for our analysis, finding $T$ satisfying
\begin{equation}
    f(T)\leq\eta
\end{equation}
within the order of
\begin{equation}
    T=O\qty(\frac{1}{\eta}\log\frac{d^2}{\varepsilon\eta})
\end{equation}
is sufficient; to provide such $T$, let us choose $T$ as
\begin{align}\label{eq:number_T2}
    T = \left\lceil \frac{(e+1)}{e\eta}\ln\qty(\frac{(e+1)(d-1)^2}{\varepsilon\eta})\right\rceil.
\end{align}
Then, we have
\begin{align}
    &f(T)\\
    &\leq f\qty(\frac{e+1}{e\eta}\ln\qty(\frac{(e+1)(d-1)^2}{\varepsilon\eta}))\\
    &=\frac{e\eta}{e+1}\qty(\frac{\ln\qty(\frac{(e+1)(d-1)^2}{\varepsilon\eta})}{\ln\qty(\frac{(e+1)(d-1)^2}{\varepsilon\eta})}+\frac{\ln\qty(\ln\qty(\frac{(e+1)(d-1)^2}{\varepsilon\eta}))}{\ln\qty(\frac{(e+1)(d-1)^2}{\varepsilon\eta})})\\
    &\leq\frac{e\eta}{e+1}\qty(1+\frac{1}{e})\\
    &=\eta,
    \label{eq:eta}
\end{align}
where the first inequality follows from the monotonicity of $f$, the second inequality follows from
\begin{equation}
    \frac{\ln x}{x}\leq \frac{1}{e}
\end{equation}
for all $x>0$.
For all $\varepsilon\leq 1/2$, $\eta\leq 1$, and $d\geq 2$,
the choice of $T$ in~\eqref{eq:number_T2} indeed satisfies
\begin{align}
    T&\geq \frac{e+1}{e\times 1}\ln\qty(\frac{(e+1)(2-1)^2}{\frac{1}{2}\times 1})\\
    &=\frac{e+1}{e}\ln(2(e+1))\\
\label{eq:epsilon_bound_required_2}
    &\geq 2.
\end{align}
Note that $\eta\leq 1$ is not necessarily the tightest upper bound to derive the last line and thus to show the theorem here. 
Therefore, we obtain from~\eqref{eq:Q_E_S_bound} and~\eqref{eq:eta}
\begin{equation}
    \beta\mathcal{Q}_E - \Delta S\leq \eta,
\end{equation}
which yields the conclusion.
\end{proof}

In Proposition~\ref{prop:MaxErasureEnergy}, we analyze the energetics involved in the finite Landauer-erasure protocol in Theorem~\ref{thm:ThmErasureBound}. In particular, we are interested how the energy scale of the Hamiltonian $H_E$ of the environment in~\eqref{eq:H_E} behaves as a function of the required infidelity $\varepsilon$ between the final state and the pure state and the required steps $T$ of the protocol in Theorem~\ref{thm:ThmErasureBound}.
Note that we have $T\geq 2$ in Theorem~\ref{thm:ThmErasureBound} due to~\eqref{eq:epsilon_bound_required_2}.

\begin{restatable}[Maximum energy of environment in finite Laudauer-eraser protocol]{proposition}{MaxErasureEnergy}\label{prop:MaxErasureEnergy}
    Let $H_E$ be the Hamiltonian  in~\eqref{eq:H_E}. The maximum energy eigenvalue of $H_E$ is bounded by
    \begin{align}
        \|H_E\|_\infty &\leq\frac{T}{\beta}\qty(\ln\frac{d-1}{\varepsilon}+1)-\frac{1}{\beta}\qty(\frac{\ln 2\pi T}{2}+\frac{1}{12T+1})\\
    &=O\qty(\frac{T}{\beta}\ln\frac{d}{\varepsilon}),
    \end{align}
    where $T\geq 2$ is the number of steps of the Landauer-erasure protocol in Theorem~\ref{thm:ThmErasureBound}, $d\geq 2$ is the dimension of the target system for the Landauer-erasure protocol, and $\varepsilon\in(0,1/2]$ is the infidelity between the final state and the pure state in~\eqref{eq:final_state_cooling}.
\end{restatable}

\begin{proof}
Since $H_E$ is a sum over terms $H_E^{(t)}\otimes\mathds{1}_{\text{rest}\neq t}$ as shown in~\eqref{eq:H_E}, the maximum eigenvalue of $H_E$ is bounded by the sum of the maxima over the Hamiltonian $H_E^{(t)}$ of  each individual subsystem, i.e.,
\begin{equation}
    \|H_E\|_\infty \leq \sum_{t=1}^T\|H_E^{(t)}\otimes\mathds{1}_{\text{rest}\neq t}\|_\infty=\sum_{t=1}^T\|H_E^{(t)}\|_\infty.
\end{equation}
Due to the definition of $H_E^{(t)}$ in~\eqref{eq:H_E_k}, we have
\begin{equation}
    \|H_E^{(t)}\|_\infty=\norm{-\frac{1}{\beta}\ln\rho[u_t]}_\infty=-\frac{1}{\beta}\ln\lambda_{\min} (\rho[u_t]),
\end{equation}
where we write the minimum eigenvalue of a density operator $\rho$ as $\lambda_{\min}(\rho)$.
Thus, we have 
\begin{align}
    \|H_E\|_\infty\leq-\frac{1}{\beta}\sum_{t=1}^T \ln \lambda_{\min} (\rho[u_t]).
\end{align}

The problem reduces to bounding $\lambda_{\min} (\rho[u_t])$.
As shown in~\eqref{eq:rho_u} and~\eqref{eq:u_t},
we have
\begin{equation}
    \rho[u_t]=\qty(1-\frac{t}{T})\rho_S+\frac{t}{T}\rho_S'.
\end{equation}
It generally holds that the minimum eigenvalue of the sum of two positive semi-definite Hermitian matrices is bounded from below by the sum of the smallest two eigenvalues of the summands \cite{Weyl1912,Johnson1989,Fulton2000}. Since $\rho\geq 0$ is not further specified, we here bound the minimum eigenvalue of $\rho[u_t]$ from below by that of $\rho_S'$, i.e.,
\begin{align}
    \lambda_{\min} (\rho[u_t])&\geq\qty(1-\frac{t}{T})\lambda_{\min} (\rho_S)+\frac{t}{T}\lambda_{\min} (\rho_S')\\
    &\geq0+\frac{t}{T}\lambda_{\min} (\rho_S')\\
    &=\frac{t}{T}\frac{\varepsilon}{d-1},
\end{align}
where the last line follows from the fact that the eigenvalues of $\rho_S'$ in~\eqref{eq:final_state_cooling} is $(1-\varepsilon)$ and $\varepsilon/(d-1)$, and we have
\begin{equation}
    \label{eq:epsilon_bound_required_3}
    1-\varepsilon\geq\frac{\varepsilon}{d-1}
\end{equation}
for all $\varepsilon\in(0,1/2]$ and $d\geq 2$.
As a result, summing over all the contributions, we arrive at the following upper bound for the maximum eigenvalue of $H_E$
\begin{align}
    \|H_E\|_\infty &\leq-\frac{1}{\beta} \sum_{t=1}^T\ln\frac{t\varepsilon}{T (d-1)}\\
    &=\frac{T}{\beta}\ln\frac{T (d-1)}{\varepsilon}-\frac{1}{\beta}\ln T!\\
    &\leq\frac{T}{\beta}\ln\frac{T (d-1)}{\varepsilon}\nonumber\\
    &\quad-\frac{1}{\beta}\qty(T\ln T-T +\frac{\ln 2\pi T}{2}+\frac{1}{12T+1})\\
    &=\frac{T}{\beta}\qty(\ln\frac{d-1}{\varepsilon}+1)-\frac{1}{\beta}\qty(\frac{\ln 2\pi T}{2}+\frac{1}{12T+1})\\
    &=O\qty(\frac{T}{\beta}\ln\frac{d}{\varepsilon}),
\end{align}
where the second inequality follows from a variant of Stirling's approximation~\cite{10.2307/2308012}
\begin{equation}
    T!\geq \sqrt{2\pi T}\left(\frac{T}{e}\right)^Te^{\frac{1}{12T+1}}.
\end{equation}
This concludes the proof.
\end{proof}

\subsection{\label{appendix:detailed_energy_consumption}Detail of proof of general upper bound on energy consumption}
In this appendix, we expand upon the expressions for the energy consumption of quantum algorithms derived in Sec.~\ref{sec:quantum_upperbound}. In particular, we prove Theorem~\ref{thm:ThmQuantumIdeal} repeated again below.
\begin{widetext}
\ThmQuantumIdeal*    
\end{widetext}

\begin{proof}
Proving the statement essentially boils down to summing up all the terms in~\eqref{eq:proof_energy_conservation_2}.
To this end, we use the expressions derived in steps~\ref{item:qstep_1} and \ref{item:qstep_2} for the gate and erasure costs and replace them line-by-line.
We also note that the contributions from energetic costs sum up to zero, and hence, in the overall energy consumption, the terms that matter are the control cost and the initialization cost.

As a prerequisite, we give some more detailed expressions for the cost of the erasure in~\eqref{eq:W_gate_E}.
We start by recalling that the number of steps $T$ in the Landauer-erasure protocol for each qubit (i.e., for dimension $d=2$) is given, due to Theorem~\ref{thm:ThmErasureBound}, by
\begin{align}\label{eq:T_asymptotic}
    T = O\left(\frac{1}{\eta}\log\frac{1}{\varepsilon\eta}\right),
\end{align}
as $\eta,\varepsilon\rightarrow 0$.
This expression will be used for expressing all the erasure-related costs in terms of $\varepsilon$ and $\eta$.
First of all, the circuit depth~\eqref{eq:D_E} of the erasure part $\mathcal{E}$ in Fig.~\ref{fig:generic_algo} can be expressed as
\begin{align}
    D_\mathcal{E} \leq O(T) = O\left(\frac{1}{\eta}\log\frac{1}{\varepsilon\eta}\right),
    \label{eq:D_E_asymptotic}
\end{align}
where the swap gate is implemented within a constant depth in our setting.
In addition, the environment control cost in~\eqref{eq:environment_control_cost} can be expressed as
\begin{align}
    E_\text{ctrl.}^{(E)} &\stackrel{\text{\eqref{eq:environment_control_cost}}}{=} O\qty(\frac{T}{\beta}\log\frac{1}{\varepsilon}) \stackrel{\text{\eqref{eq:T_asymptotic}}}{=} O\left(\frac{1}{\beta\eta}\mathrm{polylog}\frac{1}{\varepsilon\eta}\right).\label{eq:E_ctrl_E_asymptotic}
\end{align}
All these expressions together can then be inserted into the control cost for the erasure, $E_\mathrm{ctrl.}^{(\mathcal E)}$ to give
\begin{widetext}
\begin{align}
    \mathcal W_\text{gate}^{(\mathcal E)}-\Delta E^{(\mathcal E)}&=E_\mathrm{ctrl.}^{(\mathcal E)} + \mathcal Q_E \\
    &\hspace{-0.4cm}\stackrel{\text{\eqref{eq:Q_E},\eqref{eq:E_ctrl_E}}}{\leq} \left(E_\text{ctrl.}+E_\text{ctrl.}^{(E)}\right)\times WD_\mathcal{E} + \frac{W}{\beta}(\ln 2 + \eta) \\
    &\hspace{-0.62cm}\stackrel{\text{\eqref{eq:D_E_asymptotic},\eqref{eq:E_ctrl_E_asymptotic}}}{\leq} \left(E_\mathrm{ctrl.}+ O\left(\frac{1}{\beta\eta}\mathrm{polylog}\frac{1}{\varepsilon\eta}\right)\right)\times W\cdot O\left(\frac{1}{\eta}\log\frac{1}{\varepsilon\eta}\right)
    + \frac{W}{\beta}(\ln 2 + \eta) \\
    &= O\left(\left(E_\mathrm{ctrl.} + \frac{1}{\beta \eta}\right)\times\frac{W}{\eta}\mathrm{polylog}\frac{1}{\varepsilon\eta}\right),\label{eq:W_gate_E_derivation_final}
\end{align}
which leads to the bound in~\eqref{eq:W_gate_E_bound}.
Finally,
using~\eqref{eq:proof_energy_conservation_2}, we obtain
    \begin{align}
        \mathcal W &= \left(\sum_{k=1}^{M+1}E_\mathrm{ctrl.}^{(U_k)}+\sum_{k=1}^{M}E_\mathrm{ctrl.}^{(\text{in},k)}+E_\text{ctrl.}^\text{(out)}+ E_\mathrm{ctrl.}^{(\mathcal E)}\right) + \mathcal Q_E \\
        &\hspace{-0.7cm}\stackrel{\text{\eqref{eq:W_gate_U_k},\eqref{eq:W_gate_in_k},\eqref{eq:W_gate_out_k}}}{=}\sum_{k=1}^{M+1}O\left(E_\mathrm{ctrl.} \times W D_k\right) + \sum_{k=1}^{M}O\left(E_\mathrm{ctrl.} \times W  D_{\rm in}\right) + O\left(E_\mathrm{ctrl.} \times W D_{\rm out}\right) + E_\mathrm{ctrl.}^{(\mathcal E)} + \mathcal Q^{(\mathcal{E})} \\
        &=\underbrace{O\left(E_\mathrm{ctrl.} \times W \left(\sum_{k=1}^{M+1}D_k + MD_{\rm in} + D_{\rm out} \right)\right)}_{\text{all costs for all $U_k$, input, and output}}+E_\mathrm{ctrl.}^{(\mathcal E)} + \mathcal Q^{(\mathcal{E})}\\
        &\hspace{-0.2cm}\stackrel{\text{\eqref{eq:W_gate_E_derivation_final}}}{=} O\left(E_\mathrm{ctrl.} \times W D_{\rm comp} +\left(E_\mathrm{ctrl.} + \frac{1}{\beta \eta}\right)\times\frac{W}{\eta}\mathrm{polylog}\frac{1}{\varepsilon\eta}\right),
        \label{eq:WQ_derivation}
    \end{align}
\end{widetext}
where the last line follows from the definition of $D_{\rm comp}$ in~\eqref{eq:D_U}.
This concludes the proof.
\end{proof}

\subsection{\label{appendix:proofs_classical_bound}Proof on the classical lower bound on energy consumption for solving Simon's problem}

In this appendix, we prove Proposition~\ref{prp:simon}, Lemma~\ref{lemma:QueryNumber}, Proposition~\ref{prop:ClassicalEntropy}, and Lemma~\ref{prp:fraction_evaluation} used for establishing a lower bound on the energy consumption of all classical algorithms for solving Simon's problem analyzed in Sec.~\ref{sec:classical_simon}.

We start with providing the proof of Proposition~\ref{prp:simon} repeated again below for readability.
The proposition shows that a lower bound of the query complexity of Simon's problem in Problem~\ref{problem:SimonsProblem} is given for problem size $N$ and average success probability $1/2+\Delta$ by~\eqref{eq:M_N_delta}, i.e.,
\begin{equation}
    M_{N,\Delta}\coloneqq\left\lceil\sqrt{\frac{2\Delta}{1+\Delta}}\times2^{N/2}\right\rceil,
\end{equation}
where $\Delta\in(0,1/2)$ is any constant in this parameter region and is suitably tunable for our analysis.
Note that Ref.~\cite{lecture_note_cleve} have also shown a bound on $M$ to achieve the average success probability $3/4$ in Simon's problem, but for our analysis of Problem~\ref{problem:SimonsProblem}, especially for the proof of Lemma~\ref{lemma:QueryNumber}, it is critical to obtain the bound on $M$ for an average success probability between $1/2$ and $2/3$, as with Proposition~\ref{prp:simon} with parameter chosen as $\Delta<1/6$.

\SimonBound*

\begin{proof}
As in Refs.~\cite{Simon1994, Simon1997}, given some fixed input sequence $\vec x=(x_1,\ldots,x_M)$ of length $M$ in~\eqref{eq:vec_x},
we call $\vec x$ \textit{good} if $x_k=x_m\oplus s$ for some $1\leq k<m\leq M$; otherwise, we call $\vec x$ \textit{bad}.
For a fixed input sequence $\vec x$, let
\begin{equation}
\label{eq:p_good}
    p_\mathrm{good}
\end{equation}
denote the probability of $\vec x$ being good; i.e., that being bad is $1-p_\mathrm{good}$.
To succeed in estimating $b$ with probability greater than $2/3$ as required in Problem~\ref{problem:SimonsProblem},
any classical algorithm in our framework of Fig.~\ref{fig:generic_algo} needs to query the function oracle $O_f$ repeatedly at least until the input sequence $\vec x$ becomes good with some probability sufficiently larger than zero.
When $\vec x$ is good and $b=1$, i.e., $f$ is 2-to-1, the classical algorithm can find elements of $\vec x$ satisfying $x_k=x_m\oplus s$ for some $1\leq k<m\leq M$ to conclude that $f$ is 2-to-1\@.
But we note that even if $\vec x$ is good, in the case of $b=0$, i.e., the case where $f$ is 1-to-1\@, the $M$ different inputs $x\in\{0,1\}^N$ still yield $M$ distinct outputs $y=f(x)$.
Conditioned on having a good input sequence, the success probability of the classical algorithm estimating $b$ correctly has a trivial upper bound $1$.
For a bad input sequence $\vec x$, the $M$ input-output relations $(x_1,f(x_1)),\ldots,(x_M,f(x_M))$ are identical regardless of $b=0$ or $b=1$ (i.e., regardless of $f$ being 1-to-1 or 2-to-1).
Thus, conditioned on having a bad input sequence, the success probability of the classical algorithm estimating $b$ correctly cannot exceed $1/2$, i.e., a completely random guess.
Therefore, the overall success probability of the classical algorithm estimating $b$ correctly is upper bounded by
\begin{equation}
\label{eq:success_prob_simon}
    p_\mathrm{good}\times 1+(1-p_\mathrm{good})\times \frac{1}{2}=\frac{1}{2}+\frac{p_\mathrm{good}}{2}.
\end{equation}

As shown in Ref.~\cite{lecture_note_cleve} (see also Refs.~\cite{Simon1994,Simon1997}), the probability of having a good input sequence $\vec x=(x_1,\ldots,x_M)$ at the $M$th query, conditioned on the case where the input sequence at the $(M-1)$th query is bad, is at most
\begin{equation}
    \frac{M-1}{2^N-1-\frac{(M-2)(M-1)}{2}}.
\end{equation}
Following Ref.~\cite{lecture_note_cleve},  for any $M$, we bound
\begin{equation}
    \frac{M-1}{2^N-1-\frac{(M-2)(M-1)}{2}}\leq\frac{2M}{2^{N+1}-M^2}.
\end{equation}
Thus, for any sequence $\vec x$ obtained from $M$ queries, the probability of $\vec x$ being good is at most
\begin{align}
    p_\mathrm{good}&\leq\sum_{m=1}^{M}
    \frac{m-1}{2^N-1-\frac{(m-2)(m-1)}{2}}\\
    &\leq\sum_{m=1}^{M}\frac{2M}{2^{N+1}-M^2}\\
    &=\frac{2M^2}{2^{N+1}-M^2}.
\end{align}
Since $M$ is an integer, the bound~\eqref{eq:M_N_delta}, i.e.,
\begin{equation}
    M<\left\lceil\sqrt{\frac{2\Delta}{1+\Delta}}\times2^{N/2}\right\rceil
\end{equation}
holds if and only if we have
\begin{equation}
    M<\sqrt{\frac{2\Delta}{1+\Delta}}\times2^{N/2}.
\end{equation}
Therefore, for any $M$ satisfying~\eqref{eq:M_N_delta}, we have
\begin{equation}
    p_\mathrm{good}<\frac{2{\left(\sqrt{\frac{2\Delta}{1+\Delta}}\times2^{N/2}\right)}^2}{2^{N+1}-{\left(\sqrt{\frac{2\Delta}{1+\Delta}}\times2^{N/2}\right)}^2}=2\Delta.
\end{equation}
Due to~\eqref{eq:success_prob_simon}, the average success probability of the classical algorithm is at most
\begin{equation}
    \frac{1}{2}+\frac{p_\mathrm{good}}{2}<\frac{1}{2}+\Delta.
\end{equation}
\end{proof}

Lemma~\ref{lemma:QueryNumber} below aims to provide a high-probability lower bound on the required number of queries to the oracle for all classical algorithms to solve Simon's problem.
We have shown in Proposition~\ref{prp:simon} that, for any classical algorithm with the number of queries smaller than $M_{N,\Delta}$, its success probability in solving Simon's problem in Problem~\ref{problem:SimonsProblem} is upper bounded by $1/2+\Delta$.
For our purpose, however, we need to invert this statement; in words, what we need for our analysis is that to achieve a fixed success probability $2/3$, the probability that the oracle is queried at least $M\geq M_{N,\Delta}$ times should be reasonably high.
Intuitively, if the algorithm stops with $M<M_{N,\Delta}$ too often, the algorithm cannot gain enough information from the queries to successfully solve Simon's problem with probability $2/3$.
In particular, let us write the probability of the algorithm performing strictly less than $M_{N,\Delta}$ queries as
\begin{align}
    \label{eq:p_M_leq_2N4}
    p\left(M<M_{N,\Delta}\right) = \sum_{M<M_{N,\Delta}}p(M),
\end{align}
and that of more than or equal to $M_{N,\Delta}$ queries
\begin{align}
    \label{eq:p_M_geq_2N4}
    p\left(M\geq M_{N,\Delta}\right)
     = \sum_{M\geq M_{N,\Delta}}p(M),
\end{align}
where $p(M)$ is defined in~\eqref{eq:p_M} as the probability distribution of the algorithm querying the oracle exactly $M$ times.
Then, we need to show that a fair fraction of the volume of the probability distribution $p(M)$ lies in the region of $M\geq M_{N,\Delta}$; that is, $p(M\geq M_{N,\Delta})=\Omega(1)$ for large $N$.
We prove such a probabilistic bound quantitatively in the following lemma.

\begin{restatable}[High-probability lower bound on the required number of queries for all classical algorithms to solve Simon's problem]{lemma}{QueryNumber}
\label{lemma:QueryNumber}
    Let $N>0$ be a problem size and the oracle $O_f$ be chosen according to Simon's problem as described in Problem~\ref{problem:SimonsProblem}. For any probabilistic and adaptive classical algorithm for solving Simon's problem with probability greater than or equal to $2/3$, the probability of the algorithm querying the oracle $O_f$ greater than or equal to $\sqrt{2\Delta/(1+\Delta)}\times2^{N/2}$ times is bounded by
    \begin{align}
    p\left(M\geq M_{N,\Delta}\right)\geq\frac{1-6\Delta}{3-6\Delta},
    \end{align}
    where $M_{N,\Delta}$ is given by~\eqref{eq:M_N_delta}, $\Delta\in(0,1/2)$ is the parameter in Proposition~\ref{prp:simon}, and $p(M\geq M_{N,\Delta})$ is defined as~\eqref{eq:p_M_geq_2N4}.
\end{restatable}
\begin{proof}
We here more rigorously express the ideas sketched above.
To this end, let us write the algorithm's success probability as
\begin{widetext}
    \begin{align}
        \label{eq:p_succ}
        &p\left(\mathrm{success}\right)=
        p\left(\text{success} \middle| M<M_{N,\Delta}\right) p\left(M<M_{N,\Delta}\right)
        + p\left(\text{success} \middle| M\geq M_{N,\Delta} \right) p\left(M\geq M_{N,\Delta} \right)\geq \frac{2}{3},
    \end{align}
where the $p(\text{success} | M<M_{N,\Delta})$ is the success probability of the algorithm conditioned on the cases where the algorithm has queried the oracle strictly less than $M_{N,\Delta}$ times, and $p(\text{success} | M\geq M_{N,\Delta} )$ is that conditioned on the cases where the oracle has queried more than or equal to $M_{N,\Delta}$ times.
To~\eqref{eq:p_succ}, we apply the result in Proposition~\ref{prp:simon} showing that the success probability conditioned on the case where the oracle has queried strictly less than $M_{N,\Delta}$ times is at most $1/2+\Delta$; i.e., it holds that
\begin{align}
\label{eq:p_succ_smallM_bound}
    p\left(\text{success} \middle| M<M_{N,\Delta}\right) \leq \frac{1}{2}+\Delta.
\end{align}
Moreover, due to the normalization of probability distributions, we have
\begin{align}
\label{eq:p_M_small}
    &p\left(M<M_{N,\Delta}\right)= 1-p\left(M\geq M_{N,\Delta}\right),\\
\label{eq:p_succ_M_large}
    &p\left(\text{success} \middle| M\geq M_{N,\Delta}\right)\leq 1.
\end{align}
Combining~\eqref{eq:p_succ},~\eqref{eq:p_succ_smallM_bound}~\eqref{eq:p_M_small}, and~\eqref{eq:p_succ_M_large}, we find
\begin{align}
    &\left(\frac{1}{2}+\Delta\right)\left(1-p\left(M\geq M_{N,\Delta}\right)\right)+ p\left(M\geq M_{N,\Delta}\right)\geq \frac{2}{3}.
\end{align}
\end{widetext}
Therefore, we have
\begin{equation}
    p\left(M\geq M_{N,\Delta}\right) \geq \frac{1-6\Delta}{3-6\Delta},
\end{equation}
and thereby conclude the proof.
\end{proof}

The next ingredient for proving the entropic lower bound on the energy consumption for classical algorithms to solve Simon's problem in Sec.~\ref{sec:classical_simon} is Proposition~\ref{prop:ClassicalEntropy} below.
For a given number of queries to the oracle, this proposition gives a lower bound on how much entropy the algorithm produces.

\begin{restatable}[Lower bound on entropy produced by all classical algorithms with a fixed number of queries for Simon's problem]{proposition}{ClassicalEntropy}\label{prop:ClassicalEntropy}
    For fixed $M\geq0$, any classical probabilistic and adaptive algorithm structured as in Fig.~\ref{fig:generic_algo} calling the oracle $O_f$ for Simon's problem in Problem~\ref{problem:SimonsProblem} on $M$ different inputs $x\in\{0,1\}^N$ generates the computer's state with entropy
    \begin{align}
        S(\rho_M) \geq \frac{1}{2}\ln\frac{2^N !}{(2^N-M)!},
    \end{align}
    where $\rho_M$ is given by~\eqref{eq:rho_M} in the main text, and $S$ is the entropy defined according to $S(\rho)=-\tr[\rho\ln\rho]$.
\end{restatable}

\begin{proof} 

To start the proof,
we clarify an explicit form of $\rho_M$ in~\eqref{eq:rho_M}.
Problem~\ref{problem:SimonsProblem} is defined using $s$, $b$, and $f$ sampled according to the uniform distributions $p(s)$, $p(b)$, and $p(f|s,b)$, respectively, as shown by~\eqref{eq:p} in the main text and detailed in the following.
The single-bit state $b\in\{0,1\}$ is given by a fair coin toss, i.e., chosen uniformly at random from $\{0,1\}$, giving us, for each $b$,
\begin{equation}
\label{eq:p_b}
    p(b)=\frac{1}{2}.
\end{equation}
According to the prescription for Simon's problem,
$s\in\{0,1\}^N\setminus\{0\}$ is chosen uniformly at random, giving us \begin{equation}
    p(s)=\frac{1}{2^N-1},
\end{equation}
because there exist $2^N-1$ elements in $\{0,1\}^N\setminus\{0\}$.
Furthermore, for the sake of our analysis, we here write
\begin{equation}
    p(f|b)\coloneqq\sum_{s\in\{0,1\}^N\setminus\{0\}}p(s)p(f|s,b).
\end{equation}
Then, we have
\begin{align}\label{eq:distribution_for_f}
    p(f|b) = \begin{cases}
\frac{1}{2^N!} & \text{if $b=0$}\\ 
\frac{1}{2^N-1}\frac{1}{\binom{2^N}{2^{N-1}} 2^{N-1}!} &\text{if $b=1$}.
\end{cases}
\end{align}
The probabilities on the right-hand side can be derived from normalization.
In the case of $b=0$, we choose 1-to-1 functions as $f$, and $p(f|b)$ comes from $|S_{2^N}|=2^N!$, where $S_{2^N}$ is the set of permutations, i.e., 1-to-1 functions, over $2^N$ elements in $\{0,1\}^N$.
In the case of $b=1$, we choose 2-to-1 functions as $f$, and we can see $p(f|b)$ as follows.
The factor $1/(2^N-1)$ comes from $p(s)$.
Moreover, to fully define the function $f$ satisfying $f(x)=f(x\oplus s)$, half of all inputs of $f$ can be assigned the unique corresponding outputs; that is, we have to choose $2^N / 2$ ordered elements from the set $\{0,1\}^N$ of $2^N$ elements. There are $\binom{2^N}{2^{N-1}}2^{N-1}!$ such choices.
The product of these numbers constitutes $p(f|b)$ in~\eqref{eq:distribution_for_f}.

The probability distribution $p(f|b)$ allows us to describe the computer's state at the end of the computation before the output in Fig.~\ref{fig:generic_algo}, up to the reversible transformation, as
\begin{align}\label{eq:prob_adaptive_state}
    \rho_M&=\frac{1}{p(M)}\sum_{b\in\{0,1\}}p(b)\sum_{f}p(f|b) p_f(M) \rho_{f,M},
\end{align}
as shown in~\eqref{eq:rho_M} of the main text.
We can furthermore split the state $\rho_M$ into a convex sum of the case where $f$ is 1-to-1, $\rho_M^{(b=0)}$, and the case where $f$ is 2-to-1, $\rho_M^{(b=1)}$,
\begin{align}
\label{eq:rhoM_rhoMb0_plus_rhoMb1}
    \rho_M = \frac{1}{2}\rho_M^{(b=0)} + \frac{1}{2}\rho_M^{(b=1)},
\end{align}
where the prefactor $1/2$ comes from $p(b)$ in~\eqref{eq:p_b}, and the two contributions for $b=0,1$ are defined as
\begin{align}
\label{eq:rho_b}
    \rho_M^{(b)} \coloneqq \frac{1}{p(M)}\sum_fp(f|b)p_f(M)\rho_{f,M}.
\end{align}
We use the concavity of the entropy together with the decomposition~\eqref{eq:rhoM_rhoMb0_plus_rhoMb1} to give a lower bound of $S(\rho_M)$ by
\begin{align}
    S(\rho_M)&\geq \frac{1}{2}S\left(\rho_M^{(b=0)}\right) + \frac{1}{2}S\left(\rho_M^{(b=1)}\right)\\
    \label{eq:S_rho_M_sum_b0_b1}
    &\geq\frac{1}{2}S\left(\rho_M^{(b=0)}\right),
\end{align}
where the last line follows from $S\left(\rho_M^{(b=1)}\right)\geq 0$.
The strategy for the remainder of the proof is to find a lower bound on $S\qty(\rho_M^{(b=0)})$ corresponding to the 1-to-1 functions.
Although this strategy may ignore the contribution from the 2-to-1 functions, it will turn out that the contribution from the 1-to-1 functions is already sufficient for our analysis.

To bound~\eqref{eq:S_rho_M_sum_b0_b1} further, we cast the quantum-state picture in~\eqref{eq:prob_adaptive_state} (i.e., the diagonal density operator) into the corresponding probability distribution to calculate its entropy. 
To describe the input-output relation of $y=f(x)$,
let us introduce a rank-$1$ computational-basis projector $\Pi_{\vec x, \vec y}$ acting on the space of the $2N\times M$ bits for $\rho_M$ in~\eqref{eq:prob_adaptive_state} as
\begin{align}
    \Pi_{\vec x, \vec y} = \bigotimes_{k=1}^M \ket{x_k,y_k}\bra{x_k,y_k},
\end{align}
where $\vec x$ and $\vec y$ are any ordered choices of $M$ elements in the set $\{0,1\}^N$.
In particular, we have $\vec x = (x_1,\dots,x_M)$ and $\vec y = (y_1,\dots,y_M)$, where $x_k,y_k \in \{0,1\}^N$ for all $1\leq k \leq M$.
Note that the algorithm uses $M$ different inputs as $\vec x$ by construction, but we do not assume the outputs $\vec y$ to be distinct from each other in general.
To calculate the entropy $S\left(\rho_M^{(b=0)}\right)$, we consider a probability distribution obtained from the diagonal operator $\rho_M^{(b=0)}$ by
\begin{align}\label{eq:def_pxy}
    p(\vec x, \vec y) \coloneqq \tr\left[\Pi_{\vec x, \vec y}\rho_M^{(b=0)}\right],
\end{align}
so that it holds that
\begin{align}
\label{eq:S_p}
    S\left(\rho_M^{(b=0)}\right) = -\sum_{\vec x, \vec y} p(\vec x, \vec y)\ln p(\vec x, \vec y).
\end{align}
Substituting $\rho_{f,M}$ in~\eqref{eq:rho_b} with~\eqref{eq:rho_f_M} in the main text, we can express the probability $p(\vec x,\vec y)$ of the $2N\times M$ bits in~\eqref{eq:def_pxy} as
\begin{widetext}
    \begin{align}
        p(\vec x,\vec y) &= \sum_{f}
        \bigg(
            \frac{1}{p(M)} q\left(\mathrm{stop}|\{(x_\ell,f(x_\ell))\}_{1\leq\ell\leq M}\right) q(x_M|\{(x_\ell,f(x_\ell))\}_{1\leq\ell\leq M-1})     \label{eq:p_xyM} \\
            &\qquad\prod_{k=1}^{M-1}q(\mathrm{continue}|\{(x_\ell,f(x_\ell))\}_{1\leq\ell\leq k})q(x_{k}|\{(x_\ell,f(x_\ell))\}_{1\leq\ell\leq k-1}) \times p(f | b=0) \delta_{f(x_1),y_1}\cdots \delta_{f(x_M),y_M} \bigg)\nonumber \\
        &= \overbrace{\bigg(\frac{1}{p(M)} q\left(\mathrm{stop}|\{(x_\ell,y_\ell)\}_{1\leq\ell\leq M}\right)q(x_M|\{(x_\ell,f(x_\ell))\}_{1\leq\ell\leq M-1})\hspace{4cm}}^{\eqqcolon q(\vec x | \vec y)} \nonumber \\
        &\hspace{5cm}\times \prod_{k=1}^{M-1}q(\mathrm{continue}|\{(x_\ell,y_\ell\}_{1\leq\ell\leq k})q(x_{k}|\{(x_\ell,y_\ell)\}_{1\leq\ell\leq k-1})\bigg) \nonumber \\
        &\qquad\times \underbrace{\sum_{f}p(f | b= 0) \delta_{f(x_1),y_1}\cdots \delta_{f(x_M),y_M}}_{\eqqcolon p(\vec y)} \\
        &={q(\vec x |\vec y)} \times p(\vec y),
    \end{align}
\end{widetext}
where $q(\vec x|\vec y)$ and $p(\vec y)$ are defined as shown in the second line, and $\delta_{a,b}$ is a delta function, i.e.,
\begin{align}
    \delta_{a,b}=\begin{cases}
        1,&\text{if $a=b$},\\
        0,&\text{otherwise}.
    \end{cases}
\end{align}
The notations of $q(\vec x|\vec y)$ and $p(\vec y)$ are chosen suggestively to distinguish between probabilities arising from the algorithm's probabilistic strategy (written in $q$) and the setting of Simon's problem (written in $p$), as explained in more detail below.
Note that in general, $p(\vec y)$ may depend on $\vec y$, which can be any ordered sequence of $M$ elements from $\{0,1\}^N$ and may include the same elements $y_k=y_\ell$ for $1\leq k<\ell\leq M$; however, $p(\vec y)$ does not depend on the choice of $M$ different inputs $x_k$ due to uniform randomness of the choice of 1-to-1 functions $f:\{0,1\}^N\to\{0,1\}^N$ in Problem~\ref{problem:SimonsProblem} in the case of $b=0$.
By construction, the probability distribution $p(\vec x,\vec y)$ is normalized with respect to a sum over all $\vec x$ and $\vec y$, i.e.,
\begin{align}
    \label{eq:p_normalization}
    \sum_{\vec x, \vec y}p(\vec x,\vec y)=1.
\end{align}

The quantity $p(\vec y)$ can be interpreted as the probability of measuring the output sequence $\vec y$, given the computer's state $\rho^C$.
Summing over all values of $\vec y$ lets us recover the normalization, i.e.,
\begin{align}
    \sum_{\vec y}p(\vec y) &= \sum_f p(f|b=0) \sum_{\vec y} \delta_{f(x_1),y_1}\cdots \delta_{f(x_M),y_M} \\
    &= \sum_f p(f|b=0) \times 1 = 1.\label{eq:py_normalization}
\end{align}

On the other hand, $q(\vec x|\vec y)$ can be interpreted as the probability of querying the input sequence $\vec x$ conditioned on having exactly $M$ queries and on measuring the output sequence $\vec y$.
For any choice of $M\geq 1$ and $\vec y$, the function $q(\vec x|\vec y)$ is a probability distribution over $\vec x$ by construction (as shown in~\eqref{eq:rho_f_M} and~\eqref{eq:p_M} in the main text) and normalized by
\begin{equation}
\label{eq:q_normalization}
    \sum_{\vec x}q(\vec x|\vec y) = 1.
\end{equation}

These interpretations also show that we can understand $p(\vec y)$ as the distribution obtained by marginalizing $p(\vec x,\vec y)$ over $\vec x$,
\begin{align}
    \sum_{\vec x} p(\vec x,\vec y) = \left(\sum_{\vec x} q(\vec x |\vec y)\right)p(\vec y)=p(\vec y).
\end{align}
Using the fact that entropy does not increase under marginalization~\cite{Nielsen2010}, we can bound the entropy of $p(\vec x,\vec y)$ in~\eqref{eq:S_p} by that of the marginal distribution $p(\vec y)$, i.e.,
\begin{align}
    -\sum_{\vec x,\vec y} p(\vec x,\vec y)\ln p(\vec x,\vec y)\geq -\sum_{\vec y} p(\vec y) \ln p(\vec y).\label{eq:S_rho_M_sum_py}
\end{align}
This inequality allows us to analyze a lower bound of the entropy of $\rho_M$ solely through the contribution coming from properties of Simon's problem (captured by $p(\vec y)$), irrespectively of the contribution from the algorithm's strategy to solve it (captured by $q(\vec x|\vec y)$).

In the following, we estimate the entropy $-\sum_{\vec y}p(\vec y)\ln p(\vec y)$ in~\eqref{eq:S_rho_M_sum_py}.
To this end, we further split the sum over all possible output sequences $\vec y$ in~\eqref{eq:S_rho_M_sum_py} into one for $M$ distinct outputs in $\{0,1\}^N$ and the rest.
In particular, we write the set of $M$ distinct elements in $\{0,1\}^N$ as
\begin{align}
    Y_{\text{distinct}} &= \{\vec y = (y_1,\dots,y_M) \in{(\{0,1\}^N)}^M :\nonumber\\
    &\qquad y_k\neq y_\ell\text{ for all $1\leq k<\ell\leq M$}\}.
\end{align}
and its complement as
\begin{align}
    Y_{\text{indistinct}}=\{\vec y = (y_1,\dots,y_M) : y_i \in \{0,1\}^N\} \setminus Y_\text{distinct}.
\end{align}
For any $\vec y\in Y_\text{indistinct}$ that includes indistinct pairs of $y_k=y_\ell$, we have
\begin{equation}
\label{eq:p_y_indistinct}
    p(\vec y)=0.
\end{equation}
To see this, recall that $p(\vec y)$ in~\eqref{eq:p_xyM} is defined as a sum over all 1-to-1 functions, but the input $\vec x$ to each 1-to-1 function is a sequence of distinct values, i.e., $x_k\neq x_\ell \in \{0,1\}^N$ for all $1\leq k<\ell\leq M$.
Thus, the sequence of outputs $f(x_1),\dots,f(x_M)$ of any 1-to-1 function is also a sequence of pairwise distinct values in $\{0,1\}^N$.
As a consequence, we have~\eqref{eq:p_y_indistinct} if $\vec y\notin Y_\text{distinct}$.
On the other hand, for any choice of $M$ distinct values $\vec y\in Y_\text{distinct}$, the probability $p(\vec y)$ does not depend on $\vec y$; i.e., it holds for some constant $p_\mathrm{distinct}>0$ that
\begin{equation}
    p(\vec y)=p_\mathrm{distinct}.
\end{equation}
This holds due to the uniform randomness of the choice of 1-to-1 functions $f:\{0,1\}^N\to\{0,1\}^N$ in Problem~\ref{problem:SimonsProblem}, i.e.,
\begin{align}
    p(\vec y) = p(\pi(\vec y)),
\end{align}
for any permuted vector $\pi(\vec y) = (\pi(y_1),\dots,\pi(y_M))$.
Using the normalization~\eqref{eq:py_normalization}, we further have
\begin{align}
    1 &= \sum_{\vec y\in Y_\text{distinct}} p(\vec y) + \sum_{\vec y\in Y_\text{indistinct}} p(\vec y) \\
    &=\sum_{\vec y\in Y_\text{distinct}} p(\vec y) + 0 \\
    &= \left|Y_\text{distinct}\right|  p_\text{distinct}.
\end{align}
A combinatorial argument yields
\begin{equation}
    |Y_\text{distinct}| = \binom{2^N}{M}M!=\frac{2^N!}{(2^N-M)!},
\end{equation}
which is the number of ordered choices of $M$ distinct elements in the set $\{0,1\}^N$ of $2^N$ elements.
Thus, we have
\begin{align}\label{eq:p_distinct}
p_\text{distinct}=\frac{(2^N-M)!}{2^N!}.
\end{align}
Therefore, it holds that
\begin{align}
     &-\sum_{\vec y} p(\vec y)\ln p(\vec y)\\
     &= -\sum_{\vec y \in Y_\text{distinct}} p(\vec y)\ln p(\vec y)-\sum_{\vec y \in Y_\text{indistinct}} p(\vec y)\ln p(\vec y)\\
     &= -\sum_{\vec y \in Y_\text{distinct}} p(\vec y)\ln p(\vec y)+0\\
     &=\ln\frac{2^N!}{(2^N-M)!}.
     \label{eq:S_py}
\end{align}

Consequently, it follows from~\eqref{eq:S_rho_M_sum_b0_b1}, \eqref{eq:S_p}, \eqref{eq:S_rho_M_sum_py}, and~\eqref{eq:S_py} that
\begin{equation}
    S(\rho_M)\geq\frac{1}{2}\ln\frac{2^N!}{(2^N-M)!},
\end{equation}
which leads to the conclusion.
\end{proof}

Finally, we show Lemma~\ref{prp:fraction_evaluation} below to evaluate the fraction of factorials appearing in the lower bound of the entropy in the main text.

\begin{lemma}[\label{prp:fraction_evaluation}Estimation of fraction of factorials]
For any $N\in\{1,2,\dots\}$ and $\Delta\in(0,1/6)$, it holds that
\begin{equation}
    \ln\frac{2^N!}{\left(2^N-M_{N,\Delta}\right)!}\geq\sqrt{\frac{2\Delta}{1+\Delta}} 2^{N/2}N \ln 2-3,
\end{equation}
where $M_{N,\Delta}$ is defined as~\eqref{eq:M_N_delta}.
\end{lemma}

\begin{proof}
We bound the factorials using a variant of Stirling's approximation~\cite{10.2307/2308012}
\begin{align}
\label{eq:stirling_inequality}
    \sqrt{2\pi n}\left(\frac{n}{e}\right)^n e^{\frac{1}{12n+1}}\leq n! \leq \sqrt{2\pi n}\left(\frac{n}{e}\right)^n e^{\frac{1}{12n}},
\end{align}
which holds for arbitrary positive integers $n$. 
In particular, with $M_{N,\Delta}$ denoting $\left\lceil\sqrt{\frac{2\Delta}{1+\Delta}}\times2^{N/2}\right\rceil$ as in~\eqref{eq:M_N_delta}, we have
\begin{widetext}
\begin{align}
    &\ln\frac{2^N!}{\left(2^N-M_{N,\Delta}\right)!}\nonumber\\
    &\geq \ln \frac{\sqrt{2\pi 2^N}\left(\frac{2^N}{e}\right)^{2^N} e^{\frac{1}{12\times2^N+1}}}{\sqrt{2\pi (2^N-M_{N,\Delta})}\left(\frac{2^N-M_{N,\Delta}}{e}\right)^{2^N-M_{N,\Delta}} e^{\frac{1}{12(2^N-M_{N,\Delta})}}} \label{eq:tighter_stirling_00} \\
    &= \ln \left( \left(\frac{2^N}{2^N-M_{N,\Delta}}\right)^{2^N+\frac{1}{2}} (2^N-M_{N,\Delta})^{M_{N,\Delta}} e^{-M_{N,\Delta} + \frac{1}{12\times 2^N+1}-\frac{1}{12(2^N-M_{N,\Delta})}} \right) \\
    &=-\left(2^N + \frac{1}{2}\right)\ln\left(1-\frac{M_{N,\Delta}}{2^N}\right) + M_{N,\Delta}\left(N \ln 2-1\right) -\left(-M_{N,\Delta}\ln\left(1-\frac{M_{N,\Delta}}{2^N}\right)
    + \frac{1}{12}\left(\frac{1}{2^N-M_{N,\Delta}}-\frac{1}{2^N+\frac{1}{12}}\right)\right)
    \\
    &\geq \left(2^N+\frac{1}{2}\right)\frac{M_{N,\Delta}}{2^N} + M_{N,\Delta} \left( N\ln 2  - 1\right)- \left(\frac{\frac{M_{N,\Delta}^2}{2^N}}{1-\frac{M_{N,\Delta}}{2^N}} 
    + \frac{1}{12}\left(\frac{1}{2^N-M_{N,\Delta}}-\frac{1}{2^N+\frac{1}{12}}\right)\right)
    \label{eq:tighter_stirling_01_main}\\
    &\geq M_{N,\Delta} + M_{N,\Delta} \left( N\ln 2 - 1\right)
    - 3
    \label{eq:tighter_stirling_02_main}\\
    &= M_{N,\Delta} N\ln 2 - 3\\
    &\geq \sqrt{\frac{2\Delta}{1+\Delta}} 2^{N/2}N \ln 2 - 3,
\end{align}
\end{widetext}
where~\eqref{eq:tighter_stirling_01_main} follows from
\begin{equation}
    \frac{x}{1+x}\leq\ln(1+x)\leq x
\end{equation}
for all $x>-1$ with $x=-M_{N,\Delta}/2^N$,
and~\eqref{eq:tighter_stirling_02_main} holds due to
\begin{align}
&2^N+\frac{1}{2}\geq 2^N,\\
\label{eq:constant_order}
&\frac{\frac{M_{N,\Delta}^2}{2^N}}{1-\frac{M_{N,\Delta}}{2^N}} 
+ \frac{1}{12}\left(\frac{1}{2^N-M_{N,\Delta}}-\frac{1}{2^N+\frac{1}{12}}\right)\leq 3
\end{align}
for all $N\in\{1,2,\ldots\}$ and $\Delta\in(0,1/6)$.
To see that~\eqref{eq:constant_order} holds true, we define a function representing the left-hand side of~\eqref{eq:constant_order}, i.e.,
\begin{align}
\label{eq:F}
F(N,\Delta)\coloneqq\frac{\frac{M_{N,\Delta}^2}{2^N}}{1-\frac{M_{N,\Delta}}{2^N}} 
+ \frac{1}{12}\left(\frac{1}{2^N-M_{N,\Delta}}-\frac{1}{2^N+\frac{1}{12}}\right),
\end{align}
and observe that this function has a constant order $F(N,\Delta)=O(1)$ as $N\to\infty$, in particular,
\begin{align}
    \lim_{N\to\infty}F(N,\Delta)=\lim_{N\to\infty}\frac{M_{N,\Delta}^2}{2^N}=\frac{2\Delta}{1+\Delta};
\end{align}
then, for every $N$, due to the monotonicity of $F(N,\Delta)$ in terms of $\Delta$, we have
\begin{align}
    F(N,\Delta)\leq F(N,1/6)\in\qty[\frac{2}{7},3],
\end{align}
where the range can be confirmed also from Fig.~\ref{fig:F}.
\end{proof}

\begin{figure}
    \centering
    \includegraphics[width=3.4in]{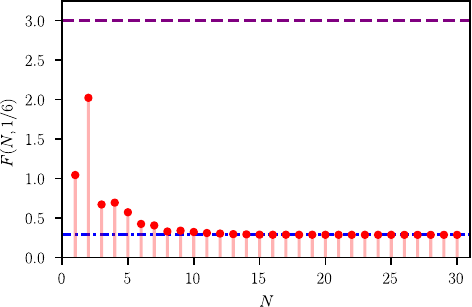}
    \caption{
    A plot of $F(N,1/6)$ in~\eqref{eq:F} for $N\in\{1,2,\ldots,30\}$ (red dots). As discussed in the main text, we have $\lim_{N\to\infty}F(N,1/6)=2\Delta/(1+\Delta)|_{\Delta=1/6}=2/7$ (blue dashed-dotted line). As the plot shows, the function $F(N,1/6)$ is upper bounded by $3$  (purple dashed line) since it takes the maximum at $N=2$ with $F(2,1/6)=2.02\cdots\leq 3$.
    }
    \label{fig:F}
\end{figure}

\subsection{\label{appendix:proofs_classical_achievability}Proof on classical upper bound for solving Simon's problem}

In this appendix, we present a proof of the achievable upper bound of query complexity of Simon's problem in Problem~\ref{problem:SimonsProblem} by the classical algorithm with $M$ uniformly random queries, as described in Proposition~\ref{prp:simon_achievability} of Sec.~\ref{sec:experiment}.
The statement of Proposition~\ref{prp:simon_achievability} is repeated again below for readability.

\SimonUpperBound*

\begin{proof}
    Recall that, as described in Sec.~\ref{sec:experiment}, with the $M$ different inputs $\vec x=(x_1,\ldots,x_M)$ chosen uniformly at random, the algorithm works as follows:
    \begin{itemize}
        \item if at least one pair $(x_k,x_m)$ satisfying $f(x_k)=f(x_m)$ for some $x_k\neq x_m$ is found in the $M$ queries, the algorithm has an evidence that $f$ is a $2$-to-$1$ function and thus outputs $a=1$; 
        \item  otherwise, the algorithm considers $f$ to be $1$-to-$1$ and outputs $a=0$.
    \end{itemize}
    If $f$ is a $1$-to-$1$ function, the algorithm always returns the correct output $a=0$; thus, it suffices to analyze the probability of outputting $a=0$ in the case of $f$ being $2$-to-$1$.
    For any $s\in\{0,1\}^N\setminus\{0\}$, 
    the set $\{0,1\}^N$ is divided into $2^{N}/2$ equivalence classes with an equivalence relation $x\sim x\oplus s$.
    If the uniformly random $M$ queries of the algorithm include a pair of $x$ and $x\oplus s$ in the same equivalence class for some $x$,
    then the algorithm finds $f(x)=f(x\oplus x)$ and thus outputs $a=1$ correctly.
    The false output, i.e., $a=0$ for the $2$-to-$1$ function $f$, arises when all the $M$ queries $x_1,\ldots,x_M$ are in different equivalence classes.

    The failure probability is given by a fraction of the number of the sequences that give a false output, divided by the number of possible sequences.
    The number of sequences that return the false output is given by $\binom{2^N/2}{M}2^M$, where $M$ unordered inputs are chosen among the $2^N/2$ equivalence classes, which amounts to the factor $\binom{2^N/2}{M}$, but for each of the $M$ equivalence classes, there are $2$ possible choices $x$ or $x\oplus s$, giving the other factor $2^M$.
    On the other hand, the number of possible sequences is given by $\binom{2^N}{M}$, the number of unordered choices of $M$ elements in $2^N$.
    All-in-all, the failure probability is given by
    \begin{equation}
        \frac{\binom{2^N/2}{M}2^M}{\binom{2^N}{M}},
    \end{equation}
and we have a bound on this failure probability by
\begin{widetext}
    \begin{align}
        \frac{\binom{2^N/2}{M}2^M}{\binom{2^N}{M}}&=\frac{2^{N-1}(2^{N-1}-1)(2^{N-1}-2)\cdots(2^{N-1}-(M-1))\times 2^M}{2^{N}(2^{N}-1)(2^{N}-2)\cdots(2^{N}-(M-1))}\\
        &=\frac{2^{N}(2^{N}-2)(2^{N}-4)\cdots(2^{N}-2(M-1))}{2^{N}(2^{N}-1)(2^{N}-2)\cdots(2^{N}-(M-1))}\\
        &=1\left(1-\frac{1}{2^N-1}\right)\left(1-\frac{2}{2^N-2}\right)\cdots\left(1-\frac{M-1}{2^N-(M-1)}\right)\\
        &\leq 1\left(1-\frac{1}{2^N}\right)\left(1-\frac{2}{2^N}\right)\cdots\left(1-\frac{M-1}{2^N}\right)\\
        &\leq e^{-0}e^{-1/2^N}e^{-2/2^N}\cdots e^{-(M-1)/2^N}\\
        &= e^{-M(M-1)/{2^{N+1}}},
    \end{align}
\end{widetext}
where the last inequality follows from $e^{x}\geq 1+x$ for all $x$.

Therefore, for any $M$ satisfying
\begin{equation}
    M>\frac{1}{2}\sqrt{\left(8\ln \frac{1}{\delta}\right)2^{N}+1}+\frac{1}{2},
\end{equation}
we have
\begin{equation}
    e^{-M(M-1)/{2^{N+1}}}<\delta.
\end{equation}
Since $M$ is an integer, with $M$ chosen as~\eqref{eq:M_choice}, the classical algorithm with the $M$ or more uniformly random queries achieves the success probability in solving Simon's problem greater than $1-\delta$.
\end{proof}

\section{\label{appendix:model_details}Details of control-cost model}

In this appendix, we aim to provide a clearer explanation of the concept of control cost discussed in Sec.~\ref{sec:energetic_costs} of the main text.
In the main text, we have argued that when performing a unitary operation on a computer, it can require more energy than just the change of energy in the memory alone.
We call this cost the control cost.
We have introduced a generic constant to account for the control cost so that our analysis should not depend on a specific architecture of the computer, whereas the exact value of control costs may depend on specific models being used.
Crucially, we assume that the control cost is non-negative yet at most in the order of the energy scale of the Hamiltonian, as shown in~\eqref{eq:control_cost_positive} and~\eqref{eq:E_ctrl}.
Here, we present a toy example that exhibits these properties.
To this end, we sketch a model that can implement arbitrary unitaries to a single-qubit target system assisted by an auxiliary system but may incur the control cost, based on the model shown, e.g., in Ref.~\cite{Aaberg2014}.

Interactions between systems are described by interaction Hamiltonians, and the time evolution of the quantum states, according to the Hamiltonians, conserves energy.
If we work with non-energy-degenerate qubits as our computational platforms, such as trapped ions, superconducting qubits, and optics, then the energy conservation restricts the set of possible operations implementable by the Hamiltonian dynamics.
In particular, if the quantum computer's memory is described by the Hamiltonian $H_S$, then unitary operations $U$ implementable within the energy-conservation law would satisfy a condition of
\begin{align}
    \label{eq:commutator_HS_U}
    [H_S,U] = 0.
\end{align}
In more general cases, a unitary operation to be applied to a target system may not satisfy this energy-conservation condition; however, the implementation can be made to conserve energy by enlarging the system it acts on to include an auxiliary system.
Such an auxiliary system for the control may also be called a catalyst or battery.
For a selection of references where this setting is discussed in further detail, see Refs.~\cite{Skrzypczyk2013,Skrzypczyk2014,Navascues2014,Aaberg2014}.
The local operation on the target system is then only the effective operation that arises by tracing out the auxiliary system.

\paragraph*{Toy example based on Ref.~\cite{Aaberg2014}.} Let us model the target system $\mathcal H_S=\mathbb C^2$ as a qubit and the auxiliary system $\mathcal H_C$ for the control as a harmonic oscillator, i.e., a half-infinite ladder. 
The two systems are governed by the Hamiltonian
\begin{align}
\label{eq:H_model}
    H = H_S + H_C = \omega \ket{1}\bra{1}_S \otimes \mathds 1_C + \mathds 1_S\otimes \sum_{n\geq 0}n\omega \ket{n}\bra{n}_C,
\end{align}
where $\{\ket{0}_S,\ket{1}_S\}$ and $\{\ket{n}_C:n=0,1,\ldots\}$ are orthonormal bases of $\mathcal H_S$ and $\mathcal{H}_C$, and $\mathds{1}_C$ and $\mathds{1}_S$ are the identity operators acting on $\mathcal H_S$ and $\mathcal{H}_C$, respectively.
If we desire to perform a unitary $U$ on our target qubit, we can do so in an energy-preserving way by defining the dilation on the composite system $\mathcal H_S \otimes \mathcal H_C$ as
\begin{align}\label{eq:VU_def}
    V[U] &\coloneqq \ket{0}\bra{0}_S\otimes\ket{0}\bra{0}_C+\nonumber\\
    &\quad\sum_{n\geq 1}\sum_{i,j=0}^1U_{i,j}\ket{i}\bra{j}_S\otimes\ket{n-i}\bra{n-j}_C,
\end{align}
where we write $U_{i,j}\coloneqq\bra{i}U\ket{j}$.
The operator $V[U]$ commutes with $H=H_S+H_C$ in~\eqref{eq:H_model} because it commutes termwise.
In particular, the term in the first line of~\eqref{eq:VU_def} commutes with $H$ because it is diagonal in terms of the same basis.
The same is true for all terms in the second line with $i=j$.
All other terms are of the form
\begin{equation}
    \text{(constant)}\times\ket{0}\bra{1}_S\otimes \ket{n}\bra{n-1}_C + \text{(h.c.)}
\end{equation}
where \text{(h.c.)}\ represents the hermitian conjugate of the first term, and $n>0$ is some integer.
These terms also commute since they preserve the degenerate subspaces $\text{span}_\mathbb{C}\{\ket{0}_S\otimes \ket{n}_C,\ket{1}_S\otimes\ket{n-1}_C\}$ of $H$ (recall, two diagonalizable matrices commute if and only if they are simultaneously diagonalizable, i.e., if and only if they preserve each other's eigenspace).

The question we now have is: for a given input state $\ket{\psi}_S=\sum_{j}\psi_j\ket{j}$ of the target system and an initially pure state $\ket{\phi}_C$ of the auxiliary system given by
\begin{align}
\label{eq:phi_C}
    \ket{\phi}_C=\sum_{n}\phi_n\ket{n}=\frac{1}{\sqrt{L}}\sum_{n=0}^{L-1}\ket{n+\ell_0}_C
\end{align}
with some constant $\ell_0\geq 1$, how close are $U\ket{\psi}_S$ and $\rho_S'=\tr_C\{V[U] \ket{\psi}\bra{\psi}_S\otimes \ket{\phi}\bra{\phi}_C V[U]^\dagger\}$ in general? 
We write the approximate implementation (in~\eqref{eq:VU_def}) of $U$ acting on the subsystem $S$ as a channel
\begin{align}
    \Lambda_U({\psi}_S) = \tr_C\big[V[U] \big(\ket{\psi}\bra{\psi}_S \otimes \ket{\phi}\bra{\phi}_C \big) V[U]^\dagger\big].
\end{align}
To estimate the quality of $\Lambda_U$ in approximating $U$, we calculate how $V[U]$ acts on the two spaces $\mathcal H_S\otimes \mathcal H_C$.
For this purpose, we define another state of the auxiliary system
\begin{align}
    \ket{\phi'}_C=\frac{1}{ \sqrt{L-2}}\sum_{n=1}^{L-2} \ket{n+\ell_0}_C,
\end{align}
which overlaps with $\ket{\phi}_C$ up to the two boundary terms $\ket{\ell_0}_C$ and $\ket{\ell_0+L-1}_C.$
With the help of this state, we can concisely write how $V[U]$ acts on $\ket{\psi}_S\otimes\ket{\phi}_C,$
i.e.,
\begin{align}
    V[U] (\ket{\psi}_S\otimes\ket{\phi}_C) &=  \sqrt{\frac{L-2}{L}}\qty(U\ket{\psi}_S)\otimes\ket{\phi'}_C \nonumber \\
    &+\sqrt{\frac{2}{L}}\ket{\chi}_{SC},\label{eq:VU_psi_phi}
\end{align}
where $\ket{\chi}_{SC}$ is some state orthogonal to $\qty(U\ket{\psi}_S)\otimes\ket{\phi'}_C.$ This state moreover satisfies
\begin{align}
\label{eq:chi_orthogonality}
    \bra{\chi}_{SC}\qty(\mathds{1}_S\otimes\ket{\phi'}_C)= 0.
\end{align}
Calculating the reduced density matrix $\rho_S'=\Lambda_U(\psi_S)$ yields, for any initial state $\ket{\psi}$,
\begin{align}
    \rho_S' &= \qty(1-\frac{2}{L})U\ket{\psi}\bra{\psi}_S U^\dagger + \frac{2}{L}\chi_S,\label{eq:rhoCprime_approx}
\end{align}
where $\chi_S = \tr_C\qty[\ket{\chi}\bra{\chi}_{SC}]$ is the remainder from the partial trace, and no cross terms from $\qty(U\ket{\psi}_S)\otimes\ket{\phi'}_C \bra{\chi}_{SC}+(\mathrm{h.c.})$ appear because of orthogonality shown in~\eqref{eq:chi_orthogonality}.

Now, we can prove that in the limit of $L\rightarrow\infty$, the channel $\Lambda_U$ converges to the original unitary channel for $U$.
In terms of channel fidelity $f(\Lambda_U,U)$ (see, e.g., Refs.~\cite{Horodecki1999,Nielsen2002,Magesan2011}), we can quantify how well $\Lambda_U$ approximates the unitary channel $U$ on average over the input states.
In our case, we have
\begin{align}
    f(\Lambda_U,U) &=\int d\mu(\psi) \braket{\psi}{U^\dagger \Lambda_U (\psi) U|\psi}_S \\
    &\hspace{-0.21cm}\stackrel{\text{\eqref{eq:rhoCprime_approx}}}{=}\hspace{-0.21cm}\int d\mu(\psi)\left(\braket{\psi}{U^\dagger U \ket{\psi}\bra{\psi} U^\dagger U|\psi} - O\qty(\frac{1}{L})\right) \nonumber \\
    &= \int d\mu(\psi) \left(1 - O\qty(\frac{1}{L})\right) \\
    &=1-O\left( \frac{1}{L} \right),\label{eq:f_LambdaU_U}
\end{align}
where the integral is over the normalized and uniform measure $d\mu(\psi)$ of the projective space $P(\mathcal H_S) \cong \mathbb CP^1$ of the single-qubit Hilbert space $\mathcal H_S\cong \mathbb C^2$.
The bound~\eqref{eq:f_LambdaU_U} follows from the approximation using~\eqref{eq:rhoCprime_approx} and the negative sign in front of the big-$O$ term comes from $|\braket{\psi}{\rho_S'|\psi}_S|< 1$ for all states $\rho_S' \neq \ket{\psi}\bra{\psi}_S$.
This shows that in the limit of large $L$, the two unitary channels corresponding to $U$ and $\Lambda_U$ coincide.

\paragraph*{Emergence of the control cost.}
In Ref.~\cite{Aaberg2014}, it is argued that the channel on the auxiliary system for the control, defined by $\rho_C \mapsto \rho_C' = \tr_S\{V(U)\rho_S\otimes\rho_C V(U)^\dagger\}$, always increases the entropy of the auxiliary system.
As a consequence, we can argue that the cost of resetting the state of the auxiliary system for the control will always be positive within this framework since the Landauer-erasure cost of reinitializing the auxiliary system is bounded by the positive entropy difference $S(\rho_C')-S(\rho_C)\geq 0$~\cite{Reeb2014}.
This fact also justifies the assumption that $E_\text{ctrl.}\geq 0$ in~\eqref{eq:control_cost_positive}.

To show this fact quantitatively, we here bound the increase of the entropy of the auxiliary system on average over the initial states $\ket{\psi}_S$ of the target system in performing $V[U]$.
To this end, we compute the average entropy change for the auxiliary system as
\begin{align}
\label{eq:DeltaS_def}
    \langle\Delta S_C\rangle &\coloneqq \int d\mu(\psi) \left|S(\rho_C')-S(\ket{\phi}\bra{\phi}_C)\right|\\
    &=\int d\mu(\psi)S(\rho_C'),
\end{align}
where the integral is over all possible initial states on the target system.
A lower bound of $S(\rho_C')$ can be established by using concavity of the entropy as
\begin{align}
    S(\rho_C') &= S(\rho_S')\\
    &= S\qty(\qty(1-\frac{2}{L})U\ket{\psi}\bra{\psi}_S U^\dagger + \frac{2}{L}\chi_S) \\
    &\geq \frac{2}{L} S(\chi_S),
\end{align}
where the first equality follows from the fact that the state $V[U](\ket{\psi}_S\otimes \ket{\phi}_C)$ in~\eqref{eq:VU_psi_phi} is pure, the second equality from~\eqref{eq:rhoCprime_approx}, and we write $\chi_S = \tr_C[\ket{\chi}\bra{\chi}_{SC}]$.
The entropy $S(\chi_S)\geq 0$ may depend on $U$ and $\ket{\psi}_S$ but is independent of the value of $L\geq 2$, which can be seen by explicitly calculating~\eqref{eq:VU_psi_phi} and observing that no factor of $L$ appears in $\ket{\chi}_{SC}$.
For some particular choices of unitary $U$, it may hold that $S(\chi_S)=0$ or even $S(\rho_C')=0$ (for example, for $U=\mathds 1$).
However, for general unitaries, the entropy does not vanish.
Under the assumption that the entropy $S(\chi_S)>0$ is strictly positive for at least a constant fraction of the initial state $\ket{\psi}_S$ among all possible choices, we asymptotically have
\begin{align}
\label{eq:SC_lower}
    \langle\Delta S_C\rangle = \Omega\qty(\frac{1}{L}),
\end{align}
which provides a lower bound of the control cost arising from reinitializing the auxiliary system by the Laudauer-erasure protocol per implementing $U$ with $V[U]$.

Apart from this control cost of reinitialization, in practice, one may also need to account for the energetic cost required for preparing the state~\eqref{eq:phi_C} of the auxiliary system.
We also note that when connecting these results to actual energy requirements in an experimental implementation of a quantum computer, other control costs like the ones to control the timing of the operations in implementing $V[U]$ should also be considered~\cite{Xuereb2023}.
As for the energetic cost,
by definition of $\ket{\phi}_C$ and $H_C$,
we have
\begin{align}
\label{eq:HC_O_L}
    \bra{\phi}H_C \ket{\phi} =\Omega(\omega L).
\end{align}
This scaling shows that preparing the state~\eqref{eq:phi_C} of the auxiliary system comes at an energetic cost that grows with the target channel fidelity for the implementation $\Lambda_U$ to approximate the gate $U$.
Letting $\varepsilon=1-f(\Lambda_U,U)$ denote the infidelity,
we can combine~\eqref{eq:f_LambdaU_U} and~\eqref{eq:HC_O_L} to qualitatively recover the scaling derived in Ref.~\cite{Chiribella2021}, i.e.,
\begin{align}
    \label{eq:HC_O_eps}
    \bra{\phi} H_C \ket{\phi} = O\left(\frac{\|H_S\|_\infty}{\varepsilon}\right).
\end{align}
A detailed analysis and the conception of a more optimal auxiliary system for the control are also given in Ref.~\cite{Chiribella2021}, where the linear dependency of the control cost $E_\text{ctrl.}=O(E_\mathrm{qubit})$ is also derived for a certain model, as in~\eqref{eq:E_ctrl} of the main text.
Note that a similar bound can also be obtained from the results in Refs.~\cite{PhysRevLett.121.110403,PhysRevResearch.2.043374,tajima2022universal_limitation,tajima2022universal_tradeoff,tajima}.
These bounds show that the assumptions on the control cost in our analysis indeed hold for the conventional theoretical models that we believe capture the essence of physics in practice.

\bibliography{bibliography}

\end{document}